\documentclass[aos, preprint]{imsart}

\RequirePackage{amsthm,amsmath,amsfonts,amssymb}
\RequirePackage[numbers]{natbib}
\RequirePackage[colorlinks,citecolor=blue,urlcolor=blue]{hyperref}
\RequirePackage{graphicx}
\usepackage{subcaption}
\usepackage{bm}
\usepackage{mathtools, stmaryrd}
\usepackage{algorithm}

\startlocaldefs
\theoremstyle{plain}

\newtheorem{theorem}{Theorem}[section]
\newtheorem{lemma}[theorem]{Lemma}
\theoremstyle{remark}
\newtheorem{definition}[theorem]{Definition}

\newtheorem{example}{Example}[section]

\newcommand{\field}[1]{\mathbb{#1}}
\newcommand{\dd}{\mathrm{d}}
\newcommand{\hd}{\mathrm{h}}
\newcommand{\fm}{\mathfrak{m}}
\newcommand{\zb}{\bm{z}}
\newcommand{\cf}{\mathsf{c}}

\newcommand{\op}{\mathrm{op}}
\newcommand{\sfy}{\mathsf{y}}
\newcommand{\xb}{\mathbf{x}}
\newcommand{\wb}{\bm{w}}
\newcommand{\fb}{\mathbf{f}}
\newcommand{\el}{\mathsf{l}}
\newcommand{\gb}{\bm{g}}

\newcommand{\Ub}{\mathbf{U}}
\newcommand{\Wb}{\mathbf{W}}

\newcommand{\phib}{\bm{\phi}}
\newcommand{\OO}{\mathrm{O}}
\newcommand{\oo}{\mathrm{o}}
\newcommand{\Mt}{\mathtt{M}}

\newcommand{\p}{\field{P}}

\newcommand{\E}{\field{E}}

\numberwithin{equation}{section}
\def\qed{\hfill$\diamondsuit$}

\def\argmin{\mathop{\mbox{argmin}}}
\theoremstyle{Conjecture} \theoremstyle{example}
\theoremstyle{remark} \theoremstyle{lemma}
\theoremstyle{definition} \theoremstyle{corol}
\theoremstyle{proposition} \theoremstyle{condition}
\newtheorem{assumption}{Assumption}[section]
\newtheorem{remark}{Remark}[section]

\newtheorem{proposition}{Proposition}[section]

\def\bW{{\bf W}}
\DeclarePairedDelimiterX{\Iintv}[1]{\llbracket}{\rrbracket}{\iintvargs{#1}}
\NewDocumentCommand{\iintvargs}{>{\SplitArgument{1}{,}}m}
{\iintvargsaux#1} %
\NewDocumentCommand{\iintvargsaux}{mm} {#1\mkern1.5mu,\mkern1.5mu#2}

\DeclarePairedDelimiter\floor{\lfloor}{\rfloor}

\endlocaldefs

\begin{document}

\begin{frontmatter}
{\title{Simultaneous Sieve Estimation and Inference for Time-Varying Nonlinear Time Series Regression}
}
\runtitle{Simultaneous Sieve Inference}

\begin{aug}
\author[A]{\fnms{Xiucai}~\snm{Ding}\ead[label=e1]{xcading@ucdavis.edu}},
\author[B]{\fnms{Zhou}~\snm{Zhou}\ead[label=e2]{zhou.zhou@utoronto.ca}}
\address[A]{Department of Statistics, University of California, Davis \printead[presep={,\ }]{e1}}

\address[B]{Department of Statistical Sciences, University of Toronto\printead[presep={,\ }]{e2}}
\end{aug}

\begin{abstract}
In this paper, we investigate time-varying nonlinear time series regression for a broad class of locally stationary time series. First, we propose sieve nonparametric estimators for the time-varying regression functions that achieve {uniform consistency}. Second, we develop a unified simultaneous inferential theory to conduct both structural and exact form tests on these functions. Additionally, we introduce a multiplier bootstrap procedure for practical implementation. Our methodology and theory require only mild assumptions on the regression functions, allow for unbounded domain support, and effectively address the issue of identifiability for practical interpretation. Technically, we establish sieve approximation theory for 2-D functions in unbounded domains, {prove two Gaussian approximation results} for affine forms of high-dimensional locally stationary time series, and calculate critical values for the maxima of the Gaussian random field arising from locally stationary time series, which may be of independent interest. Numerical simulations and two data analyses support our results, and we have developed an $\mathtt{R}$ package, $\mathtt{SIMle}$, to facilitate implementation.
\end{abstract}



\end{frontmatter}

\section{Introduction} \label{sec1}
\def\theequation{1.\arabic{equation}}
\setcounter{equation}{0}
\counterwithin{equation}{section}

Understanding the complex functional form of a relationship  between different variables is one of the main aims in modern data sciences. However, achieving this understanding presents multifaceted challenges in the presence of massive data sets. First, variables often interact in nonlinear, intricate ways that might not conform to conventional models and the underlying theory does not narrow down a specific functional or parametric form for the relationship.  Second, the observations usually span long time horizons. Consequently, the relationship is unlikely to remain stable so that the non-stationary relationships between variables evolve, fluctuate, or exhibit different patterns over distinct time periods.

To address these issues, {flexible and time-varying regression models are needed}. In this paper, we   
consider the general time-varying nonlinear time series regression model that 
{\vspace*{-5pt}  
\begin{eqnarray}\label{eq:model}
Y_{i,n}=m_0(t_i)+\sum_{j=1}^r m_j(t_i,X_{j,i,n})+\epsilon_{i,n}, \quad i=1,2,\cdots,n, 
\end{eqnarray}
where $r>0$ is some given fixed integer. Here $\{X_{j,i,n}\}$ are the covariates which can be locally stationary and depend on each other,  $t_i=i/n,$ and 
$\{\epsilon_{i,n}\}$ is a general locally stationary time series  whose covariance depend on both time and the covariates (see Assumption \ref{assum_models}) and satisfies that {$\E(\epsilon_{i,n}|\{X_{j,i,n}\}_{1 \leq j \leq r})=0.$} Moreover, $m_0$ is the time-varying trend, and $m_j, 1 \leq j \leq r,$ are some smooth functions that map the time and covariates to the conditional mean.


To ensure the identifiability of the components $m_j(\cdot,\cdot)$ when $r>1$, we follow the tradition from additive models with i.i.d. inputs, as in \cite{fan2005nonparametric}, and without loss of generality, assume that 
\begin{equation} \label{eq_identiassum}
\mathbb{E}\left(m_j(t_i, X_{j,i,n})\right)=0, \ 1 \leq j \leq r, \ 1 \leq i \leq n. 
\end{equation} 
The model (\ref{eq:model}) provides a more flexible and general framework and a wide range of models for time-varying locally stationary time series regression can fit into it. In fact, it has been used in the literature, particularly when $r=1$, where identifiability is not an issue, for instance, see \cite{KARMAKAR2021,MR4244697, MV,MV1,ZWW}. }
%
We point out that (\ref{eq:model}) can be regarded as a discretized version of the Langevin equation \cite{coffey2012langevin}\vspace*{-5pt}  
\begin{equation*}
\mathrm{d} Y_t= \sum_{j=1}^r  m_j(t, X_{j,t}) \mathrm{d} t+\sigma(t, X_{1,t}, \cdots, X_{r,t})\mathrm{d} M(t),   
\end{equation*}
where $M(t)$ is some martingale. In what follows, we summarize related results from the literature in Section \ref{sec_subrelatedresults} and provide an overview of our findings, highlighting the novel contributions, in Section \ref{sec_overview}.

\subsection{Some related results}\label{sec_subrelatedresults}
In this subsection, we summarize some related results. To the best of our awareness, few results have been established under the general model setup (\ref{eq:model}). {Most of the existing literature focuses on pointwise results based on kernel estimation, typically assuming $r=1$, which avoids identifiability issues, allowing $m_0$ to be combined with $m_1$.}
In \cite{MV}, under classical strong mixing conditions and using locally stationary processes, the author proposed a Nadaraya-Watson (NW) type estimator utilizing a product kernel. Additionally, the author established convergence rates for this estimation on any compact set and demonstrated the pointwise normality of the estimators. These results and ideas were later applied to study structural changes in \cite{MV1}. More recently, in \cite{ZWW}, the authors advanced kernel estimation theory under a more general locally stationary setting using the physical representation. They also applied their findings to model selection problems.

On the other hand, we note that (\ref{eq:model}) has also been studied in the literature under certain structural assumptions. The existing research in this area can be roughly divided into three categories.
\begin{itemize}
\item Time-independent {additive models}, i.e., $m_j(t_i, X_{j,i,n})=m_j(X_{j,i,n}),$ for all $1 \leq j \leq r.$ {In this setting, the model has been extensively studied under the assumption that the samples are either i.i.d. or from a stationary time series.} Especially, the convergence rate and  pointwise  normality of kernel based estimators have been provided in the literature; see \cite{fan2005nonparametric,fan2008nonlinear,hansen2008uniform,horowitz2009semiparametric,kristensen2009uniform, RePEc:pup:pbooks:8355,  linton_wang_2016,10.1214/009053607000000488, 10.1214/009053607000000596,10.1214/07-AOS533}  among others,  and {the  convergence of sieve least square estimators and their pointwise asymptotic normality has also been investigated under various assumptions}; see  \cite{MR3343791,CXH,Chen_2013, CHEN2015447, MR3306927} among others.
\item Time-varying {coefficient models}, i.e.,  $m_j(t_i, X_{j,i,n})=f_j(t_i)X_{j,i,n},$ for all $1 \leq j \leq r.$ In this setting, the kernel estimators and simultaneous inference have been studied for the time-varying coefficients $f_j(\cdot), 1 \leq j \leq r$, see \cite{MR1804172, 10.2307/4140562,park2015varying, MR3160574, MR3310530} among others. Additionally, we mention a recent work \cite{KARMAKAR2021} focusing on conditional heteroscedastic (CH) models, where the authors conducted statistical inference for the time-varying parameters by analyzing the negative conditional Gaussian likelihoods with kernel estimations. Even though CH models are not immediate forms of time-varying linear regression, according to the discussion in \cite[Section C.2.5]{DZ2}, a straightforward transformation can convert them into a time-varying linear regression model.

\item Time-varying nonlinear regression with multiplicative separable structure, i.e., $m_j(t_i, X_{j,i,n})$ $=f_j(t_i) g_j(X_{j,i,n}),$ for all $1 \leq j \leq r.$ In this setting, the kernel estimators and {pointwise} asymptotic normality have been established in \cite{CSW, PEI2018286}, the spline estimators and their asymptotic normality have been studied in \cite{10.2307/43305634} for functional data and in \cite{HHY} for locally stationary time series, among others.    
\end{itemize}

%



In summary, most of the relevant literature focuses on providing kernel-based estimators for the smooth functions in (\ref{eq:model}). Although kernel estimation is a useful approach, it is known to suffer from boundary issues and {lack of automatic adaptability to the smoothness structures of the underlying regression functions}. Moreover, existing works have only established convergence rates on a compact domain, which may limit their applicability when the covariates have unbounded supports. Finally, the literature has primarily focused on proving pointwise {distributional results} for general smooth functions, which may not be useful for conducting {functional or simultaneous inference on the regression coefficient functions.}

Motivated by the above challenges, in the current paper,  we study (\ref{eq:model}) in great generality.  First,  we provide  uniformly consistent nonparamteric sieve estimators for the functions $m_j'$s. {We also address the identifiability issue for the components $m_j'$s when $r>1.$} Second, we develop simultaneous inferential theory for these functions. Our results and methodologies only require mild assumptions on $m_j$'s and we also allow the supports of these functions to be unbounded. We give an overview of our results in Section \ref{sec_overview}.  

\subsection{Overview of our results and novelties}\label{sec_overview}
In this subsection, we give a heuristic overview of our main results. Throughout the paper, for the locally stationary time series, we use the general form of physical representation following \cite{DZ, MR2172215, ZWW, ZW} (c.f. Assumption \ref{assum_models}) which covers a wide range of commonly used locally stationary time series. 
The aim of this paper is to provide a systematic and unified estimation and simultaneous inferential theory for the smooth {coefficient} functions in (\ref{eq:model}) without imposing additional structural assumptions. {To the best of our knowledge, this is the first work in the literature to perform simultaneous nonparametric inference for the general time-varying nonlinear regression (\ref{eq:model}).}   

For the estimation part, we propose the sieve nonparamatric  estimators for the smooth functions. However, unlike in the standard applications of sieve expansions where the smooth functions are supported on a bounded domain \cite{MR3343791, CXH, CHEN2015447, DZ, DZ2,MR3306927}, the covariates in (\ref{eq:model}) are usually supported in $\mathbb{R}$ and hence the sieve methods and the classic approximation theory cannot be applied directly. To address this issue, we use the mapped sieve basis functions \cite{MR2486453,unboundeddomain} which can monotonically and smoothly map the unbounded domain to a compact domain. Consequently, we can construct a hierarchical sieve basis functions as in (\ref{eq_basisconstruction}) and establish the uniform approximation theory for the smooth functions on unbounded domain; see Proposition \ref{thm_approximation} for more details. 
{To address the issue of identifiability when $r > 1$, in the actual implementation, we propose a two-step sieve estimation approach. In the first step, some pilot estimators (c.f. (\ref{eq_firststeptruncationinter}) and (\ref{eq_firststeptruncation})) are constructed using the basis expansions, which are basis-dependent and may not satisfy condition (\ref{eq_identiassum}). Consequently, these estimators can introduce a time-and basis-dependent bias. In the second step, we use the method of sieves again to estimate these bias and correct our pilot estimators;  see Section \ref{sec_removenonzeromeanassumption} for further details. The estimation procedure is computationally cheap and  the proposed estimators are theoretically consistent. Especially, under mild regularity conditions, for short-range dependent locally stationary time series, our estimators achieve uniform consistency under mild regularity assumptions; see Theorem \ref{thm_consistency} and Remark \ref{rmk_minimaxoptimal} for more details. } 


Once the estimation theory is established, we can develop the inferential theory. We provide the simultaneous confidence regions (SCRs) based on our sieve estimators (c.f. (\ref{eq_scrdefinition})). Then we apply the SCRs to conduct structural testings on the smooth functions. For example, we can test whether the functions are time-invariant (c.f. Example \ref{exam_stationarytest}) or have multiplicative separability structure (c.f. Example \ref{exam_seperabletest}). The key technical inputs are to establish  Gaussian approximation results for high-dimensional locally stationary time series of the affine form. We prove such results in Theorems \ref{thm_gaussianapproximationcase} and \ref{thm_uniformconvergencegaussian} of our supplement \cite{suppl} and they can be of independent interest. Especially, when the temporal relation of the locally stationary time series decays fast enough, the approximate rate matches that of \cite{MR3571252}. 
Moreover, to analyze the critical values for the SCRs, we establish the asymptotic distribution of the maximum deviation of the sieve estimators over the joint state-time domain using the device of  volume of tubes  \cite{MR1056331,KS, MR1207215}. For practical implementation, we propose a multiplier bootstrap procedure as in \cite{10.1214/13-AOS1161, MR3174655} which is both theoretically sound and empirically accurate and powerful; see Theorem \ref{thm_boostrapping}. Numerical simulations and real data analysis are provided to support our results and methodologies. An $\mathtt{R}$ package $\mathtt{SMIle}$ is developed to ease the implementation.   

Finally, we point out that even though we focus on the object of the physical form of locally stationary time series, our methodologies and results can be applied to the locally stationary time series considered in \cite{DRW,nason2000wavelet,MV} or even the general non-stationary time series in \cite{DZ2}, after some minor modifications. Before concluding this subsection, we summarize the contributions of the current paper as follows.
\begin{enumerate}
\item[(1).] We propose the sieve estimators for the time-varying nonlinear time series regression (\ref{eq:model}) using mapped sieve basis functions. Our methodology only requires mild assumptions on the smooth functions and we allow the functions to be supported on unbounded domains. Moreover, our estimation procedure effectively addresses the issue of identifiability and {our estimator achieves the uniform consistency}. 
\item[(2).] {To the best of our knowledge, we are the first to} establish a comprehensive  and unified simultaneous inferential theory for the smooth functions in  (\ref{eq:model})  by constructing their SCRs across both $x$ and $t$. Our theory can be used to conduct both structural and exact form tests on the smooth functions.
\item[(3).] We provide a multiplier bootstrap procedure (c.f. Algorithm \ref{alg:boostrapping}) to implement our methodologies in practice. The bootstrap procedure can asymptotically mimic the distributions of the statistics proposed in the inferential theory and provide accurate SCRs. Moreover,  our bootstrap  method is robust and {adaptive to the structures of the locally stationary time series and the functions $m_j$'s}.      
\item[(4).] We establish {two Gaussian approximation} results for affine forms of high dimensional locally stationary time series. {Furthermore, we establish the asymptotic distribution of the maxima of a Gaussian random field, with randomness arising from a locally stationary time series. This allows us to study the critical values of the SCRs.}
\end{enumerate}

The paper is organized as follows. In Section \ref{sec_modelassumption}, we introduce the locally stationary time series and some technical assumptions.  In Section \ref{sec_estimation}, we propose the sieve estimators for the smooth functions and prove their theoretical consistency.  In Section \ref{sec_inference}, we develop the simultaneous inferential theory based on our sieve estimators and consider several important hypothesis testing problems. For practical applications, we propose a multiplier  bootstrap strategy. In Section \ref{sec_numerical}, we conduct extensive numerical simulations to support the usefulness of our results. A real data analysis is provided in Section \ref{sec_realdata}. An online supplement \cite{suppl} provides technical details and additional results. Specifically, Section \ref{sec_sieves} of \cite{suppl} lists commonly used sieve basis functions, while Section \ref{sec_detailedols} provides more details on the sieve estimation and Section \ref{sec_parameterchoice} discusses choices of tuning parameters. Additional discussions are included in Section \ref{appendix_additionalremark}, and further numerical results are presented in Section \ref{addtional_numerical}. Finally, technical proofs are deferred to Section \ref{sec_techinicalproof}, with auxiliary lemmas collected in Section \ref{sec_auxililarylemma}.


\vspace{4pt}

\noindent {\bf Conventions.} For any random variable $Z \in \mathbb{R},$ we denote its $q$-norm by $\| Z\|_q=\left( \mathbb{E}|Z|^q \right)^{1/q}. $ For simplicity, we  write $\| Z\| \equiv \|Z\|_2.$ Moreover, for any deterministic vector $\bm{x}=(x_1,\cdots, x_d)^\top$ $\in \mathbb{R}^d,$ we denote its Euclidean norm as $|\bm{x}|=\sqrt{\sum_{i=1}^d x_i^2}.$
For two sequences of real values $\{a_n\}$ and $\{b_n\},$ we write $a_n=\OO(b_n)$ if $|a_n| \leq C |b_n|$ for some constant $C>0$ and $a_n=\oo(b_n)$ if $|a_n| \leq \epsilon_n |b_n|$ for some positive sequence $\{\epsilon_n\}$ that $\epsilon_n \downarrow 0$ as $n \rightarrow \infty.$ Moreover, if $a_n=\OO(b_n)$ and $b_n=\OO(a_n),$ we write $a_n \asymp b_n.$ For a sequence of random variables $\{X_n\}$ and a positive sequence of $\{C_n\},$ we use the notation $X_n=\OO_{\mathbb{P}}(C_n)$ to say that $X_n/C_n$ is stochastically bounded, and  $X_n=\oo_{\mathbb{P}}(C_n)$ to say that $X_n/C_n$ converges to $0$ in probability. Throughout the paper, we omit the subscript $n$ without causing further confusion.

\section{The model and assumptions}\label{sec_modelassumption}
In this section, we introduce the model, some related notations and assumptions. {For notional simplicity, till the end of the paper, we simply write $Y_i \equiv Y_{i,n}, X_{j,i} \equiv X_{j,i,n}$ and $\epsilon_i \equiv \epsilon_{i,n}.$} In the current paper, we consider the family of locally stationary time series for ${Y_i, X_{j,i}, \epsilon_i}$ as introduced  in \cite{DRW, ZW} using the physical representation. We now state the basic setup of our model in Assumption \ref{assum_models}.


\begin{assumption}[Model setup] \label{assum_models} For the regression (\ref{eq:model}), throughout the paper, 
we assume that all $\{X_{j,i}\}, \{\epsilon_i\}$ are locally stationary time series such that for all $1 \leq j \leq r,$ \vspace*{-5pt}  
\begin{equation}\label{eq_setting2}
X_{j,i}=G_j(t_i, \mathcal{F}_i), \ \epsilon_i=D(t_i, \mathcal{F}_i), \ t_i=\frac{i}{n}. 
\end{equation}
Here $\mathcal{F}_i=(\cdots, \eta_{i-1},\eta_i),$  where we assume that    $\{\eta_i\}$ are i.i.d. random elements with {$\E(\epsilon_{i}|\{X_{j,i}\}_{1 \leq j \leq r})=0.$}

Moreover, we assume that $G_j, 1 \leq j \leq r$, and $D: [0,1] \times \mathbb{R}^{\infty} \rightarrow \mathbb{R}$ are measurable functions such that for any fixed $t \in [0,1],$ $ G_j(t,\cdot), 1 \leq j \leq r$ and $D(t,\cdot)$ are property defined random variables. Furthermore, we assume all these functions are stochastic Lipschitz in $t$ such that for some $q>2,$ some constant $C>0$ and any $s,t \in [0,1]$ \vspace*{-5pt}  
\begin{equation}\label{eq_slc}
\max \left\{ \max_{1 \leq j \leq r } \left\| G_j(s, \mathcal{F}_0)-G_j(t, \mathcal{F}_0) \right\|_q, \left\| D(s, \mathcal{F}_0)-D(t, \mathcal{F}_0) \right\|_q  \right\}\leq C|s-t|.   
\end{equation} 
\end{assumption} 

\begin{remark}\label{rmk_modelsetting}
Several remarks are in order. First, (\ref{eq_setting2}) requires that all the time series  admit  physical representations and the data generating mechanisms change slowly over time.  This is a mild assumption in the sense that many commonly used locally stationary time series, linear or non-linear, can be written into such a form. For illustrations, we refer the readers to Examples \ref{exam_linear} and \ref{exam_nonlinear} in our supplement. We also point out that, combining (\ref{eq_setting2}) with (\ref{eq:model}), we see that $Y_i$ also admits a  physical representation. {We also note that equation (\ref{eq_setting2}) implies that our model encompasses the case where $\epsilon_i = H(t_i, X_{1,i}, \dots, X_{r,i}, \eta_i)$.} Second, (\ref{eq_slc}) is used to ensure the local stationarity of the time series which is a commonly used assumption in the study of locally stationary time series. Third, (\ref{eq_setting2}) is a flexible framework and has been frequently used in other time series regression models, for example \cite{10.1214/12-AOS1010,ZWW,MR2758526}. Together with our regression (\ref{eq:model}), we extend the classes of time varying nonlinear regression models. Fourth, we point out that if one sets  $X_{j,i}=Y_{i-j},$ i.e., $Y_i=G(t_i, \mathcal{F}_i)$ for some function $G$ and $X_{j,i}=G(t_{i-j}, \mathcal{F}_{i-j}),$ (\ref{eq:model}) with (\ref{eq_setting2}) {can be used to efficiently approximate} the locally stationary nonlinear AR models as in \cite{MV} and the time-varying ARCH and GARCH models in \cite{KARMAKAR2021}. Finally, to simplify the discussion, we focus our analysis on locally stationary time series as described in (\ref{eq_setting2}). However, similar to the analysis in Section 2 of \cite{DZ2}, our discussion—particularly the estimation part—can be easily extended to general nonstationary time series. Specifically, (\ref{eq_setting2}) can be generalized to $X_{j,i}=G_{j,i}(\mathcal{F}_i), \ \epsilon_i=D_i(\mathcal{F}_i). $ Since this extension is not the primary focus of the current paper, we will explore it in future works.

\end{remark}

In what follows, 
we employ the \emph{physical dependence measure} \cite{DRW,MR2172215} to quantify the temporal dependence of the time series.  For simplicity, we state the definition using $\{\epsilon_i\}$.

\begin{definition}[Physical dependence measure]\label{defn_physcialdependence} Let $\{\eta_i'\}$ be an i.i.d. copy of $\{\eta_i\}.$ Assume that for some $q>2,$ $\| \epsilon_i\|_q<\infty.$  Then for $k \geq 0,$ we define the physical dependence measure of $\{\epsilon_i\}$ as\vspace*{-5pt}  
\begin{equation*}
\delta(k,q)=\sup_t \| D(t, \mathcal{F}_0)-D(t, \mathcal{F}_{0,k}) \|_q,  
\end{equation*}
where $\mathcal{F}_{0,k}=(\mathcal{F}_{-k-1}, \eta_{-k}', \eta_{-k+1}, \cdots, \eta_0).$ For convenience, we denote $\delta(k,q)=0$ when $k<0.$ Similarly, we can define the physical dependence measures for $\{X_{j,i}\}, 1 \leq j \leq r.$  
\end{definition}

The physical dependence measure provides a convenient tool to study the temporal dependence of time series. It quantifies the magnitude of change in the system's output when the input of the system $k$ steps ahead is replaced by an i.i.d. copy. Moreover, the strong temporal dependence of the locally stationary time series can be controlled in terms of $\delta(k,q)$ and the concentration inequalities (c.f. Lemmas \ref{lem_concentration} and \ref{lem_mdependent}  of our supplement \cite{suppl}) can be established based on them. For more details on physical dependence measure, we refer the readers to \cite{MR2485027,MR3114713, MR2172215,WUAOP, MR2827528}. In addition, for the purpose of illustration, in Examples \ref{exam_linear} and \ref{exam_nonlinear} of our supplement, we explain how the physical dependence measure can be calculated easily for the commonly used locally stationary time series.  

Armed with Definition \ref{defn_physcialdependence}, we impose the following assumption to  ensure that the temporal dependence decays fast enough so that the time series has short-range dependence. 
\begin{assumption}\label{assum_physical}
Suppose that there exists some constant $\tau>1$ and some constant $C>0$ such that \vspace*{-5pt}  
\begin{equation}\label{eq_physcialrequirementone}
\max_{0 \leq j \leq r} \delta_j(k,q) \leq C k^{-\tau}, \ \text{for all} \ k \geq 1,
\end{equation}
where $\delta_j(k,q), 1 \leq j \leq r,$ are the physical dependence measures of $\{X_{j,i}\},  1 \leq j \leq r,$ and $\delta_{0}(k,q)$ is that of $\{\epsilon_i\}.$ 
\end{assumption}

Assumption \ref{assum_physical} is readily satisfied by commonly used locally stationary time series models. In Section \ref{appendix_additionalremark1} of our supplement, we present several classes of locally stationary time series utilizing physical representation and explain how these assumptions are easily met. 

\section{Sieve estimation for time-varying nonlinear regression}\label{sec_estimation} 
In this section, we propose a nonparametric  approach to estimating the nonlinear functions $m_j(t,x), \ 1 \leq j \leq r$ {and the time-varying intercept $m_0(t)$} based on the method of sieves \cite{CXH}. {To ensure identifiability under assumption (\ref{eq_identiassum}), our proposed method involves two steps. In the first step (Section \ref{sec_generaldiscussion}), we construct pilot estimators using the expansion of sieves. However, these estimators may not satisfy assumption (\ref{eq_identiassum}). In the second step (Section \ref{sec_removenonzeromeanassumption}), we apply a correction again based on the method of sieves to ensure compliance with assumption (\ref{eq_identiassum}). The uniform consistency of our estimators is established in Section \ref{sec_theorecticalpropertyestimator}.}



\subsection{Pilot estimators based on sieve least squares}\label{sec_generaldiscussion}
In this subsection, we use the sieve least square estimation approach to {estimate $m_0(t)$ and $m_j(t,x),  1 \leq j \leq r$.}  Our methodology is nonparametric and utilizes    the sieve basis expansion. Before stating our results, we first pause to provide a heuristic overview of the method. The details will be offered in Sections \ref{sec_2dsieves} and \ref{sec_olssieveestimation}.

 Since we can only observe one realization of the time series, it is natural to impose some smoothness condition on $m_0(t)$ and $m_j(t,x)$ (c.f. Assumption \ref{assum_smoothnessasumption}). First, under these conditions, {$m_0(t)$ can be well approximated by {\vspace*{-5pt}  
 \begin{equation}\label{eq_intercepttruncation}
 m_{0,c}(t)=\sum_{\ell=1}^{c_0} \beta_{0,\ell} \phi_\ell(t),
 \end{equation}}
 and $m_j(t,x)$ can be well approximated by 
}  $m_{j,c,d}(t,x)$ (c.f. Proposition \ref{thm_approximation}) denoted as {\vspace*{-5pt}  
\begin{equation}\label{eq_firststeptruncation1}
m_{j, c,d}(t,x)=\sum_{\ell_1=1}^{c_j} \sum_{\ell_2=1}^{d_j} \beta_{j, \ell_1, \ell_2} b_{\ell_1, \ell_2}(t,x).
\end{equation} } 
Here $\{\phi_\ell(t)\}$ are commonly used sieve basis functions in $[0,1]$ and $\{b_{\ell_1, \ell_2}(t,x)\}$ are the sieve basis functions on $[0,1] \times \mathbb{R},$ and $\{\beta_{0,\ell}\}$ and $\{\beta_{j, \ell_1, \ell_2}\}$ are the coefficients to be estimated. Here {$c_j \equiv c_j(n), \ d_j \equiv d_j(n)$} depend on the smoothness of $m_j, 0 \leq j \leq r.$  Second, due to the unboundedness of domain of $x,$ $\{b_{\ell_1, \ell_2}(t,x)\}$ needed to be properly constructed. For detailed discussion,
we refer the readers to Section \ref{sec_2dsieves}.

 Finally, in view of (\ref{eq_intercepttruncation}) and (\ref{eq_firststeptruncation1}), given a set of sieve basis functions, it suffices to estimate the coefficients. We will use the ordinary least square (OLS) to achieve this goal and obtain the estimators $\{\widehat{\beta}_{0,\ell}\}$ and $\{\widehat{\beta}_{j, \ell_1, \ell_2}\}$. This aspect will be discussed in Section \ref{sec_olssieveestimation}. Based on the above calculation, we can obtain our pilot estimators below \vspace*{-5pt}  
 \begin{equation}\label{eq_firststeptruncationinter}
\widehat{m}^*_{0, c}(t,x)=\sum_{\ell=1}^c  \widehat{\beta}_{0, \ell} \phi_{\ell}(t), 
\end{equation} 
and \vspace*{-8pt}  
\begin{equation}\label{eq_firststeptruncation}
\widehat{m}^*_{j, c,d}(t,x)=\sum_{\ell_1=1}^c \sum_{\ell_2=1}^d \widehat{\beta}_{j, \ell_1, \ell_2} b_{\ell_1, \ell_2}(t,x). 
\end{equation}  

{The pilot estimators (\ref{eq_firststeptruncationinter}) and (\ref{eq_firststeptruncation}) are constructed using the basis expansions (\ref{eq_intercepttruncation}) and (\ref{eq_firststeptruncation1}), which are basis-dependent and may not satisfy condition (\ref{eq_identiassum}). Consequently, these estimators can introduce a time- and basis-dependent biases. This issue will be addressed in Section \ref{sec_removenonzeromeanassumption} through a correction procedure.} 
 
\subsubsection{Mapped hierarchical sieves for 2-D smooth functions on unbounded domain}\label{sec_2dsieves}
In this subsection, we discuss how to use sieve basis functions to approximate a smooth 2-D function. Note that for $m_j(t,x), 1 \leq j \leq r,$ when $x$ is defined on a compact domain in $\mathbb{R}$, the results have been established, for example, see \cite[Section 2.3.1]{CXH} for a comprehensive review. However, in many real applications, the domain of $x$ is unbounded. Therefore, we need to modify the commonly used sieves to accommodate for practical applications. 

In the literature, there exist three different approaches to deal with the unbounded domain. The first  method is to apply a finite interval method, such as Chebyshev polynomials, to an interval $x \in [-L,L]$ where $L$ is large but
finite. This method is also known as domain truncation. The second way is to use a basis that is intrinsic to the unbounded domain such as Hermite functions. The third approach is to map the infinite interval into a finite domain through a change of coordinate and then apply a finite interval method. We refer the readers to \cite{MR1874071,MR2486453,CXH,unboundeddomain} for a review.

In the current paper, we employ the third method that we will apply some suitable mappings to map the unbounded domain to a compact one. {The reasons are twofold. First, the first and second methods typically result in slower approximate rates, requiring more basis functions to achieve a desired level of accuracy. Second, in most statistical applications, the functions have rapidly decaying tails, which ensures the appropriateness of the mapping method.} The mapping idea is popular in numerical PDEs for handling the boundary value problems. For a review of the mapping strategy, we refer the readers to \cite{MR2486453,unboundeddomain}. In what follows, without loss of generality, we assume that the domain for $x$ is either $\mathbb{R}$ or $\mathbb{R}_+.$ The mappings will map the domain to a compact interval, say, $[-1,1].$ We state the definitions of the mappings as follows.  
\begin{definition}[Mappings]\label{defn_mappings}
For some positive scaling factor $s>0,$ consider a family of mappings of the form: \vspace*{-3pt}  
\begin{equation*}
x=g(y;s), \ s>0, \ y \in I:=[-1,1], \ x \in \Lambda:=\mathbb{R}_+ \ \text{or} \ \mathbb{R},
\end{equation*}
such that  \vspace*{-5pt}  
\begin{align*}
& \frac{\dd x}{\dd y}=g'(y;s)>0, \ y \in I; \\
& g(-1; s)=
\begin{cases}
0, & \Lambda=\mathbb{R}_+ \\
-\infty, & \Lambda=\mathbb{R}
\end{cases}, \ 
g(1; s)=\infty.
\end{align*}
Since the above mapping is invertible, we denote the inverse mapping by 
\begin{equation}\label{eq_usxdefinition}
y=g^{-1}(x;s)=u(x;s), \ x \in \Lambda, y \in I, s>0.
\end{equation}
\end{definition}
Based on the above definition, it is easy to see that $\frac{1}{2}u(x;s)+\frac{1}{2}$ will map $\mathbb{R}$ or $\mathbb{R}_+$ to $[0,1].$ {Moreover, as recommended in \cite{unboundeddomain}, it is advisable to use $s=1$ for the mappings in practice.} In Section \ref{appendix_mapping} of the supplement, we provide several examples of commonly used mappings.

Throughout the paper, without loss of generality, we assume that $m_j(t,x), 1 \leq j \leq r,$ takes value on the whole real line $\mathbb{R}$ for $x.$ Similar arguments apply when $x \in \mathbb{R}^+.$ Using the mappings in Definition \ref{defn_mappings}, due to  monotonicity,  we have linked $m_j(t,x): [0,1] \times \mathbb{R} \rightarrow \mathbb{R} $ to 
\begin{equation}\label{eq_transformedmjty}
\widetilde{m}_j(t, y):=m_j(t, g(2y-1;s)): [0,1] \times [0,1] \rightarrow \mathbb{R}. 
\end{equation} 
Based on the above arguments, we can construct a sequence of mapped sieve basis functions following \cite{unboundeddomain}. Let $\{\phi_i(\cdot)\}$ be an orthonormal base of the smooth functions defined on $[0,1].$ Recall (\ref{eq_usxdefinition}). Denote the mapped version of $\{\phi_i\}$ as $\{\widetilde{\phi}_i\}$ such that for $x \in \mathbb{R}$ \vspace*{-5pt}  
\begin{equation}\label{eq_defnmappedbasis}
\widetilde{\phi}_i(x)=\phi_i \circ \sfy(x), \ \sfy(x):=\frac{u(x;s)+1}{2}. 
\end{equation} 
Note that $\{\widetilde{\phi_i}\}$ is a sequence of orthogonal basis of the functional space defined on $\mathbb{R}$ \cite{unboundeddomain}. Denote $\{\varphi_i\}$ as the orthonormal basis functions based on $\{\widetilde{\phi_i}\}.$ Armed with $\{\phi_i\}$ and $\{\varphi_i\},$ we can construct the hierarchical sieve basis functions following \cite{CXH} as \vspace*{-2pt}  
\begin{equation}\label{eq_basisconstruction}
\{\phi_i(t)\} \otimes \{\varphi_j(x)\}. 
\end{equation} 
For examples of the commonly used sieve and mapped sieve basis functions and their properties, we refer the readers to  Section \ref{sec_sieves} of our supplement \cite{suppl}. We also refer the reader to Figure \ref{fig_basis} of our supplement for an illustration of the basis and the corresponding mapped basis functions. In our $\mathtt{R}$ package $\mathtt{SIMle},$ one can use $\mathtt{bs.gene}$ to generate many commonly used basis functions and  $\mathtt{bs.gene.trans}$ for the mapped basis functions.

Then we proceed to state the approximation theory using (\ref{eq_firststeptruncation1}), where $\{b_{\ell_1, \ell_2}(t,x)\}$ is the collection of  sieve basis functions as in (\ref{eq_basisconstruction}).  








\begin{assumption}\label{assum_smoothnessasumption}
{For some constant $\mathsf{m}_0,$ denote $\mathtt{C}^{\mathsf{m}_0}([0,1])$ as the function space on $[0,1]$ of continuous functions that have continuous first $\mathsf{m}_0$ derivatives.  We assume that the intercept $m_0(t) \in \mathtt{C}^{\mathsf{m}_0}([0,1]).$}
Moreover, for some sequence of constants $\mathsf{m}_{kj}, k=1,2, 1 \leq j \leq r,$ we assume that for $\widetilde{m}_j(t,y), 1 \leq j \leq r$ in (\ref{eq_transformedmjty})
\begin{equation}\label{eq_importantregulaityassumption}
\frac{\partial\widetilde{m}_j}{\partial t} \in \mathtt{C}^{\mathsf{m}_{1j}}([0,1]),  \ \text{and} \  \frac{\partial \widetilde{m}_j}{\partial y} \in \mathtt{C}^{\mathsf{m}_{2j}}([0,1]),  
\end{equation}
and for all $\alpha_j=\alpha_{j1}+\alpha_{j2}$ with $0 \leq \alpha_{jk} \leq \mathsf{m}_{kj}, k=1,2,$ $\partial^{\alpha_j} \widetilde{m}_j/(\partial t^{\alpha_{j1}} \partial y^{\alpha_{j2}})$ is always uniformly bounded and continuous in both $t \in [0,1]$ and $y \in [0,1].$ 
%
\end{assumption}

\begin{remark}
We provide a few remarks on Assumption \ref{assum_smoothnessasumption}. First, the first condition in (\ref{eq_importantregulaityassumption}) requires that $\widetilde{m}_j(t,\cdot)$ is a smooth function with respect to $t$ on the compact interval $[0,1].$  In view of the definition (\ref{eq_transformedmjty}), it suffices that the original function $f_j(t):=m_j(t,\cdot)$ is also a smooth function with respect to $t$ of order $\mathsf{m}_{1j}.$ Second,  the second condition in (\ref{eq_importantregulaityassumption}) requires that after $\mathbb{R}$ being mapped to $[0,1]$ using the mappings satisfying Definition \ref{defn_mappings}, $\widetilde{m}_j(\cdot,y)$ is a smooth function with respect to $y$ on $[0,1].$ In view of the mappings in Definition \ref{defn_mappings}, we see that it will be satisfied if the original function $h_j(x):=m_j(\cdot, x)$ belongs to the generalized Schwartz space of order $\mathsf{m}_{2j},$ i.e., $h_j(x) \in \mathtt{C}^{\mathsf{m}_{2j}}(\mathbb{R}),$ and for all $0 \leq \alpha,\beta \leq \mathsf{m}_{2j},$ there exists some positive constant $C_{\alpha,\beta}$ so that $|x|^{\alpha} |\partial^\beta h_j(x)|\leq C_{\alpha,\beta};$ see Section \ref{sec_Schwartzfunction} of our supplement for more details. For an illustration, we can see the decay properties of the mapped basis functions as in the right panel of Figure \ref{fig_basis} of the supplement.

\end{remark}

The main result of this subsection can be summarized as follows. 

\begin{proposition}\label{thm_approximation}
Suppose Assumption \ref{assum_smoothnessasumption} holds. Furthermore, let ${\phi_i(t)}$ be commonly used sieve basis functions and
let ${\varphi_j(x)}$ be the corresponding mapped sieve basis functions. 
{For $m_{0,c}(t)$ defined in (\ref{eq_intercepttruncation}) and   $m_{j,c,d}(t,x)$ defined in (\ref{eq_firststeptruncation1}) with $b_{\ell_1, \ell_2}(t,x)$ constructed according to (\ref{eq_basisconstruction}), we have that}\vspace*{-5pt} 
\begin{equation*}
\sup_{t \in [0,1]} \left| m_0(t)-m_{0, c}(t) \right|={\OO\left( c_0^{-\mathsf{m}_0} \right)},
\end{equation*}
\vspace*{-5pt} and for $1 \leq j \leq r$ \vspace*{-5pt} 
\begin{equation*}
\sup_{t \in [0,1]} \sup_{x \in \mathbb{R}}\left| m_j(t,x)-m_{j, c,d}(t,x) \right|={\OO\left( c_j^{-\mathsf{m}_{1j}}+d_j^{-\mathsf{m}_{2j}} \right)}. 
\end{equation*}
\end{proposition}

\begin{remark}
Proposition \ref{thm_approximation} establishes the approximation results for a 2-D function on an unbounded domain using the mapped sieve basis functions. Similar results have been established for the compact domain by directly using the basis functions $\{\phi_i(t)\} \otimes \{\phi_j(x)\}$ instead of (\ref{eq_basisconstruction}); see \cite[Section 2.3.1]{CXH} for a comprehensive summary.
\end{remark}
%
%

\subsubsection{Pilot estimation for the coefficients}\label{sec_olssieveestimation}
{According to the approximation result in Proposition \ref{thm_approximation} and the definitions of $m_{0,c}(t)$ in (\ref{eq_intercepttruncation}) and $m_{j,c,d}(t,x)$ in (\ref{eq_firststeptruncation1}), we find that in order to estimate $m_j(t,x),$ it suffices to estimate the coefficients $\{\beta_{0,\ell}\}$ and $\{\beta_{j,\ell_1, \ell_2}\}$ and then construct the pilot estimator as in (\ref{eq_firststeptruncationinter}) and  
(\ref{eq_firststeptruncation}).} 

In what follows, we will estimate these coefficients simultaneously using one OLS. Under some mild conditions, inserting the results of Proposition \ref{thm_approximation} into (\ref{eq:model}) will result in { \vspace*{-5pt} 
\begin{equation*}
\small Y_i=m_{0,c}(t_i)+\sum_{j=1}^r m_{j,c,d}(t_i,X_{j,i})+\epsilon_i+\mathrm{o}_{\mathbb{P}}(1). 
\end{equation*}
Since $m_{j,c,d}, 1 \leq j \leq r,$ may not satisfy (\ref{eq_identiassum}), we rewrite the above equation as follows \vspace*{-5pt} 
\begin{equation*}
Y_i=\left( m_{0,c}(t_i)+\sum_{j=1}^r \chi_j(t_i) \right)+\sum_{j=1}^r \left(m_{j,c,d}(t_i,X_{j,i})-\chi_j(t_i)\right)+\epsilon_i+\mathrm{o}_{\mathbb{P}}(1),
\end{equation*}
where for $1 \leq j \leq r,$ 
\begin{equation}\label{eq_meanchisquare}
\chi_j(t_i) \equiv \chi_{j,c,d}(t_i):=\mathbb{E}(m_{j,c,d}(t_i,X_{j,i})). 
\end{equation} 
To ensure identifiability, it is necessary to estimate both $m_{j,c,d}(t,x)$ and its mean $\chi_j(t)$. In this section, we will construct a pilot estimator for $m_{j,c,d}(t,x)$ as in (\ref{eq_firststeptruncation})  without addressing the identifiability issue, and then account for it by estimating $\chi_j(t)$ in the next section. Note  that {for the basis functions  $\{b_{\ell_1, \ell_2}(t,x)\}$  constructed according to (\ref{eq_basisconstruction}) and using (\ref{eq_intercepttruncation}) and (\ref{eq_firststeptruncation1}), we can write }\vspace*{-5pt}  
\begin{align}\label{eq_linearregression}
Y_i&=m_{0,c}(t_i)+\sum_{j=1}^r m_{j,c,d}(t_i,X_{j,i})+\epsilon_i+\mathrm{o}_{\mathbb{P}}(1) \nonumber \\
& {=\sum_{\ell=1}^{c_0} \beta_{0,\ell} \phi_\ell(t_i)+ \sum_{j=1}^r \sum_{\ell_1=1}^{c_j} \sum_{\ell_2=\mathsf{g}}^{d_j} \beta_{j, \ell_1, \ell_2} b_{\ell_1, \ell_2}(t_i,X_{j,i})+\epsilon_i+\mathrm{o}_{\mathbb{P}}(1)}, 
\end{align} 
where {for the first mapped orthonormal basis fucntion $\varphi_1(x)$ as in (\ref{eq_basisconstruction}),} $\mathsf{g}=1$ if $\varphi_1(x) \not\equiv \text{constant} $ and otherwise $\mathsf{g}=2,$ to avoid the possible issue of multicollinearity. For simplicity, in what follows, we assume $\mathsf{g} = 1$, with straightforward modifications applicable when $\mathsf{g} = 2$. The above equation is essentially a linear regression so we can now apply OLS to estimate the unknown coefficients.

Throughout the paper, we will use the following short-hand notation \vspace*{-5pt}  
\begin{equation}\label{eq_defnp}
p:=\sum_{j=1}^r c_j d_j, \ \text{and} \ \mathsf{p}:=p+c_0. 
\end{equation}
Let the vector $\bm{\beta}_1=(\beta_1, \cdots, \beta_{p})^\top \in \mathbb{R}^{p}$ collect all these coefficients $\{\beta_{j, \ell_1, \ell_2}\}$ in the order of the indices $\Iintv{1,r} \times \Iintv{1,c_j} \times \Iintv{1,d_j}.$ Moreover, let $\bm{\beta}_0=(\beta_{0,1}, \cdots, \beta_{0,c})^\top \in \mathbb{R}^{c_0}$ and $\bm{\beta}=(\bm{\beta}_0^\top, \bm{\beta}_1^\top)^\top \in \mathbb{R}^{\mathsf{p}}.$ The OLS estimator for $\bm{\beta}$ can be written as 
\begin{equation}\label{eq_betaolsform}
\widehat{\bm{\beta}}=(W^\top W)^{-1} W^\top \bm{Y},
\end{equation}  
where $\bm{Y}=(Y_i)_{1 \leq i \leq n} \in \mathbb{R}^{n}$ and $W$ is the design matrix whose detailed construction based on (\ref{eq_linearregression}) can be found in 
Section \ref{sec_detailedols} of our supplement.{ We can then partition $\widehat{\bm{\beta}}=(\widehat{\bm{\beta}}_0^\top,\widehat{\bm{\beta}}_1^\top)^\top$ in the same way as in the definition of $\bm{\beta}.$ Our proposed pilot estimator for $m_{0,c}(t)$ (c.f.  (\ref{eq_firststeptruncationinter})) can be written as { \vspace*{-5pt}  
\begin{equation}\label{eq_interceptestimation}
\widehat{m}^*_{0,c}(t)=\widehat{\bm{\beta}}_0^\top \bm{\phi}_0(t),
\end{equation} }
where $\bm{\phi}_0(t) \in \mathbb{R}^{c_0}$ is the collection of the basis functions $\{\phi_\ell(t)\}_{1 \leq \ell \leq c_0}.$ 

In addition, for $1 \leq j \leq r,$ denote the diagonal matrix $\mathsf{I}_j \in \mathbb{R}^{p \times p}$ such that $\mathsf{I}_j=\oplus_{k=1}^ r \delta_k(j) \mathbf{I}_{c_jd_j},$ where $\oplus$ is the direct sum for matrices, $\delta_k(j)=1$ when $k=j$ and $0$ otherwise, and $\mathbf{I}_{c_jd_j}$ is the $c_jd_j \times c_jd_j$ identity matrix. Furthermore,  let $\bm{b}_j \in \mathbb{R}^{c_jd_j}$ be the collection of the basis functions $\{b_{i,j}(t,x)\}_{\{1 \leq i \leq c_j, 1 \leq j \leq d_j \}}$ and $\bm{b}=(\bm{b}_1^\top, \cdots, \bm{b}_r^\top)^\top \in \mathbb{R}^p.$  Then our proposed pilot estimator for $m_{j}(t,x)$ (c.f. (\ref{eq_firststeptruncation})) can be written as {\vspace*{-5pt}  
\begin{equation}\label{eq_proposedestimator}
\widehat{m}^*_{j,c,d}(t,x)=(\widehat{\bm{\beta}}_1 \mathsf{I}_j)^\top \bm{b}.  
\end{equation}}
}  
\begin{remark}
{In the literature (see Chapter 8.5.2 of the monograph \cite{fan2008nonlinear} or \cite{bookchap}), the estimation procedure within the additive model framework employs the backfitting algorithm, which iteratively estimates the functions one at a time. In our current paper, since $r$ is bounded and the unknown functions are assumed to be smooth, the total number of unknown parameters in $\bm{\beta}$, denoted as $\mathsf{p}$ in (\ref{eq_defnp}), typically diverges at a much slower rate compared to the sample size $n$. Consequently, all of them can be estimated using a single OLS regression.} 
\end{remark}

\subsection{Bias correction: towards satisfying the assumption (\ref{eq_identiassum})}\label{sec_removenonzeromeanassumption} The previous section provides the pilot estimators  (\ref{eq_firststeptruncationinter}) and (\ref{eq_firststeptruncation}), which are basis-dependent and may not satisfy condition (\ref{eq_identiassum}). Recall (\ref{eq_firststeptruncation1}) and (\ref{eq_meanchisquare}). To ensure (\ref{eq_identiassum}) and correct the bias part, the time-varying mean function $\chi_j(t), \ 1 \leq j \leq r$, should  be estimated and removed from its pilot estimators.  We will address this issue via a bias correction procedure in this section.
Under some regularity conditions, we can write that\vspace*{-5pt}  
\begin{equation}\label{eq_trueintegralexpansion}
\chi_j(t_i)=\sum_{\ell_1=1}^{c_j} \sum_{\ell_2=1}^{d_j} \beta_{j, \ell_1, \ell_2} \phi_{\ell_1}(t_i) \mathbb{E} \varphi_{\ell_2}(X_{j,i}).  
\end{equation} 

Since the coefficients $\{\beta_{j, \ell_1, \ell_2}\}$ has been estimated in (\ref{eq_betaolsform}), it suffices to estimate $\mathbb{E} \varphi_{\ell_2}(X_{j,i}).$ For notional convenience, we now introduce a sequence of functions $\vartheta_{j, \ell_2}(t_i), \ 1 \leq j \leq r, \ 1 \leq \ell_2 \leq d_j,$ defined as  $\vartheta_{j,\ell_2}(t_i)=\mathbb{E} \varphi_{\ell_2}(X_{j,i}).$
Moreover, under Assumption \ref{assum_models}, we have that for all $t \in [0,1]$ \vspace*{-5pt}  
\begin{equation}\label{eq_definitionvartheta}
\vartheta_{j,\ell_2}(t)=\mathbb{E} \varphi_{\ell_2}(G_j(t,\cdot)).
\end{equation}
Note that $\varphi_{\ell_2}(x)$ is a mapped basis function. If it is smooth and $G_j(t, \cdot)$ satisfies certain regularity conditions, then $\vartheta_{j,\ell_2}$ inherits the smoothness. Consequently, we can use the method of sieves again to estimate them. We summarize these assumptions as follows. 
\begin{assumption}\label{assum_baiscorrectionassump}
For $1 \leq j \leq r$ and $1 \leq \ell_2 \leq d_j,$ suppose there exist a sequence of constants constant $\{\mathsf{n}_{j, \ell_2}\},$ so that the sequence of functions $\{\vartheta_{j, \ell_2}(t)\}$ defined in (\ref{eq_definitionvartheta}) are $\mathtt{C}^{\mathsf{n}_{j, \ell_2}}([0,1]).$
\end{assumption}

Under the above assumptions, it is sufficient to estimate the sequence of functions $\{\vartheta_{j,\ell_2}(t)\}$ using the method of sieves again, as discussed in Section \ref{sec_generaldiscussion}. More concretely, by denoting $\epsilon_{j,\ell_2,i}=\varphi_{\ell_2}(X_{j,i})-\vartheta_{j,\ell_2}(t_i),$ we have the following regression equations
\begin{equation}\label{eq_equationsetupsetup}
\varphi_{\ell_2}(X_{j,i})=\vartheta_{j,\ell_2}(t_i)+\epsilon_{j,\ell_2,i}, \ i=1,2,\cdots, n.
\end{equation}
Similar to (\ref{eq_intercepttruncation}), for some constant $c_{j,\ell_2},$ $\vartheta_{j,\ell_2}(t)$ can be approximated by { \vspace*{-5pt}  
\begin{equation}\label{eq_definemorenotations}
\vartheta_{j,\ell_2,c}(t)=\sum_{k=1}^{c_{j,\ell_2}} \gamma_k \phi_k(t), \ \text{where} \ \gamma_k \equiv \gamma_k(j,\ell_2),  
\end{equation}
where we recall that $\{\phi_k(t)\}$ are the basis functions defined on $[0,1]$.} Together with (\ref{eq_equationsetupsetup}), we can estimate the coefficients $\{\gamma_k\}$ in the same ways as in Section \ref{sec_olssieveestimation}. As a consequence, we can obtain the estimators for $\{\vartheta_{j,\ell_2}(t)\},$ denoted as $\{\widehat{\vartheta}_{j,\ell_2}(t)\}.$

With the above preparation, we can estimate $\chi_j(t_i)$ using
\begin{equation}\label{eq_estimationchi}
\widehat{\chi}_j(t)=\sum_{\ell_1, \ell_2} \widehat{\beta}_{j,\ell_1, \ell_2} \phi_{\ell_1}(t) \widehat{\vartheta}_{j,\ell_2}(t).
\end{equation} 
Applying the above corrections to the pilot estimators (\ref{eq_interceptestimation}) and (\ref{eq_proposedestimator}),  we can finally obtain our estimators for  $m_0(t)$ and $m_j(t,x)$  using {\vspace*{-5pt}  
\begin{equation}\label{eq_generalestimateone}
\widehat{m}_0(t)=\widehat{m}^*_{0,c}(t)+\sum_{j=1}^r \widehat{\chi}_j(t),
\end{equation} }
and for $1 \leq j \leq r$ {\vspace*{-5pt}  
\begin{equation}\label{eq_generalestimatetwo}
\widehat{m}_j(t,x)=\widehat{m}^*_{j,c,d}(t,x)-\widehat{\chi}_j(t).
\end{equation}}


\subsection{Theoretical properties of the proposed estimators}\label{sec_theorecticalpropertyestimator}
In this subsection, we prove the uniform consistency for the proposed estimators (\ref{eq_generalestimateone}) and (\ref{eq_generalestimatetwo}). We first prepare some notations and assumptions. {For $1 \leq j \leq r$ and $1 \leq i \leq n,$ we define  $\bm{w}_j(i)=(w_{ji,k})_{\{1 \leq k \leq d_j\}} \in \mathbb{R}^{d_j},$ such that \vspace*{-5pt}  
\begin{equation}\label{eq_desigmatrixform}
w_{ji,k}=\varphi_{k} (X_{j,i}). 
\end{equation}
Moreover, we denote $\bm{u}_j(i)=(u_{ji,k})_{\{1 \leq k \leq d_j\}} \in \mathbb{R}^{d_j},$ where  
\begin{equation}\label{eq_ddd}
u_{ji,k}=w_{ji,k} \epsilon_i.
\end{equation}}

According to Lemma \ref{lem_locallystationaryform} of the supplement, we have that for $1 \leq j \leq r,$ there exist measurable functions $\mathbf{W}_j(t, \cdot)=(W_{j1}(t,\cdot), \cdots, W_{jd_j}(t, \cdot)), \ \mathbf{U}_j(t, \cdot)=(U_{j1}(t,\cdot), \cdots, U_{jd_j}(t,\cdot))$ satisfying the stochastic Lipschitz continuity as in (\ref{eq_slc}) such that $\{\bm{u}_j(i)\}$ is mean zero and \vspace*{-5pt}  
\begin{equation}\label{eq_locallystationaryform}
\bm{w}_j(i)=\Wb_j(t_i, \mathcal{F}_i), \ \bm{u}_j(i)=\Ub_j(t_i, \mathcal{F}_i), \ t_i=\frac{i}{n}.
\end{equation} 

%
%

{Recall (\ref{eq_defnp}). With the above notations, we define the block-wise stochastic process $\widetilde{\bW}(t, \cdot) \in \mathbb{R}^p:=(\widetilde{\bW}^\top_1(t, \cdot),\cdots, \widetilde{\bW}^\top_r(t, \cdot))^\top,$ where $\widetilde{\bW}_j(t,\cdot) \in \mathbb{R}^{c_j d_j}, 1 \leq j \leq r,$ is defined as follows {\vspace*{-5pt}  
\begin{equation*}
\widetilde{\bW}_j(t, \cdot):=\bW_j(t, \cdot) \otimes \phib_j(t),
\end{equation*} 
where $\phib_j(t)=(\phi_1(t), \cdots, \phi_{c_j}(t))^\top \in \mathbb{R}^{c_j}.$} Similarly, we define $\widetilde{\Ub}(t, \cdot) \in \mathbb{R}^p:=(\widetilde{\Ub}^\top_1(t, \cdot),\cdots, \widetilde{\Ub}^\top_r(t, \cdot))^\top,$ where $\widetilde{\Ub}_j(t,\cdot) \in \mathbb{R}^{c_j d_j}, 1 \leq j \leq r,$ is defined as follows{\vspace*{-5pt}  
\begin{equation}\label{eq_realU}
\widetilde{\Ub}_j(t, \cdot):=\Ub_j(t, \cdot) \otimes \phib_j(t). 
\end{equation}}

We now define the long-run covariance matrices of $\widetilde{\bW}(t, \cdot)$ and $\widetilde{\Ub}(t, \cdot)$ as follows \vspace*{-5pt}  
\begin{equation}\label{eq_longrunwitht}
\begin{gathered}
\Pi(t):=\sum_{\mathsf{s}=-\infty}^{+\infty} \operatorname{Cov} \left( \widetilde{\Wb}(t, \mathcal{F}_0), \widetilde{\Wb}(t, \mathcal{F}_\mathsf{s}) \right),  \\ \Omega(t):=\sum_{\mathsf{s}=-\infty}^{+\infty} \operatorname{Cov} \left( \widetilde{\Ub}(t, \mathcal{F}_0), \widetilde{\Ub}(t, \mathcal{F}_\mathsf{s}) \right). 
\end{gathered}
\end{equation}
 Using the above notations, we further denote the integrated long-run covariance matrices $\Pi \in \mathbb{R}^{p \times p}$ and $\Omega \in \mathbb{R}^{p \times p}$ as {\vspace*{-5pt}  
\begin{equation}\label{eq_longruncovariancematrix}
\Pi:=\int_0^1 \Pi(t)  \dd t, \ \Omega:=\int_0^1 \Omega(t) \dd t,
\end{equation} }
Moreover, for each $1 \leq j \leq r$ and $\{\vartheta_{j,\ell_2}(t)\}$ defined in (\ref{eq_definitionvartheta}), we denote $\bm{\vartheta}_j(t)=(\vartheta_{j,1}, \cdots, \vartheta_{j,d_j})^\top \in \mathbb{R}^{d_j}$  and 
\begin{equation}\label{eq_defnbmmj}
\bm{f}_j(t)=\phib_j(t) \otimes \bm{\vartheta}_j(t)  \in \mathbb{R}^{c_j d_j}.
\end{equation}
Let $\bm{f}(t)=(\bm{f}^\top_1(t), \cdots, \bm{f}^\top_r(t))^\top \in \mathbb{R}^{p},$ and the $c_0 \times p$ matrix $\Pi_d(t)=\phib_0(t) \otimes \bm{f}^\top(t) \in \mathbb{R}^{c_0 \times p},$ and $\Pi_d:=\int_0^1 \Pi_d(t) \mathrm{d}t.$ Recall (\ref{eq_setting2}). We denote \vspace*{-5pt}  
$$\Omega_0(t):=\sum_{\mathsf{s}=-\infty}^{+\infty} \operatorname{Cov} \left( D(t, \mathcal{F}_0) \otimes \phib_0(t), D(t, \mathcal{F}_\mathsf{s}) \otimes \phib_0(t)  \right),$$ and $\Omega_0:=\int_0^1  \Omega_0(t) \mathrm{d}t \in \mathbb{R}^{c_0 \times c_0}.$ Finally, we denote
\begin{equation}\label{eq_Pibar}
\overline{\Pi}:=
\begin{pmatrix}
\mathbf{I}_{c_0} & \Pi_d \\
\Pi_d^* & \Pi
\end{pmatrix}, \ \   \overline{\Omega}:=
\begin{pmatrix}
\Omega_0 & \mathbf{0} \\
\mathbf{0} & \Omega
\end{pmatrix} \in \mathbb{R}^{\mathsf{p} \times \mathsf{p}}.
\end{equation}
}

In the current paper, we will need the following regularity assumption on $\overline{\Pi}$ and $\overline{\Omega}.$ It is frequently used in the statistics literature to guarantee that both $\overline{\Pi}$ and $\overline{\Omega}$ are invertible, for instance, see \cite{MR3476606,MR3161455,DZ,MR2719856}. 
\begin{assumption}\label{assum_updc} For $\overline{\Pi}$ and $\overline{\Omega}$ defined in (\ref{eq_Pibar}), we assume that there {exists a small universal constant $0<\kappa \leq 1$ such that {for sufficiently large $n$}
\begin{equation*}
 \kappa \leq  \min\{\lambda_\mathsf{p}(\overline{\Pi}), \lambda_\mathsf{p}(\overline{\Omega}) \} \leq \max\{\lambda_1(\overline{\Pi}), \lambda_1(\overline{\Omega}) \} \leq \kappa^{-1},
\end{equation*} 
where $\lambda_\mathsf{p}(\cdot)$ is the smallest eigenvalue of the given matrix and $\lambda_1(\cdot)$ is the largest eigenvalue of the given matrix. }
\end{assumption}

Then we state the results regarding the consistency of our proposed estimator (\ref{eq_proposedestimator}). We define that for $0 \leq j \leq r$ 
\begin{equation}\label{eq_defnxic}
\xi_j:=\sup_{1 \leq i \leq c_j} \sup_{t \in [0,1]} |\phi_i(t)|, \ \gamma_j:=\sup_t |\phib_j(t)|, \ \iota_j:=\sup_x |\mathbf{v}_j(x)|,
\end{equation}
where $\mathbf{v}_j(x):=(\varphi_{\ell}(x))^\top \in \mathbb{R}^{d_j},$ and we recall that $| \phib_j(t) | \equiv | \phib_j(t) |_2$ is the $\ell_2$ norm of $\phib_j(t)=(\phi_1(t), \cdots, \phi_{c_j}(t))^\top \in \mathbb{R}^{c_j},$ and  $| \bm{b} | \equiv | \bm{b} |_2$ is the $\ell_2$ norm of $\bm{b} \in \mathbb{R}^p$ which is the collection of the basis functions. Denote 
\begin{equation}\label{eq_xbasisbound}
\varsigma_j:=\sup_{1 \leq \ell \leq d_j} \sup_{x \in \mathbb{R}} \left( |\varphi_\ell(x)|+|\varphi_\ell'(x)| \right).  
\end{equation}
Based on the above notations, we further denote
\begin{equation}\label{eq_defnxic}
\xi:=\max_{0 \leq j \leq r} \xi_j,\ \varsigma:= \max_{1 \leq j \leq r} \varsigma_j, \ \iota:=\max_{1 \leq j \leq r} \iota_j, \ \gamma=\sup_{0 \leq j \leq r} \gamma_j,  \ \zeta=\sup_{t,x}| \bm{b} |. 
\end{equation}

\begin{theorem}\label{thm_consistency}
Suppose Assumptions \ref{assum_models}-- \ref{assum_updc} and the assumption of (\ref{eq_identiassum}) hold. Moreover, we assume {\vspace*{-5pt}  
\begin{equation}\label{eq_parameterassumption}
\mathsf{p}\left( \frac{\xi^2 \varsigma^2}{\sqrt{n}}+\frac{\xi^2 \varsigma n^{\frac{2}{\tau+1}}}{n}\right)=\oo(1).
\end{equation} }
For $1 \leq j \leq r,$ denote \vspace*{-5pt}  
\begin{equation*}
\mathfrak{F}_j:=\sum_{\ell=1}^{d_j} \left( c_{j,\ell}^{-\mathsf{n}_{j \ell}}+\gamma_\ell \xi_\ell \varsigma_\ell \sqrt{\frac{c_{j, \ell}}{n}} \right)^2.
\end{equation*}
Then we have that for all $1 \leq j \leq r$ and the estimators in (\ref{eq_generalestimateone}) and (\ref{eq_generalestimatetwo}) {\vspace*{-5pt}  
\begin{align}\label{eq_rate}
\sup_{t \in [0,1],x \in \mathbb{R}}\left | m_j(t,x)-\widehat{m}_{j}(t,x) \right|&=\OO_{\mathbb{P}}\Big(\xi \varsigma( \gamma  \iota+\zeta) \sqrt{\frac{\mathsf{p}}{n}}+c_j^{-\mathsf{m}_{1j}}+d_j^{-\mathsf{m}_{2j}}+\gamma \sqrt{\mathfrak{F}_j} \Big), 
\end{align}}
and for the time-varying intercept {\vspace*{-5pt}  
\begin{align}\label{eq_rate2}
\sup_{t \in [0,1]}\left | m_0(t)-\widehat{m}_{0}(t) \right |&=\OO_{\mathbb{P}}\left(\gamma \xi \varsigma \iota \sqrt{\frac{\mathsf{p}}{n}}+c_0^{-\mathsf{m}_0}+ \gamma \sqrt{\sum_{j=1}^r \mathfrak{F}_j } \right). 
\end{align} }
\end{theorem}

{
\begin{remark}\label{rmk_minimaxoptimal}
Theorem \ref{thm_consistency} implies that our proposed sieve least square estimators (\ref{eq_generalestimateone}) and (\ref{eq_generalestimatetwo}) are consistent under mild  conditions. We focus on the discussion for (\ref{eq_generalestimatetwo}). First, the error rate on the right-hand side of (\ref{eq_rate}) contains three parts. {The first part $\xi \varsigma (\gamma \iota+\zeta) \sqrt{\mathsf{p}/n}$ quantifies the error between   (\ref{eq_firststeptruncation1}) and its estimator (\ref{eq_firststeptruncation}),  and the dominant error part between $\chi_j(t)$ and its estimator $\widehat{\chi}_j(t)$. The second part $c_j^{-\mathsf{m}_{1j}}+d_j^{-\mathsf{m}_{2j}}$ quantifies the deterministic error using (\ref{eq_firststeptruncation1}) to approximate the function $m_j(t,x).$} The last part is the possibly minor error part using  (\ref{eq_estimationchi}) to estimate  $\chi_j(t).$

Second, $\xi (\xi_j), \varsigma (\varsigma_j), \iota (\iota_j)$ and $\zeta$ can be calculated for specific sieve basis functions and so does the convergence rate in (\ref{eq_rate}). For example, when $\{\phi_i(t)\}$ are chosen as the Fourier basis functions and $\{\varphi_i(t)\}$ as the mapped Fourier basis functions, then  $\xi_j=\OO(1), \iota_j, \varsigma_j=\OO(d_j)$ and $\zeta=\OO(\sqrt{p}). $ Additionally, if we assume that $\mathsf{m}_{1j}=\mathsf{m}_{2j}=\infty$ such that $m_j(t,x)$ is infinitely differentiable, by choosing $c_j=d_j=\OO(\log n),$ the rate on  the right-hand side of (\ref{eq_rate}) reads as $\OO(n^{-1/2} \log^3 n )$ which matches the optimal uniform (i.e. sup-norm) convergence rate as obtained in \cite{CHEN2015447,DZ, MR673642}. For the magnitudes of $\xi, \varsigma, \iota$ and $\zeta$ for more general sieves, we refer the readers to Section \ref{sec_sieves} of our supplement \cite{suppl}.

Finally, the condition (\ref{eq_parameterassumption}) ensures that $n^{-1} W^\top W$ in the OLS estimator  $\widehat{\bm{\beta}}$ in (\ref{eq_betaolsform}) will converge to $\overline{\Pi}$ in (\ref{eq_Pibar}), which guarantees the regularity behavior of $\widehat{\bm{\beta}}$; see (\ref{eq_consistencyconvergency}) for more details. In fact (\ref{eq_parameterassumption}) can be easily   satisfied. For example, when $\xi=\OO(1), \varsigma=\OO(\log n)$ and $\mathsf{p}=\log^2n,$ we only require $\tau>1$ for the physical dependence measure as in (\ref{eq_physcialrequirementone}).   
\end{remark}
}

\section{Simultaneous inference for the nonlinear regression}\label{sec_inference}
In this section, we conduct simultaneous inference for our model (\ref{eq:model}).  To be concise, we only focus on the nonlinear functions $m_j(t,x), 1 \leq j \leq r,$ based on our proposed sieve least square estimator (\ref{eq_generalestimatetwo}). Similar arguments apply to the time-varying intercept; see Remark \ref{rem_meaninference} of our supplement for more details. In what follows, we assume the error part in Theorem \ref{thm_consistency} satisfies the following mild assumption.
\begin{assumption}\label{assum_debiasassumption} Till the end of the paper, we assume that for some constant  $\varepsilon>\frac{1}{2}$ \vspace*{-5pt}  
\begin{equation}\label{eq_assumptionerrorreduce}
c_0^{-\mathsf{m}_0}+\sum_{j=1}^r \left( c_j^{-\mathsf{m}_{1j}}+d_j^{-\mathsf{m}_{2j}}+\sum_{\ell=1}^{d_j} c_{j, \ell}^{-\mathsf{n}_{j, \ell}} \right)=\OO(n^{-\varepsilon}).  
\end{equation}
Moreover, we assume that minor error parts in (\ref{eq_rate}) and (\ref{eq_rate2}) (i.e., the last terms on their right-hand sides) are much smaller than the dominate parts (i.e., the first terms on their right-hand side).  
\end{assumption}

\begin{remark}
{First,  (\ref{eq_assumptionerrorreduce}) is the commonly used under-smoothing condition which ensures the deterministic approximation error, or the bias part, in Proposition \ref{thm_approximation} is negligible, allowing the focus to be on the analysis of $m_{j,c,d},$ for $1 \leq j \leq r.$ } This approach is commonly used in the literature on the applications of the sieve method to time series; see, for example, \cite{CXH, CHEN2015447, DZ, DZ2}. Moreover, (\ref{eq_assumptionerrorreduce}) is a mild assumption when reasonably large values of $c_j$ and $d_j$ are chosen. For example, if $\mathsf{m}_0=\mathsf{m}_{1j}=\mathsf{m}_{2j}=\infty,$ it can be easily satisfied by choosing $c_j, d_j \asymp \log n.$ In general,  it will be satisfied if we choose 
\begin{equation}\label{eq_cdchoice}
c_0 \asymp n^{C/\mathsf{m}_0}, \ c_j \asymp n^{C/\mathsf{m}_{j1}}, \ d_j \asymp n^{C/\mathsf{m}_{j2}},
\end{equation}
for some constant $C>\frac{1}{2}.$ This shows that we only need a relatively small amount of basis functions once $\mathsf{m}_0, \mathsf{m}_{jk}, k=1,2,$ are reasonably large. 

{Second, the second assumption ensures that the estimation error of each auxiliary mean shift in (\ref{eq_definitionvartheta}) is negligible compared to the estimator of the entire function $m_j(t,x)$. Theoretically, this is a natural assumption since, compared to the 1-D mean shift functions $\vartheta_{j, \ell_2}(t)$, the function $m_j(t,x)$ is a 2-D function, which inherently results in an asymptotically dominating estimation error. Practically, it can be readily addressed by employing relatively smooth basis functions. We illustrate this using (\ref{eq_rate}) as an example. For instance, if we choose orthogonal polynomials as the basis functions, it is straightforward to see that all $\mathsf{n}_{j, \ell_2} = \infty$ in Assumption \ref{assum_baiscorrectionassump}. This implies that one only needs $c_{j, \ell} = \mathrm{O}(\log n)$ to ensure a negligible error. Consequently, the last term can be bounded by $\mathrm{O}(\gamma (\log n)^{3/2} \varsigma \sqrt{d_j/n})$, which is much smaller than the first term. }

Finally, we note that under Assumption \ref{assum_debiasassumption}, as demonstrated in our technical proof, a useful expression can be derived to quantify the error term for $\widehat{m}_{j} - m_j$. Specifically, when $r=1$, by recalling (\ref{eq_defnbmmj}) for $\bm{f}$ and (\ref{eq_proposedestimator}) for $\bm{b}$, we have
\begin{equation*}
\widehat{m}_{1} - m_1 = (\bm{b}^\top - \bm{f}^\top)(\bm{\beta}_1 - \widehat{\bm{\beta}}_1)(1 + \mathrm{o}_{\mathbb{P}}(1)).
\end{equation*}
The adjustment of $\bm{f}$ from the basis functions $\bm{b}$ ensures that our proposed estimator satisfies the identifiability assumption.
\end{remark} 
\subsection{Problems setup and simultaneous confidence region}\label{sec_problemsetup}
For $1 \leq j \leq r$ and {a coverage probability $1-\alpha,$} our goal in this section is to construct a $(1-\alpha)$ simultaneous confidence region (SCR) of $m_j(t,x)$ based on the estimator (\ref{eq_generalestimatetwo}), denoted as $\{\Upsilon_{\alpha}(t,x), 0 \leq t \leq 1, x \in \mathbb{R}\}.$  More concretely, the SCR {aims to have the following property:} 
\begin{equation}\label{eq_scrdefinition}
\lim_{n \rightarrow \infty}\mathbb{P}\left\{ m_j(t,x) \in \Upsilon_\alpha(t,x), \ \forall \ 0 \leq t \leq 1, \ \forall \ x \in \mathbb{R} \right\}=1-\alpha.
\end{equation}

Before introducing various concerned hypothesis testing problems, we first explain how to construct the SCR using the point estimator $\widehat{m}_{j}(t,x)$  in (\ref{eq_generalestimatetwo}). In light of Assumption \ref{assum_debiasassumption} and Proposition \ref{thm_approximation}, for each $1 \leq j \leq r$ and fixed $0 \leq t \leq 1$ and $x \in \mathbb{R},$ we write
\begin{equation}\label{eq_definitiont1}
\mathsf{T}_{j}(t,x)=\widehat{m}_{j}(t,x)-m_j(t,x).
\end{equation}
We now introduce some important quantities. Till the end of the paper, for notional simplicity,  by recalling (\ref{eq_defnbmmj}) for $\bm{f}$ and (\ref{eq_proposedestimator}) for $\bm{b}$, we denote 
\begin{equation}\label{eq_verctorconstruction}
\bm{r} \equiv \bm{r}(t,x):=\bm{b}-\bm{f} \in \mathbb{R}^p,  \ \text{and} \ \ \overline{\bm{r}}=(\mathbf{0}^\top, \bm{r}^\top)^\top \in \mathbb{R}^{\mathsf{p}}. 
\end{equation} 

Recall (\ref{eq_Pibar}) and $\mathsf{I}_j$ in (\ref{eq_proposedestimator}).  Denote $\overline{\mathsf{I}}_j=(\bm{0}^\top, \mathsf{I}_j^\top)^\top \in \mathbb{R}^{\mathsf{p}}$ and 
\begin{equation}\label{eq_defhtx}
h_j(t,x)=\sqrt{\bm{l}_j^\top \overline{\Omega} \bm{l}_j}, \ 1 \leq j \leq r,
\end{equation}
where $\bm{l}_j\equiv \bm{l}_j(t,x):=\overline{\Pi}^{-1} \overline{\bm{r}} \overline{\mathsf{I}}_j.$ {We refer to Remark \ref{rmk_hj(tx)} of our supplement for more discussions on the properties of $h_j(t,x).$ As noted in Remark \ref{rem_meanzerodiscussions} of our supplement, for each fixed pair $(t,x)$, $\sqrt{n}\mathsf{T}_j(t,x)$ can be well approximated by a quadratic form of a high-dimensional locally stationary time series. Furthermore, as shown in Theorem \ref{thm_asymptoticdistribution} below, it is asymptotically Gaussian with mean zero and variance $h_j(t,x)^2$.} Given $\alpha,$ we can utilize
\begin{equation}\label{eq_scrdefinition}
\mathcal{R}_j^{\pm}(t,x):=\widehat m_{j}(t,x)\pm \cf_{\alpha}\frac{h_j(t,x)}{\sqrt{n}},\quad t\in[0,1], x\in\mathbb{R},
\end{equation}
where $\cf_{\alpha} \equiv \cf_{\alpha}(j,n)$ is the critical value  that 
\begin{equation}\label{eq_calpha}
\lim_{n \rightarrow \infty} \p(\sup_{t,x}\left| \frac{\sqrt{n}\mathsf{T}_{j}}{h_j(t,x)}\right|\le \cf_{\alpha})= 1-\alpha. 
\end{equation}
Consequently, for each $1 \leq j \leq r,$ the SCR can be constructed as 
\begin{equation}\label{eq_SCRformaldefinition}
\left\{ \left[\mathcal{R}_j^-(t,x), \ \mathcal{R}_j^+ (t,x)\right], 0 \leq t \leq 1, x \in \mathbb{R} \right\}.
\end{equation}
%
Armed with the SCRs, we proceed to demonstrate a few applications on how to use them to infer the regression functions $m_j(t,x).$ For the first application as in Example \ref{exam_exactform}, we are interested in testing whether the underlying function $m_j(t,x)$ is identical to some pre-given function $\mathfrak{m}_{j,0}(t,x).$ Under the null hypothesis, the pre-given function $\mathfrak{m}_{j,0}(t,x)$ should be embedded into the SCR in (\ref{eq_SCRformaldefinition}). We summarize this in the following example. 

\begin{example}[Testing exact form]\label{exam_exactform} We consider the hypothesis that 
\begin{equation}\label{eq_nullhypotheis}
\mathbf{H}_0: \ m_j(t,x) \equiv \mathfrak{m}_{j,0}(t,x),
\end{equation} 
where $\mathfrak{m}_{j,0}(t,x)$ is some pre-given function.  For the test, we can calculate the SCR as in (\ref{eq_SCRformaldefinition}). If $\mathfrak{m}_{j,0}(t,x)$ can be embedded into the SCR, we shall accept $\mathbf{H}_0$ in (\ref{eq_nullhypotheis}). Otherwise, we need to reject it at level $\alpha$. 
\end{example}

For the second application, we are interested in testing some special structural assumptions of $m_j(t,x)$. In this paper, we present two important examples as in Examples \ref{exam_stationarytest} and \ref{exam_seperabletest} below. One is to test whether $m_j(t,x)$ is independent of the time (c.f. (\ref{eq_hotestinghomogeneity})) and the other is to test whether $m_j(t,x)$ has a multiplicative (i.e., separable) structure in $t$ and $x$ (c.f. (\ref{eq_hotestseparability})). In these applications, under the specific structural assumptions, $m_{j}(t,x)$ can be estimated with faster convergence rates than those using the general approach (\ref{eq_proposedestimator}). For example, in Example \ref{exam_stationarytest}, under the assumption that   $m_j$ is stationary (i.e., independent of time), instead of using $\OO(c_jd_j)$ hierarchical basis functions as in (\ref{eq_basisconstruction}), we only need $\OO(d_j)$ mapped sieve basis functions so that we will have a better rate compared to using all the basis functions as illustrated in (\ref{eq_rate}). {Similarly, assuming that $m_j(t,x)$ is separable in $t$ and $x$, as in Example \ref{exam_seperabletest}, we can estimate the functions solely in $t$ and $x$ separately. This allows us to use only $\OO(c_j + d_j)$ basis functions.}

Based on the above observations, if the functions $m_{j}(t,x)$ under the null hypothesis  can be estimated with faster convergence rates than those estimated using the general approach (\ref{eq_proposedestimator}), it suffices to check whether the estimated  functions under the null hypothesis (with faster convergence rates) can be embedded into the SCR constructed  generally using (\ref{eq_proposedestimator}). We summarize these in the following examples. 
 
\begin{example}[Testing time-homogeneity]\label{exam_stationarytest} We consider the structural assumption  that 
\begin{equation}\label{eq_hotestinghomogeneity}
\mathbf{H}_0: \ m_j(t,x) \equiv m_j(x), \ \ \forall  t \in [0,1],
\end{equation}
for some function $m_j(\cdot).$  Under the null assumption of (\ref{eq_hotestinghomogeneity}), we can estimate the function only using the basis $\{\varphi_j(x)\}$ instead of (\ref{eq_basisconstruction}). Under the assumption of (\ref{eq_assumptionerrorreduce}), by an argument similar to Theorem \ref{thm_consistency}, we can obtain a convergence rate $\iota \xi \zeta (d_j/n)^{1/2}$ which is faster than the rate in (\ref{eq_rate}). 
\end{example}

\begin{example}[Testing separability]\label{exam_seperabletest} We consider the structural assumption that 
\begin{equation}\label{eq_hotestseparability}
\mathbf{H}_0: m_j(t,x)=C_j f_j(t)g_j(x), \ \forall \ t \in [0,1] \ \text{and} \ x \in \mathbb{R}, 
\end{equation}
for some constant $C_j$ and some functions $f_j(\cdot)$ and $g_j(\cdot).$ Without loss of generality,  we can assume that $\int f_j(t) \mathrm{d} t=1$ and  $\int g_j(x) \mathrm{d} x=1,$ so that $C_j=\int \int m_j(t,x) \mathrm{d}x \mathrm{d}t.$ Moreover, we shall have that  $g_j(x)=C_j^{-1}\int m_j(t,x) \mathrm{d} t, \ f_j(t)=C_j^{-1}\int m_j(t,x) \mathrm{d} x.$    {Therefore, one can obtain the estimates of $f_j(t)$ and $g_j(x)$ by integrating the estimator (\ref{eq_generalestimatetwo}). Under the null assumption of (\ref{eq_hotestseparability}), by a discussion similar to Theorem \ref{thm_consistency},}
 we can obtain a convergence rate $\varsigma \gamma(c_j/n)^{1/2}+\varsigma \zeta (d_j/n)^{1/2}$ which is faster than the rate in (\ref{eq_rate}).  
\end{example} 

For illustrations of how to use SCR to conduct the tests presented in Examples \ref{exam_exactform}, \ref{exam_stationarytest}, and \ref{exam_seperabletest}, we refer readers to Figures \ref{fig_exact}, \ref{fig_homogene}, and \ref{fig_sepa} in our supplement \cite{suppl}.


\subsection{Asymptotic theory}
In this subsection, we provide the asymptotic results on the SCR in (\ref{eq_scrdefinition}). To see the validity of (\ref{eq_scrdefinition}), we need to establish the asymptotic normality for the point estimator $\widehat{m}_{j}(t,x)$ and {understand the asymptotic properties of the critical value $\cf_{\alpha}$}. These will be established in Theorems \ref{thm_asymptoticdistribution} and \ref{thm_uniformdistribution}, respectively.  

We start preparing some notations.  
Recall (\ref{eq_xbasisbound}) and (\ref{eq_defnxic}). Recall (\ref{eq_defnp}). Denote $\Theta$ as { 
\begin{align}\label{eq_controlparameter}
\Theta:=\frac{\mathsf{p}^{7/4}}{\sqrt{n}} \hd^3 \fm^2 +\left[ \sqrt{\mathsf{p} n} (\zeta+\gamma \iota) \xi^2 \varsigma \hd^{-(q-1)}\right]^{2/3}&+\left[\sqrt{\mathsf{p}} (\zeta+\gamma \iota) \left( \mathsf{p} \xi^2 \varsigma^2 \fm^{-\tau+1}+ \mathsf{p}n \xi^4 \varsigma^2 \hd^{-(q-2)} \right)  \right]^{2/3} \nonumber \\
&+\left[\mathsf{p}^{3/2} \xi (\zeta+\gamma \iota) \varsigma \fm^{-\tau+1} \right]^{2/3},
\end{align}} 
where $\fm \equiv \fm(n)$ and $\hd \equiv \hd(n)$ are large, diverging values introduced only for technical purposes. Specifically, $\fm$ is used in the $\fm$-dependent approximation, while $\hd$ is employed to properly truncate the time series. Furthermore, one can choose the best $\fm$ and $\hd$ according to (\ref{eq_controlparameter}) so that $\Theta$ can be minimized.} As can be seen in Theorem \ref{thm_gaussianapproximationcase} of our supplement \cite{suppl}, $\Theta$ is used to control the error rates of the point-wise Gaussian approximation. Those parameters are only needed in the theoretical investigations and are not
needed in practical implementation of our methodology. 
As will be seen later, if $\Theta=\oo(1),$ after being properly scaled, $\mathsf{T}_{j}$ in (\ref{eq_definitiont1}) will be asymptotically standard Gaussian; see Theorem \ref{thm_asymptoticdistribution} for a more precise statement and Remark \ref{rem_remaftertheorem41} for more discussions on the parameter $\Theta.$ 
 
\begin{theorem}\label{thm_asymptoticdistribution}
Suppose Assumptions \ref{assum_models}--\ref{assum_debiasassumption} as well as the assumptions of Theorem \ref{thm_consistency} hold.  For $\Theta$ in (\ref{eq_controlparameter}), we assume that 
\begin{equation}\label{eq_onebound}
\Theta=\oo(1). 
\end{equation}
Moreover, we assume that the basis functions are selected such that $|\bm{r}|$, as defined in (\ref{eq_verctorconstruction}), is uniformly bounded from below for all $t \in [0,1]$ and $x \in \mathbb{R}$. {Then we have that when $n$ is sufficiently large, for all $(t,x) \in [0,1] \times \mathbb{R}$ and $\mathsf{T}_j(t,x)$ in (\ref{eq_definitiont1})   
\begin{equation}\label{eq_thmonepartone}
\frac{\sqrt{n} \mathsf{T}_{j}(t,x)}{h_j(t,x)} \simeq \mathcal{N}(0,1),
\end{equation}
where $h_j(t,x)$ is defined in (\ref{eq_defhtx}) and $\simeq$ means convergence in distribution.}    
\end{theorem}

Theorem \ref{thm_asymptoticdistribution} establishes the {point-wise} asymptotic normality of $\mathsf{T}_j$ under the assumption (\ref{eq_onebound}) which can be easily satisfied. In particular, combining with the discussion in Remark \ref{rmk_minimaxoptimal}, it is not hard to see that when $\tau>0$ is large enough so that the temporal relation decays fast enough, the basis functions are chosen so that $\xi, \varsigma=\mathrm{O}(1)$ and $\zeta=\mathrm{O}(1),$ in order to guarantee that $\Theta=\mathrm{o}(1),$ we can allow $\mathsf{p}$ to be as large as $\mathrm{O}(n^{2/7-\delta})$ for some sufficiently small constant $\delta>0.$  This matches with the 
dimension setting for high dimensional convex Gaussian approximation \cite{MR3571252}. Moreover, for many commonly used sieve basis functions—such as Fourier basis functions, orthogonal polynomials, and wavelet basis functions—the quantity $|\bm{r}|$ typically diverges with $\mathsf{p}$, allowing it to be uniformly bounded from below. For further details, see \cite{CHEN2015447}.

{In addition to the pointwise convergence results, we also establish uniform convergence results for the statistics. In particular, to construct the SCRs in (\ref{eq_SCRformaldefinition}), it is essential to understand the asymptotic behavior of $c_\alpha$ in (\ref{eq_calpha}).} Recall (\ref{eq_xbasisbound}), (\ref{eq_defnxic}) and $\mathfrak{m}$ and $\mathrm{h}$ in (\ref{eq_controlparameter}). Denote the control parameter $\Theta^*$ as
\begin{align}\label{eq_thetastar}
\Theta^*:=&n^{-1/2} \mathsf{p}^{7/4} \mathrm{h}^3  \mathfrak{m}^2+\sqrt{\mathsf{p}n} \mathrm{h}^{-2} \varsigma \xi \mathfrak{m}^{-\tau-1/2}+n \mathsf{p}^{-1/2} \xi^2 \varsigma \mathfrak{m}^{-3/2} \mathrm{h}^{-q-1} \nonumber  \\
&+n^{1/2} \mathsf{p}^{-1/2} \mathrm{h}^{-2} \mathfrak{m}^{-3/2} \varsigma^2 \left( \mathsf{p} \xi^2 \fm^{-\tau+1}+ \mathsf{p}n \xi^4 \hd^{-(q-2)} \right).  
\end{align}
{As can be seen in Theorem \ref{thm_uniformconvergencegaussian} of our supplement \cite{suppl}, $\Theta^*$ is used to control the error rates of the uniform Gaussian approximation.}

Recall (\ref{eq_betaolsform}), (\ref{eq_proposedestimator}) and (\ref{eq_longruncovariancematrix}). We further denote 
\begin{equation}\label{eq_ljdefinition}
l_j(t,x):=n^{-1}\overline{\bm{r}}^\top \overline{\mathsf{I}}_j  \overline{\Pi}^{-1}.
\end{equation} 
For $x \in \mathbb{R},$ using the mappings in Definition \ref{defn_mappings}, we write
\begin{equation}\label{eq_ll}
\widetilde{l}_j(t,\widetilde{x}):=l_j(t,g(2\widetilde{x}-1;s)) \equiv l_j(t,x), \ \widetilde{x} \in [0,1], 
\end{equation}  
and $\widetilde{\bm{r}} \equiv \widetilde{\bm{r}}(t, \widetilde{x})=\bm{r}(t, g(2\widetilde{x}-1);s) \equiv \bm{r}(t,x). $
Recall $\bm{\epsilon}=(\epsilon_i)_{1 \leq i \leq n} \in \mathbb{R}^{n}.$ Denote 
\begin{equation}\label{eq_defntdefn}
T_j(t, \widetilde{x}):=\frac{\sqrt{n} \widetilde{l}_j(t,\widetilde{x}) \sqrt{\operatorname{Cov}(W^\top \bm{\epsilon})}}{\widetilde{h}_j(t, \widetilde{x})}, \  \widetilde{h}_j(t,\widetilde{x}):=h_j(t,g(2\widetilde{x}-1);s). 
\end{equation} 
Armed with the above notations, let the manifold $\mathcal{M}_j, 1 \leq j \leq r,$ be defined as follows
\begin{equation}\label{eq_manifolddefinition}
\mathcal{M}_j=:\left\{ T_j(t, \widetilde{x}): (t,\widetilde{x}) \in [0,1] \times [0,1]  \right\}.
\end{equation}

\begin{theorem}\label{thm_uniformdistribution} {Suppose the assumptions of Theorem \ref{thm_asymptoticdistribution} hold with (\ref{eq_onebound}) replaced by $\Theta^*=\oo(1).$} Then critical value $\cf_{\alpha}$ in (\ref{eq_calpha}) satisfies the following expansion
\begin{align}\label{eq_expansionformula}
\alpha=\frac{\cf_{\alpha} \kappa_0}{\sqrt{2} \pi^{3/2}} \exp\left(-\frac{\cf_{\alpha}^2}{2} \right)+\frac{\eta_0}{2 \pi} \exp\left(-\frac{\cf_{\alpha}^2}{2} \right)&+2(1-\Phi(\cf_{\alpha})) +\oo\left( \exp\left(-\frac{\cf^2_{\alpha}}{2} \right)\right),
\end{align}
where $\kappa_0 \equiv \kappa_0(n), \eta_0 \equiv \eta_0(n)$ are the area and the length of the boundary of $\mathcal{M}_j$ in (\ref{eq_manifolddefinition}) and $\Phi(\cdot)$ is the cumulative distribution function of a standard Gaussian random variable. Moreover, suppose $c_j,d_j,$ are chosen according to (\ref{eq_cdchoice}), and there exist some constants $\alpha_1, \alpha_2 \geq 0$ so that
\begin{equation}\label{eq_assumptionbasisderivative}
\int \frac{\| \nabla_t \widetilde{\bm{r}} \|_2}{\| \widetilde{\bm{r}} \|_2} \mathrm{d}t \dd \widetilde{x} \asymp n^{\alpha_1}, \ \int \frac{\| \nabla_{\widetilde{x}} \widetilde{\bm{r}} \|_2}{\| \widetilde{\bm{r}} \|_2} \mathrm{d}t \dd \widetilde{x} \asymp n^{\alpha_2}, \ 
\end{equation}
then we have 
\begin{equation}\label{eq_quantitiesbound}
\cf_{\alpha} \asymp  \log^{1/2} n.
\end{equation}
\end{theorem}

\begin{remark}\label{rem_remaftertheorem41}
Two remarks are in order. First, Theorem \ref{thm_uniformdistribution} provides an accurate representation for the critical value $\cf_{\alpha}$ as in (\ref{eq_expansionformula}). Especially, under the condition of (\ref{eq_assumptionbasisderivative}), we obtain the exact rate of $\cf_\alpha$ in (\ref{eq_quantitiesbound}). Inserting this back into (\ref{eq_expansionformula}), we can see that the error term can be made sufficiently small so that solving (\ref{eq_expansionformula}) will yield an accurate estimation for $\cf_{\alpha}.$ 
We note that the commonly used sieve basis functions can easily satisfy (\ref{eq_assumptionbasisderivative}). For further details, please refer to Assumption 4 in \cite{CHEN2015447}.

Second, based on Theorems \ref{thm_asymptoticdistribution} and \ref{thm_uniformdistribution}, we observe that the approach of using SCRs for various tests, as illustrated in Examples \ref{exam_exactform}--\ref{exam_seperabletest}, can achieve asymptotic power of one, provided the alternative deviates from the null hypothesis at a rate greater than the order of the SCR width, i.e., $\mathrm{O}(n^{-1/2} \cf_{\alpha} h_j(t,x))$. From the definition in (\ref{eq_defhtx}), it is evident that {$h_j(t,x) \asymp \sqrt{c_j d_j}$.} Thus, the proposed tests will achieve asymptotic power of one if the width of the SCRs is much larger than {$\sqrt{c_j d_j \log n/n}$}.

%
%
\end{remark}


%

\subsection{Practical implementation: a multiplier bootstrap procedure}\label{sec_practicalimplementation}

As we can see from Theorem \ref{thm_asymptoticdistribution}, it is difficult to use {(\ref{eq_thmonepartone}) and (\ref{eq_expansionformula})} directly as many of the quantities are unknown and hard to be estimated. For example, in Theorem \ref{thm_asymptoticdistribution}, the variance part $h_j(t,x)$ in (\ref{eq_defhtx}) involves the long-run covariance matrix of $\{\bm{u}_i\},$ i.e., $\Omega(t)$ defined in (\ref{eq_longrunwitht}). For another instance, in  Theorem \ref{thm_uniformdistribution}, the direct use of the results require prior knowledge of the area $\kappa_0$ and the length $\eta_0$ of the boundary of the manifold which also relies on detailed knowledge of the covariance structures of $\{\bm{u}_i\}$.

 To address this issue, in this section, we propose a practical method as in \cite{DZ2,MR3174655} which utilizes high-dimensional multiplier bootstrap statistics to mimic the distributions of $\{\bm{u}_i\}$. We now introduce some notations before presenting the ideas. {We note that since the identifiability of $m_j$'s is not required for estimating the residuals, we can use}\vspace*{-5pt}  
\begin{equation}\label{eq_defnresidual}
\widehat{\epsilon}_i=Y_i-\widehat{m}^*_{0,c}(t_i)-\sum_{j=1}^r \widehat{m}^*_{j,c,d}(t_i, X_{j,i}),
\end{equation}   
where $\widehat{m}^*_{0,c}(\cdot)$ is defined in (\ref{eq_interceptestimation}), and $\widehat{m}^*_{j,c,d}(\cdot,\cdot), 1 \leq j \leq r,$ are defined in (\ref{eq_proposedestimator}). Corresponding to (\ref{eq_ddd}), we set for $1 \leq j \leq r$ 
\begin{equation}\label{eq_xihatdefinition}
\widehat{u}_{ji,k}:=w_{ji,k} \widehat{\epsilon}_i, \ 1 \leq k \leq d_j, 
\end{equation}
and $\widehat{\bm{u}}_j(i) \in \mathbb{R}^{d_j}$ can be defined based on the above notation. For some diverging parameter $m \equiv m(n),$ we set $\widehat{\epsilon}_{i,m}=\sum_{o=i}^{i+m} \widehat{\epsilon}_o, \ \widehat{\bm{u}}_{j,m}(i)=\sum_{o=i}^{i+m} \widehat{\bm{u}}_j(o). $ 

Based on the above quantities, we further define a block-wise vector $\widehat{\Ub}(i,m) \in \mathbb{R}^{\mathsf{p}}:=(\widehat{\epsilon}_{i,m} \otimes \phib_0(t_i), \widehat{\bm{u}}_{1,m}(i) \otimes \phib_1(t_i), \cdots, \widehat{\bm{u}}_{r,m}(i) \otimes \phib_r(t_i)).$ Furthermore, we define
\begin{equation}\label{eq_defnstatisticXi}
\Xi:=\frac{1}{\sqrt{n-m-r} \sqrt{m}} \sum_{i=1}^{n-m} \widehat{\Ub}(i,m) R_i,
\end{equation}
where $R_i, \ r+1 \leq i \leq n-m,$ are i.i.d. $\mathcal{N}(0,1)$ random variables. Note that $\Xi$ can be always calculated once we have the data set and the window-size (i.e., block length) parameter $m$.

Armed with the above notations, we now proceed to state our ideas. Recall (\ref{eq_betaolsform}). First, for (\ref{eq_scrdefinition}), instead of using (\ref{eq_definitiont1}), based on (\ref{eq_defnstatisticXi}), we utilize the following statistic
\begin{equation}\label{eq_defnt1k}
\widehat{\mathsf{T}}_{j}(t,x):=\Xi^\top \widehat{\overline{\Pi}}^{-1} \widehat{\overline{\bm{r}}} \overline{\mathsf{I}}_j, \  \ \widehat{\overline{\Pi}}=\frac{1}{n} W^\top W,\vspace*{-5pt}  
\end{equation}
where $\widehat{\overline{\bm{r}}} $ is defined in the same way as in (\ref{eq_verctorconstruction}) by replacing $\bm{f}$ with its estimator $\widehat{\bm{f}}.$  Before stating the theoretical results, we pause to heuristically discuss the motivation of using  $\widehat{\mathsf{T}}_{j}(t,x)$ to mimic the distribution of  $\mathsf{T}_{j}(t,x)$. For the ease of discussion, we consider that $r=1$ and omit the subscript $j.$ Recall (\ref{eq_realU}). As will be seen in (\ref{eq_fundementalexpression})  of our supplement \cite{suppl}, $\sqrt{n} \mathsf{T}$ can be well approximated  by \vspace*{-5pt}
\begin{equation*}
\bm{z}^\top \overline{\Pi}^{-1}\overline{\bm{r}}, \ \text{where} \ \bm{z}:=\frac{1}{\sqrt{n}} \sum \widetilde{\Ub}(i, \mathcal{F}_i).
\end{equation*}{
According to Theorem \ref{thm_asymptoticdistribution}, the above quantity is asymptotically Gaussian. On the other hand, by the construction of $\widehat{\mathsf{T}}$ (c.f. (\ref{eq_defnt1k})) using (\ref{eq_defnstatisticXi}), when conditional on the data, $\widehat{\mathsf{T}}$ is Gaussian. Therefore, to mimic the distributional behavior of $\sqrt{n} \mathsf{T}$, it suffices to show that the covariance structure of $\sqrt{n} \mathsf{T}$—which depends on the covariance matrix of $\bm{z}$ and the value of $\overline{\Pi}^{-1}$—can be well approximated by that of $\widehat{\mathsf{T}}$. To see this, note first that, as established in (\ref{eq_consistencyconvergency}) of our supplement \cite{suppl}, $\overline{\Pi}^{-1}$ can be well approximated by $\widehat{\overline{\Pi}}^{-1}$, which is incorporated in $\widehat{\mathsf{T}}$ in (\ref{eq_defnt1k}). Moreover, by construction, $\Xi$ in (\ref{eq_defnstatisticXi}) is a multiplier bootstrapped statistic based on estimates of $\bm{z}$ and provides a good approximation to the covariance of $\bm{z}$ conditional on the data. {In particular, we have $\operatorname{Cov}(\Xi|\{(Y_i, X_{j,i})\}) \approx \overline{\Omega}$}, where $\overline{\Omega}$ is defined in (\ref{eq_Pibar}); see Lemmas \ref{lem_stepone}--\ref{lem_residualclosepreparation} of our supplement \cite{suppl} for further discussion.
Consequently, it follows that, conditional on the data, the distribution of $\widehat{\mathsf{T}}$  provides a good approximation to that of $\sqrt{n} \mathsf{T}$. }

{
Combining the above discussion with a uniform Gaussian approximation result in Theorem \ref{thm_uniformconvergencegaussian} of our supplement, we see that (\ref{eq_SCRformaldefinition}) can serve as the SCR for $m_j(t,x)$, where $1 \leq j \leq r$. Despite relying on two unknown quantities, $\cf_\alpha$ and $h_j(t,x)$, these can be efficiently estimated using bootstrap as described in (\ref{eq_defnt1k}). Specifically, given the data, we can generate multiple copies of $\widehat{\mathsf{T}}_j(t,x)$ by sampling multiple copies of $\Xi$ in (\ref{eq_defnstatisticXi}), denoted as ${\widehat{\mathsf{T}}_{j,k}(t,x)}$. Consequently, $h_j(t,x)$ can be estimated using the standard deviation of ${\widehat{\mathsf{T}}_{j,k}(t,x)}$, denoted as $\widehat{h}_j(t,x)$. {That is, generate a large number of $B$  i.i.d. copies of $\{\Xi^{(k)}\}_{k=1}^B$ as in (\ref{eq_defnstatisticXi}). Compute $\widehat{\mathsf{T}}_{j,k}, k=1,2,\cdots, B,$  correspondingly as in (\ref{eq_defnt1k}).  Calculate the sample standard deviation (s.t.d.) of $\{\widehat{\mathsf{T}}_{j,k}\}$ and denote it as $\widehat{h}_j(t,x).$ The critical value $\cf_\alpha$ can then be estimated via its definition in (\ref{eq_calpha}) using ${\widehat{\mathsf{T}}_{j,k}(t,x)}$ and $\widehat{h}_j(t,x),$ denoted as $\widehat{\cf}_\alpha.$  Based on the above results and discussions,  we can use the following Algorithm \ref{alg:boostrapping} to  present the detailed procedure for constructing SCRs using \vspace*{-5pt}
\begin{equation}\label{eq_hahahahahahhaha}
\widehat{\mathcal{R}}_j^{\pm}(t,x):=\widehat m_{j}(t,x)\pm \widehat{\cf}_{\alpha}\frac{\widehat{h}_j(t,x)}{\sqrt{n}},\quad t\in[0,1], x\in\mathbb{R},
\end{equation}

Denote $\Psi(m)$ as \vspace*{-5pt}
\begin{equation}\label{eq_defnpsim}
\Psi(m)=\mathsf{p}\xi^2 \varsigma^2 \left(\frac{1}{m}+\left(\frac{m}{n} \right)^{1-\frac{1}{\tau}}+\sqrt{\frac{m}{n}} \right),
\end{equation}
as well as \vspace*{-5pt}
\begin{align}\label{eq_ratetwo}
\mathfrak{E}:= \Bigg\{ (\gamma  \iota+\zeta) \log \mathsf{p}   \left[ \Psi(m)+\mathsf{p}\left( \frac{\varsigma^2 \xi^2}{\sqrt{n}}+\frac{\varsigma^2 \xi^2 n^{\frac{2}{\tau+1}}}{n}\right)  \right] &+\left[ \mathsf{p}(\gamma  \iota+\zeta)\left( \frac{\varsigma^2 \xi^2}{\sqrt{n}}+\frac{\xi^2 \varsigma^2 n^{\frac{2}{\tau+1}}}{n}\right) \right] \Bigg\}^{2/3} \nonumber \\
&+\left[ \frac{\mathsf{p} \xi^2 \varsigma^2}{\sqrt{n}} \log \mathsf{p} \left[\xi \varsigma( \gamma  \iota+\zeta)^2 \sqrt{\frac{\mathsf{p}}{n}} \right]\right]^{2/3}. 
\end{align}

\begin{theorem}\label{thm_boostrapping}
Suppose the assumptions of Theorem \ref{thm_asymptoticdistribution} hold and $B$ is sufficiently large in the estimation of $\widehat{h}_j(t,x)$. Recall (\ref{eq_thetastar}). Moreover, we assume that 
\begin{align}\label{eq_boothstrappingextraassumption}
\Theta^*+ \mathfrak{E}=\oo(1). 
\end{align}
Then when conditional on the data $\{(Y_i, X_{j,i})\},$
for $1 \leq j \leq r,$ {
\begin{equation*}
\mathbb{P} \left( m_j(t,x) \in [\widehat{\mathcal{R}}_j^-(t,x), \widehat{\mathcal{R}}_j^+(t,x)], \forall t \in [0,1], \ x \in \mathbb{R} \right)=1-\alpha+\mathrm{o}_{\mathbb{P}}(1). 
\end{equation*}}
\end{theorem}}}

{We point out that, as shown in Theorem \ref{thm_uniformconvergencegaussian} of our supplement, \eqref{eq_thetastar} is used to control the rate of the uniform Gaussian approximation for $\sup_{t,x}|\mathsf{T}_j(t,x)|$, while \eqref{eq_ratetwo} is used to bound the error between the bootstrapped statistic and $\sup_{t,x}|\mathsf{T}_j(t,x)|$. As discussed below \eqref{eq_controlparameter} and in Remark \ref{rem_remaftertheorem41}, the condition in \eqref{eq_boothstrappingextraassumption} is readily satisfied.}
Furthermore, the error term (\ref{eq_defnpsim}) is used to control the closeness  between $\operatorname{Cov}(\Xi)$ (c.f. (\ref{eq_defnstatisticXi})) and $\overline{\Omega}$ (c.f. (\ref{eq_Pibar})), {as well as the closeness between $\bm{r}$ (c.f. (\ref{eq_verctorconstruction})) and its estimator}. In fact, when $\tau$ is large enough and the functions are sufficiently smooth, (\ref{eq_boothstrappingextraassumption}) indicates that the optimal choice of $m$ should satisfy that $m \asymp n^{1/3}.$ In Section \ref{sec_parameterchoice} of the supplement \cite{suppl}, we will discuss how to choose the tuning parameter  $m$ practically. {Based on the SCRs, we can further consider the hypothesis testing problems as discussed in Examples \ref{exam_exactform}--\ref{exam_seperabletest} using the strategies discussed thereby.  We point out that the tuning parameters $c_j,d_j$ and $m$ should be chosen before using these algorithms. This will be discussed in Section \ref{sec_parameterchoice} of our supplement \cite{suppl}. Finally, we mention that the algorithm can be implemented using the function $\mathtt{auto.SCR}$ in our $\mathtt{R}$ package $\mathtt{SIMle}.$

 }


\begin{algorithm}[!ht]
\caption{\bf Multiplier bootstrap for constructing SCR}
\label{alg:boostrapping}
\vspace{1pt}
\normalsize
\begin{flushleft}
\noindent{\bf Inputs:} Coverage probability $1-\alpha,$ the function index $1 \leq j \leq r,$  the tuning parameters $c_0, c_j$, $d_j$ and $m$ chosen by the data-driven procedure demonstrated in Section \ref{sec_parameterchoice} of the supplement \cite{suppl}, time series $\{X_{j,i}\}$ and $\{Y_i\},$ and sieve basis functions.
\vspace{2pt}

\noindent{\bf Step one:} Compute $\widehat{\overline{\Pi}}^{-1}$ as in (\ref{eq_defnt1k}), the estimate $\widehat{m}_{j}(t,x)$ as in (\ref{eq_generalestimatetwo}), and the residuals $\{\widehat{\epsilon}_i\}$ according to (\ref{eq_defnresidual}).
\vspace{2pt}


\noindent{\bf Step two:} Construct a sequence of uniform grids of $[0,1] \times [0,1],$ denoted as $(t_i, y_l)$, $1 \leq i$ $\leq C_1$, $1 \leq l \leq C_2,$ where $C_1$ and $C_2$ are some large integers (say $C_1=C_2=2,000$). Using the mappings in Definition \ref{defn_mappings} and calculate $x_l=g(2y_l-1;s).$  
\vspace{2pt}

\noindent{\bf Step three:} For each pair of $(t_i, x_l),$ calculate the associated s.t.d. following the discussions above (\ref{eq_hahahahahahhaha}) and denote them as $\widehat{h}_j(t_i, x_l).$ 
\vspace{2pt}

\noindent{\bf Step four:} For each pair of $(t_i, x_l),$  generate $M$ (say 1,000) i.i.d. copies of  $\{\widehat{\mathsf{T}}_{j,k}(t_i, x_l)\}$, $1 \leq$ $k \leq M.$ For each $k,$ calculate $\mathcal{T}_k$ as follows
\begin{equation*}
\mathcal{T}_k:=\sup_{i,l}\left| \frac{\widehat{\mathsf{T}}_{j,k}(t_i, x_l)}{\widehat{h}_j(t_i, x_l)} \right|.
\end{equation*}
Let $\mathcal{T}_{(1)} \leq \mathcal{T}_{(2)} \leq \cdots \leq \mathcal{T}_{(M)}$ be the order statistics of $\{\mathcal{T}_k\}.$ Calculate $\widehat{\cf}_{\alpha}$ as
\begin{equation*}
\widehat{\cf}_{\alpha}=\mathcal{T}_{\lfloor M(1-\alpha) \rfloor},
\end{equation*}
where $\lfloor x \rfloor$ denotes the largest integer smaller or equal to $x.$
\vspace{2pt}

\noindent{\bf Output:} The SCR can be represented as $\widehat{m}_{j,c,d}(t,x) \pm \frac{\widehat{\cf}_{\alpha}}{\sqrt{n}} \widehat{h}_j(t,x).$
\end{flushleft}

\end{algorithm}

\section{Numerical simulations}\label{sec_numerical}
In this section, we conduct extensive numerical simulations to illustrate the usefulness of our results. All the proposed methods in Sections \ref{sec_estimation} and \ref{sec_inference} can be implemented using our $\mathtt{R}$ package $\mathtt{SIMle}.$ 
{To maintain clarity and accommodate space constraints, we report the results for $r=1$ in (\ref{eq:model}) in this section. In this case, the issue of identifiability does not arise, even without the assumption in (\ref{eq_identiassum}). In Section \ref{simu_addtional_identiyissue} of our supplement, we present the results for $r=2$. }



\subsection{Simulation setup}\label{sec_simulationsettup}
In this subsection, we introduce our simulation setup.  We consider that $r=1$ in (\ref{eq:model}) and focus on the nonlinear AR regression setting as in \cite{MV}. More specifically, we consider the following model, which has garnered significant research interest in the literature \cite{MV,ZWW} \vspace*{-5pt}
\begin{equation}\label{eq_fundementalsimulationmodel}
X_i=m(t_i, X_{i-1})+\sigma(t_i, X_{i-1})\epsilon_i.
\end{equation}
Note that we will let $\epsilon_i$ be a locally stationary and possibly nonlinear time series so that the variance of $\sigma(t_i, X_{i-1})\epsilon_i$ depends on both $X_{i-1}$ and $t_i.$ For the model (\ref{eq_fundementalsimulationmodel}), we will consider various settings as follows.

First, regarding $m(t,x)$ and $\sigma(t,x), t \in [0,1], x\in \mathbb{R},$ for some  $0 \leq \delta \leq 1,$ we will consider the following setups
\begin{enumerate}
\item[(1).] $m(t,x)=(1+x^2)^{-4}+ \delta \sin (2 \pi t) \exp(-x^2)$, and  $$\sigma(t, x)=1.5 \exp(-x^2/2)\left(2+\sin(2 \pi t) \right).$$
\item[(2).] $m(t,x)=( \delta \sin (2 \pi t)+1)\exp(-x^2/2),$ and $$\sigma(t,x)=0.5 \exp(-x^2) \cos(2 \pi t) +1.$$ 
\item[(3).] $m(t,x)=2t( \delta \exp(-2tx^2)+ \pi^{-1/2}\exp(-x^2/2)),$ and 
\begin{equation*}
\sigma(t, x):=
\begin{cases}
0.7(1+x^2), & |x| \leq 1; \\
1.4, &  |x| >1 \ \text{and} \ 0 \leq t<0.5, \\
2, & |x|>1 \ \text{and} \ 0.5 \leq t<1. 
\end{cases}
\end{equation*}
\end{enumerate}
Note that when $\delta=0,$ $m(t,x)$ in the above setups (1)-(3) correspond respectively to the null hypotheses in Examples \ref{exam_exactform}--\ref{exam_seperabletest}.  
Second, for the locally stationary time series $\{\epsilon_i\}$, we consider three different settings as in Section \ref{sec_simulationsetupepsilon} of our supplement. 

%
\subsection{Simultaneous confidence regions}
In this subsection, we examine the  performance of our proposed sieve estimators and their associated SCRs constructed from Algorithm \ref{alg:boostrapping} using the simulated coverage probabilities. For models (1)--(3) in Section \ref{sec_simulationsettup}, we choose $\delta=1.$ We also compare our methods with the kernel based estimation methods \cite{MV, ZWW}. Since model (2) is separable, we utilize the method in \cite{CSW} which shows better performance for separable functions. Moreover, since the support of the time series can span over $\mathbb{R}$, to facilitate the comparison with the kernel methods, we focus on the region that $(t,x) \in [0,1] \times [-10, 10].$ For the kernel method, we use the Epanechnikov kernel and the cross-validation approach as in \cite{10.1214/18-AOS1743} to select the bandwidth. For our sieves method, the parameters are chosen according to Section \ref{sec_parameterchoice} of our supplement \cite{suppl}. 

The results are reported in Table \ref{table_cp} of the supplement. {We conclude that our estimators achieve reasonably high accuracy and outperform kernel estimators across all commonly used sieve basis functions, as kernel methods may suffer from boundary issues for both $t$ and $x$. } Especially, our estimators have already had a good performance even the sample size is moderate when $n=500.$ Finally, in practice, based on our simulations, we recommend to use orthogonal wavelet basis functions.

\subsection{Accuracy and power of SCR based tests}
In this subsection, we examine the performance of
Algorithm \ref{alg:boostrapping} when used to test the hypotheses in Examples \ref{exam_exactform}--\ref{exam_seperabletest}, i.e., (\ref{eq_nullhypotheis}), (\ref{eq_hotestinghomogeneity}) and (\ref{eq_hotestseparability}). In terms of the three models (1)-(3) in Section \ref{sec_simulationsettup}, the null hypotheses can all be formulated as 
\begin{equation}\label{eq_equaivalent}
\mathbf{H}_0: \delta=0 \ \  \text{Vs}  \ \ \mathbf{H}_a: \delta>0.
\end{equation}

In Table \ref{table_typeoneerror} of the supplement, we report the simulated type one error rates of the above three tests under the null that $\delta=0.$ For the ease of statements, we call the above three tests as exam form, stationarity and separability tests, respectively. We find that our Algorithm \ref{alg:boostrapping} is reasonably accurate for all these tests.  Then we examine the power of our methodologies as $\delta$ increases away from zero. In Figure \ref{fig_selfcomparison} of the supplement, we report the results for all the three tests using our multiplier bootstrap method based on Daubechies-9 basis functions. It can be concluded that once $\delta$ deviates away from $0$ a little bit, our method will be able to reject the null hypothesis.  Moreover, in \cite{HHY}, the authors proposed a weighted $L_2$-distance test statistic to test the separability hypothesis for $m(t,x)$ defined on compact domains. In Figure \ref{fig_jasacomparison} of the supplement, we compare our multiplier bootstrap method with the  weighted $L_2$-distance based method. It can be seen that our method has better performance when $\delta$ is small, i.e., weak alternatives.



\section{Real data analysis}\label{sec_realdata}
In this section, we apply our proposed methodologies to analyze the circulatory and respiratory data from Hong Kong \cite{cai2000efficient, fan1999statistical,MR1804172, MR2758526}. Additionally, in Section \ref{supp_additionalrealdata} of our supplement, we also analyze a financial dataset to further demonstrate the effectiveness of our methods.

The Hong Kong circulatory and respiratory data consists of daily hospital admissions and daily measurements of pollutants sulfur dioxide ($\operatorname{SO}_2$), nitrogen dioxide ($\operatorname{NO}_2$), and dust (RSP). The data was collected between January 1st, 1994 and December 31st, 1995, with the sample size $n=730.$ Such a dataset has been analyzed under the additive time-varying linear models that \vspace*{-5pt} 
\begin{equation}\label{eq_linearmodel}
Y_i=\beta_0(i/n)+\sum_{j=1}^3 \beta_j(i/n) X_{j,i}+\epsilon_i, \ i=1,2,\cdots,n, \vspace*{-5pt}
\end{equation}
to examine the extent to which the association varies over time {assuming $\E(\epsilon_{i}|\{X_{j,i}\}_{1 \leq j \leq r})=0.$} In particular, \cite{cai2000efficient, fan1999statistical,MR1804172} analyzed such a dataset assuming i.i.d. observations and \cite{MR2758526} assumed locally stationary observations. Nevertheless, (\ref{eq_linearmodel}) is a strong assumption and a special case of our proposed model (\ref{eq:model}). 

In what follows, we shall investigate the data under our general model (\ref{eq:model}), and we shall also compare our findings with the previous results. More specifically, we consider the model that \vspace*{-5pt}
\begin{equation}\label{eq_modelapplication}
Y_i=m_0(t)+\sum_{j=1}^3 m_j(X_{j,i}, i/n)+\epsilon_i, \ 1 \leq i \leq n, \vspace*{-5pt}
\end{equation} 
where $\{Y_i\}$ is the time series of the daily total number of hospital admissions of circulation and respiration, and $\{X_{j,i}\}, j=1,2,3,$ represent the time series of daily levels of $\operatorname{SO}_2, \operatorname{NO}_2$ and dust, respectively. 

Inspired by \cite{MR1804172,MR2758526}, for the response, we center $\{Y_i\}$ by its average; for the three pollutants, we center $\{X_{j,i}\}$ by their averages and normalized by their standard deviations. We point out that (\ref{eq_linearmodel}) is actually a special case of the general model (\ref{eq_modelapplication}) and of separable form. We now test whether (\ref{eq_linearmodel}) is reasonable using our Algorithm \ref{alg:boostrapping}.

First,  we test whether each $m_j, j=1,2,3,$ in (\ref{eq_modelapplication}) is time-invariant using the approach as discussed in Example       \ref{exam_stationarytest}. 
It turns out that the null hypotheses are rejected for all $j = 1,2,3,$ {(with $p$-values being $0.01, 0.01, 0.03$ under the $95\%$ coverage probability)}. For an illustration, we refer the readers to Figure \ref{fig_timevaryingrealdata} of our supplement. This shows that all three pollutants' effect varies significantly with time and  we need to apply a time-varying model. 


Second, we test whether each $m_j, j=1,2,3,$ in (\ref{eq_modelapplication}) has a separable form using the approach as discussed in Example       \ref{exam_seperabletest}. It turns out that the null hypotheses are rejected again for all $j=1,2,3$ {(with $p$-values being $0.02, 0.03, 0.005$ under the $95\%$ coverage probability)}. For an illustration, we refer the readers to Figure  \ref{fig_separablerealdata} of our supplement. This indicates that the model in (\ref{eq_linearmodel}) {is not suitable for modeling the conditional expectation of $Y_i$ based on $\{X_{j,i}\}_{1 \leq j \leq r}.$} 

%
%
%

Then we estimate all functions $m_j(t,x), j=1,2,3,$ using our method with the general model (\ref{eq_modelapplication}). For illustrations, please refer to Figure \ref{fig_ourestimation} of our supplement.
Several key insights can be drawn compared to the results obtained in \cite{cai2000efficient, fan1999statistical, MR1804172, MR2758526}. First, our SCRs are generally narrower for all three covariates, especially near the boundaries, leading to more reliable outcomes. {This is primarily because the sieve method adapts to the smoothness of the underlying functions, potentially achieving better convergence rates}. Second, we observe that overall, {the coefficient functions of} \(\operatorname{NO}_2\) and dust  exhibit greater fluctuations than that of \(\operatorname{SO}_2\), suggesting that they have a more complex impact on daily hospital admissions.




\begin{acks}[Acknowledgments]
The authors are grateful to Jianqing Fan and Wenyang Zhang for providing the circulatory and respiratory dataset of Hong Kong, and  Likai Chen and Ekaterina Smetanina for providing resources of the financial dataset. The first author is partially supported by NSF DMS-2515104. 
\end{acks}


\begin{center}
{\bf Supplement to "Simultaneous sieve estimation and inference for time-varying nonlinear time series regression" }
\end{center}

\numberwithin{equation}{section}
\renewcommand{\thesection}{\Alph{section}}
\renewcommand{\thefigure}{\Alph{figure}}
\renewcommand{\thetable}{\Alph{table}}

\section{A brief summary of sieve spaces approximation theory}\label{sec_sieves}
In this section, we give a brief overview on the sieve approximation theory \cite{CXH} for compact domain and the mapped sieve approximation theory for the unbounded domain \cite{MR1874071,MR2486453,unboundeddomain}.  We start with introducing some commonly used sieve basis functions.

\begin{example}[Trigonometric and mapped trigonometric polynomials]\label{examplebasis_fourier} For $x \in [0,1],$ we consider the following trigonometric polynomials 
\begin{equation*}
\left\{1, \sqrt{2} \cos(2 k \pi x), \sqrt{2} \sin (2k\pi x), \cdots \right\}, \ k \in \mathbb{N}.
\end{equation*}
The above basis functions form an orthonormal basis 
of $L_2([0,1]).$ Note that the classic trigonometric basis function is well suited for approximating periodic functions on $[0,1].$ 

When $x \in \mathbb{R},$ let $u(x): \mathbb{R} \rightarrow [0,1]$ be constructed following Example \ref{example_mappings}. Since $u(x)$ in Example \ref{example_mappings} satisfies that $u'(x) \in L_2(\mathbb{R})$, the following mapped trigonometric polynomials form an orthonomal basis for $L_2(\mathbb{R})$ 
\begin{equation*}
\left\{ \sqrt{u'(x)}, \ \sqrt{2 u'(x) } \cos(2 k \pi u(x)), \ \sqrt{2 u'(x)} \sin(2 k \pi u(x)),\cdots \right\}, \ k \in \mathbb{N}, 
\end{equation*}
where we used inverse function theorem. 
\end{example}

\begin{example}[Jacobi and mapped Jacobi polynomials]\label{exam_jacobibasis} For $I=(-1,1)$ and some constants $\alpha, \beta>-1,$ denote the Jacobi weight function as
\begin{equation*}
\omega^{\alpha,\beta}(y)=(1-y)^{\alpha}(1+y)^{\beta},
\end{equation*}
and associated weighted $\mathcal{L}^2$ space as $\mathcal{L}^2_{\omega^{\alpha,\beta}}(I).$ With the above notations, the Jacobi polynomials are the orthogonal polynomials $\{J_n^{\alpha,\beta}(y)\}$ such that 
\begin{equation*}
\int_I J_n^{\alpha, \beta}(y) J_m^{\alpha,\beta}(y) \omega^{\alpha,\beta}(y) \dd y=\gamma_n^{\alpha,\beta} \delta_{n,m}, 
\end{equation*}
where $\delta_{n,m}$ is the Kronecker function, and 
\begin{equation*}
\gamma_n^{\alpha,\beta}=\frac{2^{\alpha+\beta+1} \Gamma(n+\alpha+1) \Gamma(n+\beta+1)}{(2n+\alpha+\beta+1) \Gamma(n+1) \Gamma(n+\alpha+\beta+1)}.
\end{equation*}
More explicitly, the Jacobi polynomials can be characterized using the  Rodrigues' formula \cite{MR0000077} 
\begin{equation*}
J_n^{\alpha,\beta}(x)=\frac{(-1)^n }{2^n n!}(1-x)^{-\alpha}(1+x)^{-\beta} \frac{\dd^n }{\dd x^n} \left( (1-x)^\alpha(1+x)^\beta(1-x^2)^n \right).
\end{equation*}
Then for $x \in [0,1],$ the following sequence forms an orthonormal basis for $L_2([0,1])$ 
\begin{equation*}
\left\{ \mathsf{J}^{\alpha,\beta}_k(x):= \sqrt{\frac{2 \omega^{\alpha, \beta}(2x-1)}{\gamma_k^{\alpha,\beta}} } J_k^{\alpha,\beta}(2x-1)  \right\}, \ k \in \mathbb{N}. 
\end{equation*}
We mention that when $\alpha=\beta=0,$ the Jacobi polynomial reduces to the Legendre polynomial and when $\alpha=\beta=-0.5,$ the Jacobi polynomial reduces to the Chebyshev polynomial of the first kind. 

When $x \in \mathbb{R},$ let $u(x)$ be the mapping as in Example \ref{examplebasis_fourier}. Then the following sequence provides an orthonormal basis for $L_2(\mathbb{R})$ (see \cite[Section 2.3]{unboundeddomain}) 
\begin{equation*}
\left\{ \sqrt{u'(x)}\mathsf{J}_k^{\alpha,\beta}(u(x)) \right\}.
\end{equation*}


\end{example}

%

\begin{example}[Wavelet and mapped Wavelet basis]\label{examplebasis_wave}
For $N \in \mathbb{N},$ a Daubechies (mother) wavelet of class $D-N$ is a function $\psi \in L_2(\mathbb{R})$ defined by 
\begin{equation*}
\psi(x):=\sqrt{2} \sum_{k=1}^{2N-1} (-1)^k h_{2N-1-k} \varphi(2x-k),
\end{equation*}
where $h_0,h_1,\cdots,h_{2N-1} \in \mathbb{R}$ are the constant (high pass) filter coefficients satisfying the conditions
$\sum_{k=0}^{N-1} h_{2k}=\frac{1}{\sqrt{2}}=\sum_{k=0}^{N-1} h_{2k+1},$
as well as, for $l=0,1,\cdots,N-1$
\begin{equation*}
\sum_{k=2l}^{2N-1+2l} h_k h_{k-2l}=
\begin{cases}
1, & l =0 ,\\
0, & l \neq 0.
\end{cases}
\end{equation*} 
And $\varphi(x)$ is the scaling (father) wavelet function is supported on $[0,2N-1)$ and satisfies the recursion equation
$\varphi(x)=\sqrt{2} \sum_{k=0}^{2N-1} h_k \varphi(2x-k),$ as well as the normalization $\int_{\mathbb{R}} \varphi(x) dx=1$ and
$ \int_{\mathbb{R}} \varphi(2x-k) \varphi(2x-l)dx=0, \ k \neq l.$
Note that the filter coefficients can be efficiently computed as listed in \cite{ID92}.  The order $N$, on the one hand, decides the support of our wavelet; on the other hand, provides the regularity condition in the sense that
\begin{equation*}
\int_{\mathbb{R}} x^j \psi(x)dx=0, \ j=0,\cdots,N, \  \text{where} \ N \geq d. 
\end{equation*}
We will employ Daubechies wavelet with a sufficiently high order when forecasting in our simulations and data analysis. In the present paper, to construct a sequence of orthogonal wavelet, we will follow the dyadic construction of \cite{ID98}. For a given $J_n$ and $J_0,$ we will consider the following periodized wavelets on $[0,1]$
\begin{equation}\label{eq_constructone}
\Big\{ \varphi_{J_0 k}(x), \ 0 \leq k \leq 2^{J_0}-1; 
\psi_{jk}(x), \ J_0 \leq j \leq J_n-1, 0 \leq k \leq 2^{j}-1  \Big\},\ \mbox{ where}
\end{equation} 
\begin{equation*}
\varphi_{J_0 k}(x)=2^{J_0/2} \sum_{l \in \mathbb{Z}} \varphi(2^{J_0}x+2^{J_0}l-k) , \
\psi_{j k}(x)=2^{j/2} \sum_{l \in \mathbb{Z}} \psi(2^{j}x+2^{j}l-k),
\end{equation*}
or, equivalently \cite{MR1085487}
\begin{equation}\label{eq_meyerorthogonal}
\Big\{ \varphi_{J_n k}(x), \  0 \leq k \leq  2^{J_n-1} \Big\}.
\end{equation}

When $x \in \mathbb{R},$ let $u(x)$ be the mapping as in Example \ref{examplebasis_fourier}. Then using (\ref{eq_meyerorthogonal}), the following sequence provides an orthonormal basis for $L_2(\mathbb{R})$ 
\begin{equation*}
\left\{ \sqrt{u'(x)}\varphi_{J_n k}(u(x)) \right\}.
\end{equation*}
Similarly, we can construct the mapped wavelet functions using  (\ref{eq_constructone}). 

\end{example}


Next, we collect some important properties of the above basis functions and some useful results on the approximation errors. 

\begin{lemma}\label{lem_deterministicapproximation}
Suppose Assumption \ref{assum_smoothnessasumption} holds. Then for any fixed $x \in \mathbb{R},$  denote 
\begin{equation*}
m_{j,c}(t,x)=\sum_{j=1}^c a_j \phi_j(t), \ 1 \leq j \leq r, 
\end{equation*}
where we assume that $m_j(t,x)=\sum_{j=1}^{\infty} a_j \phi_j(t), \ a_j \equiv a_j(x).$ Then for the basis functions in Examples \ref{examplebasis_fourier}--\ref{examplebasis_wave},  we have that
\begin{equation*}
\sup_t |m_j(t,x)-m_{j,c}(t,x) |=\OO(c^{-\mathsf{m}_1}).
\end{equation*}
\end{lemma}
\begin{proof}
See Section 2.3.1 of \cite{CXH}. 
\end{proof}

\begin{lemma}\label{lem_deterministicapproximation2}
Suppose Assumption \ref{assum_smoothnessasumption} holds. Then for any fixed $t \in [0,1],$  denote 
\begin{equation*}
m_{j,d}(t,x)=\sum_{j=1}^d b_j \varphi_j(x),
\end{equation*}
where we assume that $m_j(t,x)=\sum_{j=1}^{\infty} b_j \varphi_j(x), \ b_j \equiv b_j(t).$ Then for the mapped basis functions in Examples \ref{examplebasis_fourier}--\ref{examplebasis_wave},  we have that
\begin{equation*}
\sup_{x \in \mathbb{R}} |m_{j}(t,x)-m_{j,d}(t,x) |=\OO(d^{-\mathsf{m}_2}).
\end{equation*}
\end{lemma}
\begin{proof}
Recall (\ref{eq_transformedmjty}). Since 
\begin{equation*}
\sup_{x \in \mathbb{R}} |m_{j}(t,x)-m_{j,d}(t,x) |=\sup_{y \in [0,1]}|\widetilde{m}_{j}(t,x)-\widetilde{m}_{j,d}(t,x)|,
\end{equation*}
the proof follows from Section 2.3.1 of \cite{CXH}. Or see Section 6.2 of \cite{MR1176949}. 
\end{proof}

In the following lemma, we provide some controls for the quantities $\xi_c$ and $\zeta$ defined in (\ref{eq_defnxic}).  
\begin{lemma}\label{lem_basisl2norm}
For the basis functions in Example \ref{examplebasis_fourier}, we have that $\xi_c=\OO(1)$ and $\zeta=\OO(\sqrt{p});$ for the basis functions in Example \ref{exam_jacobibasis}, we have that  $\xi_c=\OO(1)$ and $\zeta=\OO(p);$ finally, for the basis functions in Example \ref{examplebasis_wave}, we have that  $\xi_c=\OO(\sqrt{c})$ and $\zeta=\OO(\sqrt{p}).$ 
\end{lemma}
\begin{proof}
See \cite{DZ} or Section 3 of  \cite{MR3343791}. 
\end{proof}

Finally, for an illustration, in Figure \ref{fig_basis}, we provide a plot of the Legendre and the mapped Legendre basis functions. In our $\mathtt{R}$ package $\mathtt{SIMle},$ one can use $\mathtt{bs.gene}$ to generate many commonly used basis functions and  $\mathtt{bs.gene.trans}$ for the mapped basis functions.  

\begin{figure}[!ht]
\hspace*{-1cm}
\begin{subfigure}{0.4\textwidth}
\includegraphics[width=6.4cm,height=5.5cm]{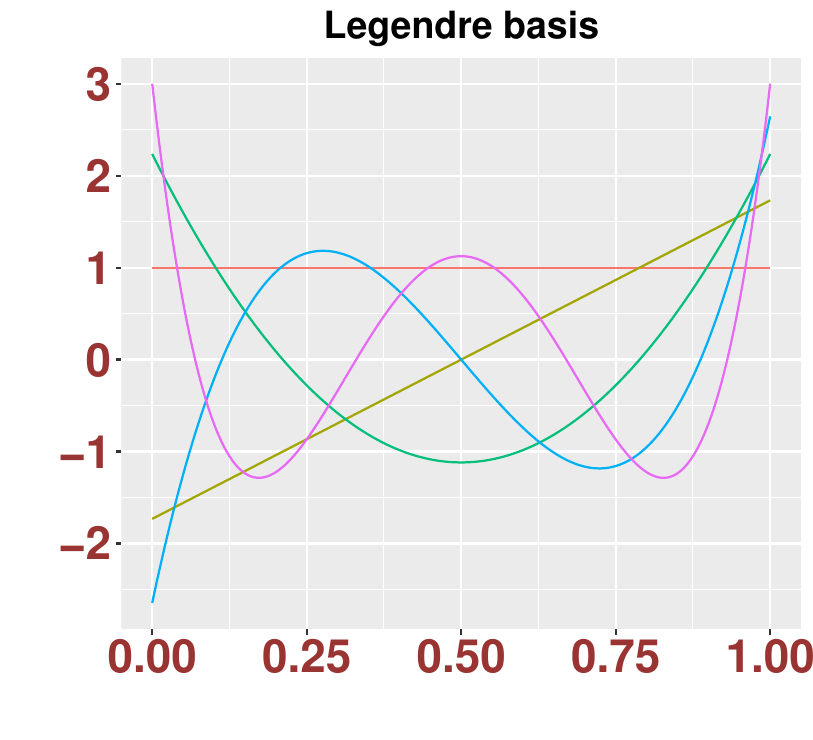}
\end{subfigure}
\hspace{0.7cm}
\begin{subfigure}{0.4\textwidth}
\includegraphics[width=6.4cm,height=5.5cm]{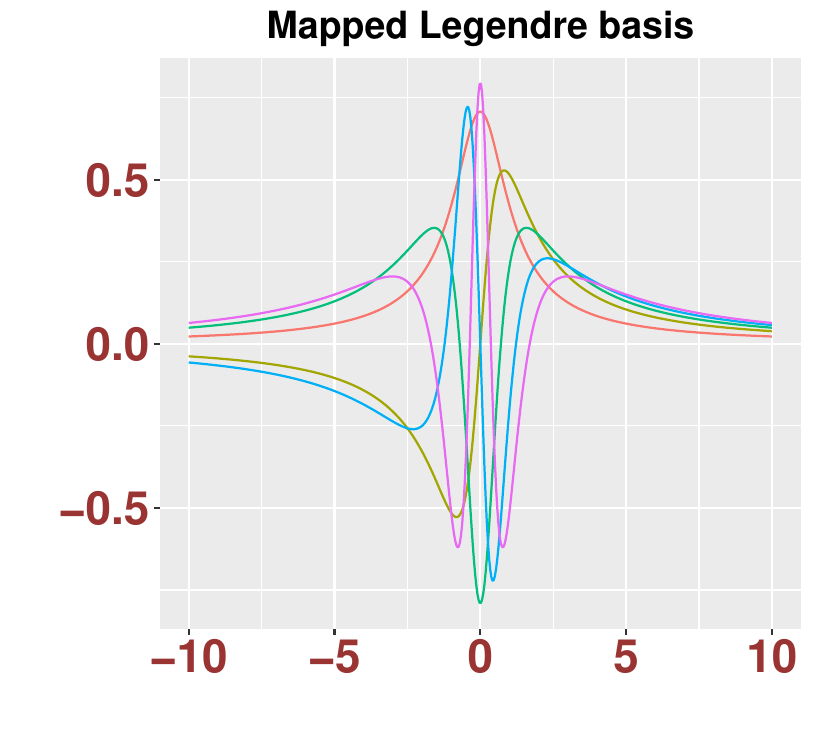}
\end{subfigure}
\vspace*{-0.4cm}
\caption{{ \footnotesize The first five normalized Legendre basis and mapped Legendre basis functions. We refer the readers to Example \ref{exam_jacobibasis} of our supplement \cite{suppl} for precise definitions. The left panel corresponds to the Legendre basis functions which can be generated using $\mathtt{bs.gene}$ from our $\mathtt{R}$ package  $\mathtt{SIMle}$. The right panel corresponds to the mapped Legendre basis using the mapping (\ref{eq_y}) with $s=1$ which can be generayed using $\mathtt{bs.gene.trans}.$ The plots can be made using the functions $\mathtt{bs.plot}$ and $\mathtt{bs.plot.trans}$ in our package.}  }
\label{fig_basis}
\end{figure}


\section{Detailed discussion for the OLS estimation for (\ref{eq_betaolsform})}\label{sec_detailedols}
Let the vector $\bm{\beta}_1=(\beta_1, \cdots, \beta_{p})^\top \in \mathbb{R}^{p}$ collect all these coefficients $\{\beta_{j, \ell_1, \ell_2}\}$ in the order of the indices $\Iintv{1,r} \times \Iintv{1,c_j} \times \Iintv{1,d_j}.$ That is, for $1 \leq k \leq p,$ and its associated $1 \leq  \ell \leq r$ that $\sum_{j=1}^\ell c_j d_j<k \leq \sum_{j=1}^{\ell+1} c_j d_j,$ we have  
\begin{equation}\label{eq_indices}
\beta_k=\beta_{l+1, \ell_1, \ell_2+1},  \ \ell_2=\floor*{(k-\sum_{j=1}^\ell c_j d_j)/c_{\ell+1}}, \ \ell_1=k-\sum_{j=1}^\ell c_j d_j-(\ell_2+1) d_{\ell+1}, 
\end{equation} 
where  $\floor*{\cdot}$ is the floor of a given real value. {Moreover, we denote $\bm{\beta}_0=(\beta_{0,1}, \cdots, \beta_{0,c})^\top \in \mathbb{R}^{c_0}$ and $\bm{\beta}=(\bm{\beta}_0^\top, \bm{\beta}_1^\top)^\top \in \mathbb{R}^{\mathsf{p}}.$ }
Recall the basis function $b_{\ell_1, \ell_2}(t,x)$ has the form of (\ref{eq_basisconstruction}). Denote $W_1 \in \mathbb{R}^{n \times p}$ as the matrix of the linear regression (\ref{eq_linearregression}) whose entry satisfies that
\begin{equation}\label{eq_designmatrix}
W_{1,ik}=\phi_{\ell_1}(t_i) \varphi_{\ell_2+1}(X_{l+1,i}),
\end{equation}
where $\ell_1$ and $\ell_2$ are defined in (\ref{eq_indices}). {In addition, denote $W_0 \in \mathbb{R}^{n \times c_0}$ whose $i$th row is $(\phi_1(t_i), \cdots, \phi_{c_0}(t_i)) \in \mathbb{R}^{c_0}.$ Then we construct the design matrix $W=(W_0^\top,W_1^\top)^\top \in \mathbb{R}^{n \times \mathsf{p}}.$} 

\section{Parameter selection}\label{sec_parameterchoice}
In this subsection, we discuss how to choose the important parameters in practice using the data driven approach as in \cite{bishop2013pattern}. For simplicity, we assume $m_0(t) \equiv 0$ and $r=1.$ In this setting, we only need to choose one pair of $(c,d).$ The general case can be handled similarly.     

We first discuss how to choose $c$ and $d$ using cross validation. For a given integer $l,$ say $l=\lfloor 3 \log_2 n \rfloor,$ we divide the time series into two parts: the training part $\{X_i\}_{i=1}^{n-l}$ and the validation part $\{X_i\}_{i=n-l+1}^n.$  With some preliminary initial pair $(c,d)$, we propose a sequence of candidate pairs  $(c_i, d_j), \ i=1,2,\cdots, u, \ j=1,2,\cdots, v,$ in an appropriate neighborhood of $(c,d)$ where $u, v$ are some given large integers. For each pair of the choices $(c_i, d_j),$ we estimate $\widehat{m}_{j,c,d}(t,x)$ as in (\ref{eq_proposedestimator}).  Then using the fitted model, we forecast the time series in the validation part of the time series.  Let $\widehat X_{n-l+1,ij}, \cdots, \widehat X_{n,ij}$ be the forecast of $X_{n-l+1},..., X_n,$ respectively using the parameter pair $(c_i, d_j)$. Then we choose the pair $(c_{i_0},d_{j_0})$ with the minimum sample MSE of forecast, i.e.,
 \begin{equation*} 
({i_0},{j_0}):= \argmin_{((i,j): 1 \leq i \leq u, 1 \leq j \leq v)} \frac{1}{l}\sum_{k=n-l+1}^n (X_k-\widehat X_{k,ij})^2.
 \end{equation*}

Then we discuss how to choose $m$ for practical implementation. In \cite{MR3174655}, the author used the minimum volatility (MV) method to choose the window size $m$ for the scalar covariance function. The MV method does not depend on the specific form of the underlying time series dependence structure and hence is robust to misspecification of
the latter structure \cite{politis1999subsampling}. The MV method utilizes the fact that the covariance structure of $\widehat{\Omega}$ becomes stable when the
block size $m$ is in an appropriate range, where $\widehat{\Omega}=E[\Phi\Phi^*|(X_1,\cdots,X_n)]=$ is defined as 
{ 
\begin{equation}\label{eq_widehatomega}
\widehat{\Omega}:=\frac{1}{(n-m-r+1)m} \sum_{i=b+1}^{n-m} \Big[ \Big(\sum_{j=i}^{i+m} \widehat{\bm{x}}_i \Big) \otimes \Big( \phib(\frac{i}{n}) \Big) \Big] \times \Big[ \Big(\sum_{j=i}^{i+m} \widehat{\bm{x}}_i \Big) \otimes \Big( \phib(\frac{i}{n}) \Big) \Big]^\top.
\end{equation}
}    Therefore, it desires to minimize the standard errors of the latter covariance structure in a suitable range of candidate $m$'s.

In detail, for a give large value $m_{n_0}$ and a neighbourhood control parameter $h_0>0,$  we can choose a sequence of window sizes $m_{-h_0+1}<\cdots<m_1< m_2<\cdots<m_{n_0}<\cdots<m_{n_0+h_0}$  and obtain $\widehat{\Omega}_{m_j}$ by replacing $m$ with $m_j$ in (\ref{eq_widehatomega}), $j=-h_0+1,2, \cdots, n_0+h_0.$ For each $m_j, j=1,2,\cdots, m_{n_0},$ we calculate the matrix norm error of $\widehat{\Omega}_{m_j}$ in the $h_0$-neighborhood, i.e., 
\begin{equation*}
\mathsf{se}(m_j):=\mathsf{se}(\{ \widehat{\Omega}_{m_{j+k}}\}_{k=-h_0}^{h_0})=\left[\frac{1}{2h_0} \sum_{k=-h_0}^{h_0} \| \overline{\widehat{\Omega}}_{m_j}-\widehat{\Omega}_{m_j+k} \|^2 \right]^{1/2},
\end{equation*}
where $\overline{\widehat{\Omega}}_{m_j}=\sum_{k=-h_0}^{h_0} \widehat{\Omega}_{m_j+k} /(2h_0+1).$
Therefore, we choose the estimate of $m$ using 
\begin{equation*}
\widehat{m}:=\argmin_{m_1 \leq m \leq m_{n_0}} \mathsf{se}(m).
\end{equation*}
Note that in \cite{MR3174655} the author used $h_0=3$ and we also use this choice in the current paper.

\section{Some additional examples and remarks}\label{appendix_additionalremark}

In this section, we provide some additional  examples and remarks. 

\subsection{Examples of locally stationary time series}\label{appendix_additionalremark1}
In this subsection, we present two classes of locally stationary time series that satisfy the conditions outlined in Section \ref{sec_modelassumption}.

\begin{example}[Locally stationary linear time series]\label{exam_linear} Let $\{\epsilon_i\}$ be i.i.d. random variables and $\{a_j(t)\}$ be continuously differentiable functions on $[0,1].$ Consider the locally stationary linear process such that 
\begin{equation*}
G(t, \mathcal{F}_i)=\sum_{k=0}^{\infty} a_k(t) \epsilon_{i-k}.
\end{equation*} 
It is easy to see that the physical dependence measure of the above linear process, denoted as $\delta(k,q)$ satisfies that
\begin{equation*}
\delta(k,q)=\OO(\sup_t |a_k(t)|). 
\end{equation*} 
Consequently, according to \cite[Example 2.4]{DZ}, (\ref{eq_slc}) and Assumption \ref{assum_physical} are satisfied if the following conditions hold 
\begin{equation*}
\sup_t|a_k(t)| \leq Ck^{-\tau}, \ \ \sum_{k=0}^{\infty} \sup_{t \in [0,1]} |a_k'(t)|^{\min\{2,q\}}<\infty. 
\end{equation*} 
\end{example}

\begin{example}[Locally stationary nonlinear time series]\label{exam_nonlinear} Let $\{\epsilon_i\}$ be i.i.d. random variables. Consider a locally stationary process as follows 
\begin{equation*}
G(t, \mathcal{F}_i)=R(t, G(t, \mathcal{F}_{i-1}), \epsilon_i),
\end{equation*}
where $R$ is some measurable function. The above expression is quite general such that many locally stationary time series can be written in the above form. For example, the threshold autoregressive model, exponential autoregressive model and bilinear autoregressive models. According to \cite[Theorem 6]{ZW}, suppose that for some $x_0,$ $\sup_t\|R(t,x_0,\epsilon_i) \|_q<\infty.$ Denote 
\begin{equation*}
\chi:=\sup_t L(t), \ \text{where} \ L(t):=\sup_{x \neq y} \frac{\|R(t,x, \epsilon_0)-R(t,y,\epsilon_0) \|_q}{|x-y|}.
\end{equation*}
Note that $\chi<1$  and $\delta(k,q)=\OO(\chi^k).$ Therefore, Assumption \ref{assum_physical} is satisfied. Moreover, by \cite[Proposition 4]{ZW}, (\ref{eq_slc}) holds if 
\begin{equation*}
\sup_{t \in [0,1]} \| M(G(t, \mathcal{F}_0)) \|_q<\infty,  \ \text{where} \ M(x):=\sup_{0\leq t<s \leq 1}\frac{\|R(t,x,\epsilon_0)-R(s,x,\epsilon_0) \|_q}{|t-s|}. 
\end{equation*}


As a concrete example, we consider the time varying threshold autoregressive (TVTAR) model denoted as 
\begin{equation*}
G(t, \mathcal{F}_i)=a(t)[G(t, \mathcal{F}_{i-1})]^++b(t)[-G(t, \mathcal{F}_{i-1})]^+.
\end{equation*}
In this case, Assumption \ref{assum_physical} and (\ref{eq_slc}) will be satisfied if 
\begin{equation*}
a(t), b(t) \in \mathsf{C}^1([0,1]), \ \text{and} \ \sup_{t \in [0,1]}\left( |a(t)|+|b(t)| \right)<1. 
\end{equation*}
\end{example}

\subsection{Examples of mappings}\label{appendix_mapping}
In this subsection, we list a few examples of the mappings as defined in Definition \ref{defn_mappings}. 

\begin{example}\label{example_mappings} In this example, we provide some mappings when $\Lambda=\mathbb{R}.$ The first mapping is the following algebraic mapping 
\begin{equation}\label{eq_y}
x=\frac{sy}{\sqrt{1-y^2}}, \ y=\frac{x}{\sqrt{x^2+s^2}}.
\end{equation}
The second mapping is the following logarithmic mapping
\begin{equation}\label{eq_y1}
x=\frac{s}{2}\ln \frac{1+y}{1-y}, \  y=\tanh(s^{-1}x)=\frac{e^{s^{-1} x}-e^{-s^{-1}x}}{e^{s^{-1}x}+e^{-s^{-1}x}}.
\end{equation}
\end{example}
\begin{example}\label{example_mappings2} In this example, we provide some mappings when $\Lambda=\mathbb{R}_+.$ The first mapping is the following algebraic mapping 
\begin{equation*}
x=\frac{s(1+y)}{1-y}, \ y=\frac{x-s}{x+s}.
\end{equation*}
The second mapping is the following logarithmic mapping 
\begin{equation*}
x=\frac{s}{2} \ln \frac{3+y}{1-y}, \ y=1-2\tanh(s^{-1}x). 
\end{equation*}
\end{example}

\subsection{The Schwartz space and its generalization}\label{sec_Schwartzfunction} In this subsection, we provide a general class of functional space so that the conditions of Assumption \ref{assum_smoothnessasumption} are satisfied. We first define the Schwartz space \cite[Chapter 8]{lang1993real}. 
\begin{definition}[Schwartz space] The Schwartz space $\mathcal{S}(\mathbb{R}^d)$ is the topological vector space of Schwartz functions $f: \mathbb{R}^d \rightarrow \mathbb{C}$ satisfying
\begin{enumerate}
\item $f \in \mathtt{C}^{\infty}(\mathbb{R}^d);$
\item For all multi-indices $\alpha, \beta,$ there exists some positive constant $C_{\alpha, \beta}<\infty$ so that 
\begin{equation*}
|x^{\alpha} \partial_x^\beta f(x)| \leq C_{\alpha,\beta}. 
\end{equation*}
\end{enumerate} 
\end{definition}

Based on the above definition, we can see that $f$ and all of its derivatives are rapidly decreasing.  Moreover, one can also check that for $f \in \mathcal{S}(\mathbb{R}^d),$ for all multi-indices $\alpha$ and any integer $N,$ there exists some positive constant $C_{N, \alpha}$ so that 
\begin{equation*}
|\partial^{\alpha}_x f(x)| \leq C_{N, \alpha}(1+|x|)^{-N}.
\end{equation*}

It is easy to check from the chain rules for higher derivatives \cite{huang2006chain} that the second condition in (\ref{eq_importantregulaityassumption}) will be satisfied with $\mathsf{m}_2=\infty$ if $m_j(\cdot,x) \in \mathcal{S}(\mathbb{R}).$ For general $\mathsf{m}_2,$ we can extend the Schwartz space as follows.  

\begin{definition}[Generalized Schwartz space] The generalized Schwartz space of order $\mathsf{m}$ $\mathcal{S}_g(\mathbb{R}^d, \mathsf{m})$ is the topological vector space of generalized Schwartz functions $f: \mathbb{R}^d \rightarrow \mathbb{C}$ satisfying
\begin{enumerate}
\item $f \in \mathtt{C}^{\mathsf{m}}(\mathbb{R}^d);$
\item For all multi-indices $\alpha, \beta$ that $|\alpha|, |\beta| \leq \mathsf{m},$ there exists some positive constant $C_{\alpha, \beta}<\infty$ so that 
\begin{equation*}
|x^{\alpha} \partial_x^\beta f(x)| \leq C_{\alpha,\beta}. 
\end{equation*}
\end{enumerate} 
\end{definition}

It is easy to check from the chain rules for higher derivatives \cite{huang2006chain} that the second condition in (\ref{eq_importantregulaityassumption}) will be satisfied  if $m_j(\cdot,x) \in \mathcal{S}_g(\mathbb{R},\mathsf{m}_2).$

\subsection{Additional remarks} In this subsection, we provide some additional remarks. 

\begin{remark} 
We point out that as discussed in Remark \ref{rmk_modelsetting}, in the context of locally stationary nonlinear AR model that $X_{j,i}=Y_{i-j},$ our construction still applies with minor changes. The modification is necessary since we only observe one time series $\{X_i\}$ in this setting instead of a pair $\{(Y_i, X_i)\}$. Therefore, we shall only consider the response $Y_i=X_{r+i}, 1 \leq i \leq n-r$ so that $\bm{Y} \in \mathbb{R}^{n-r}.$ Consequently, the design matrix will be $W \in \mathbb{R}^{(n-r) \times rcd}$ with (\ref{eq_designmatrix}) replaced by 
\begin{equation*}
W_{ik}=\phi_{\ell_1}(t_i) \varphi_{\ell_2+1}(X_{i-(l+1)}). 
\end{equation*}
\end{remark}

\begin{remark}\label{rem_meaninference}
In this remark, we briefly explain how to construct the SCR for the time-varying intercept $m_0(t).$ Similar to the discussions in Section \ref{sec_problemsetup}, we shall use the statistic
\begin{equation*}
\mathsf{T}_0(t)=\widehat{m}_{0}(t)-m_0(t). 
\end{equation*}
Let the long-run variance of $\{\epsilon_i\}$ be $\sigma(t)$ and denote $\Sigma \in \mathbb{R}^{c_0 \times c_0}$ as
\begin{equation*}
\Sigma=\int_0^1 \sigma(t) \otimes (\phib_0(t) \phib_0^\top(t)) \mathrm{d} t.
\end{equation*}
Denote 
\begin{equation*}
h_0(t)=\sqrt{\phib_0^\top(t) \Sigma \phib_0(t)}.
\end{equation*}
Similar to the results in Theorem \ref{thm_asymptoticdistribution}, we can show that $\sqrt{n} \mathsf{T}_0(t)$ will be asymptotically Gaussian with mean zero and variance $h_0(t)^2.$ Then the SCR and the associated multiplier bootstrap procedure can be constructed accordingly. Using this approach, one can perform inference on the intercept, such as determining whether it is significant or time-varying.  
\end{remark}

{
\begin{remark}\label{rmk_hj(tx)} We provide some remarks for the quantity $h_j(t,x)$ in (\ref{eq_defhtx}). First, using the lower bound from Assumption \ref{assum_updc}, we have that for some constant $C>0$ 
\begin{equation*}
h_j^2(t,x) \geq C |\bm{l}_j|^2. 
\end{equation*}
Then combining with the assumptions of Theorem \ref{thm_asymptoticdistribution}, we can see that uniformly for $(t,x),$ there exists some constant $C_1>0$
\begin{equation*}
h_j^2(t,x) \geq C_1. 
\end{equation*}

Second, using Assumption \ref{assum_updc}, we can see that uniformly for $(t,x),$ we have that 
\begin{equation*}
\frac{|\bm{l}_j|}{h_j(t,x)} \asymp 1. 
\end{equation*}

\end{remark}
}

\section{Additional simulation results and real data analysis}\label{addtional_numerical} In this section, we provide some additional simulation results and an additional financial real data analysis.
\subsection{Simulation setup for $\epsilon_i$ used in Section \ref{sec_simulationsettup}}\label{sec_simulationsetupepsilon}
For the locally stationary time series $\{\epsilon_i\}$ used in Section \ref{sec_simulationsettup}, let $a_1(i/n) \equiv 0.3, \ a_2(i/n)=0.3 \sin(2 \pi i/n),$ and $\eta_i, i=1,2,\cdots,n,$ be i.i.d. standard Gaussian random variables,  we consider the following settings
\begin{enumerate}
\item[(a).] Time-varying linear AR(2) model  
\begin{equation*}
\epsilon_i=\sum_{j=1}^2 a_j(\frac{i}{n}) \epsilon_{i-j}+\eta_i.
\end{equation*} 
\item[(b).] Self-exciting threshold auto-regressive (SETAR) model
\begin{equation*}
\epsilon_i=
\begin{cases}
a_1(\frac{i}{n}) \epsilon_{i-1}+\eta_i, & \epsilon_{i-1} \geq 0, \\
a_2(\frac{i}{n}) \epsilon_{i-1}+\eta_i, & \epsilon_{i-1}<0.
\end{cases}
\end{equation*}
\item[(c).] First order bilinear model 
\begin{equation*}
\epsilon_i=\left(a_1(\frac{i}{n}) \eta_{i-1}+a_2(\frac{i}{n}) \right)\epsilon_{i-1}+\eta_i.
\end{equation*}
\end{enumerate}

\subsection{Simulation results for Section \ref{sec_numerical}} In this section, we report the simulation results of Section \ref{sec_numerical}.

\begin{table}[!ht]
\begin{center}
\setlength\arrayrulewidth{1pt}
\renewcommand{\arraystretch}{1.5}
{\fontsize{10}{10}\selectfont 
\begin{tabular}{|c|ccc|ccc|ccc|ccc|}
\hline
& \multicolumn{6}{c|}{\large nominal level: $90\%$} & \multicolumn{6}{c|}{\large nominal level: $95\%$} \\
\hline
      & \multicolumn{3}{c|}{$n=500$}                                                                                                                       & \multicolumn{3}{c|}{$n=800$}                                                                                                                        & \multicolumn{3}{c|}{$n=500$}                                                                                                                       & \multicolumn{3}{c|}{$n=800$}                                                                                                                        \\ \hline
Model/$\epsilon_i$ & \multicolumn{1}{c|}{(a)} & \multicolumn{1}{c|}{(b)} & \multicolumn{1}{c|}{(c)} &  \multicolumn{1}{c|}{(a)} & \multicolumn{1}{c|}{(b)} & \multicolumn{1}{c|}{(c)}  & \multicolumn{1}{c|}{(a)} & \multicolumn{1}{c|}{(b)} & \multicolumn{1}{c|}{(c)} &  \multicolumn{1}{c|}{(a)} & \multicolumn{1}{c|}{(b)} & \multicolumn{1}{c|}{(c)} \\ 
\hline
     & \multicolumn{6}{c|}{Sieve estimators (Fourier basis)}                                                                                                                                                                                                                                                                                           & \multicolumn{6}{c|}{Sieve estimators (Fourier basis)}                                                                                                                                                                                                                                                                                          \\
   \hline
(1)     & 0.853             &  0.835                         &                0.862           &  0.856                          &                0.865           &  0.884                                     & 0.931 & 0.939 & 0.929 & 0.938 & 0.94 & 0.938 \\
(2)    & 0.842         &  0.848 &   0.855                    &         0.876                & 0.885                         & 0.908 & 0.929 & 0.961 & 0.935 & 0.939 & 0.958 & 0.945                \\
(3)    & 0.852  &  0.835  &   0.864                       &        0.913           &    0.874                      & 0.865 & 0.963 & 0.936 & 0.961 & 0.955 & 0.954 & 0.947   \\
\hline
      & \multicolumn{6}{c|}{Sieve estimators (Legendre basis)}                                                                                                                                                                                                                                                                                          & \multicolumn{6}{c|}{Sieve estimators (Legendre basis)}                                                                                                                                                                                                                                                                                         \\
       \hline
(1)     &   0.826            &   0.845              &         0.852                  &     0.843                      &    0.849 &                               0.857  & 0.934 & 0.962 & 0.932 & 0.948 & 0.943 & 0.953          \\
(2)     &  0.928       &  0.875   &  0.918                  &          0.921                & 0.883                          &  0.894 & 0.962 & 0.939 & 0.943 & 0.957 & 0.945 & 0.944  \\
(3)     & 0.858   &  0.876    & 0.925                        &         0.861                   & 0.857                           &       0.91   & 0.938 & 0.941 & 0.963 & 0.943 & 0.954 & 0.948                        \\
 \hline
  & \multicolumn{6}{c|}{Sieve estimators (Daubechies-9 basis)}                                                                                                                                                                                                                                                                                         & \multicolumn{6}{c|}{Sieve estimators (Daubechies-9 basis)}                                                                                                                                                                                                                                                                                         \\
       \hline
(1)     &     0.858             &   0.861               &                           0.914 & 0.883                           &   0.867 &     0.91 & 0.941 & 0.959 & 0.961 & 0.945 & 0.948 & 0.956                                   \\
(2)     & 0.867        &   0.855  &  0.859                 &       0.875                    &  0.868                         &  0.877 & 0.961 & 0.959 & 0.939 & 0.956& 0.948 & 0.943    \\
(3)     & 0.913   & 0.924   & 0.863                        &       0.89                    &    0.906                    &     0.879 & 0.937 & 0.964 & 0.958 & 0.944 & 0.956 & 0.947                              \\
 \hline
       & \multicolumn{6}{c|}{Kernel estimators}                                                                                                                                                                                                                                                                                         & \multicolumn{6}{c|}{Kernel estimators}                                                                                                                                                                                                                                                                                         \\
       \hline
(1)     &        0.756         &   0.811                &                           0.765 &            0.754           &   0.834 & 0.817                                    & 0.913 & 0.9 & 0. 886 & 0.896 & 0.922 & 0.915    \\
(2)    & 0.788      &  0.812   &  0.797                                                &        0.814                   &                       0.853 &  0.84  & 0.899 & 0.904 & 0.908 & 0.912 & 0.9 & 0.913  \\
(3)     & 0.746   &  0.798   &    0.754                      &         0.81                   &  0.807                          &           0.798   & 0.903 & 0.899 & 0.906 & 0.921 & 0.918 & 0.918                     \\
 \hline
\end{tabular}
}
\end{center}
\caption{Simulated coverage probabilities at $90\%$ and  $95\%$ nominal levels. For models (1)--(3) in Section \ref{sec_simulationsettup}, we set $\delta=1.$ The coverage probabilities are based on 5,000 simulations on the region $(t,x) \in [0,1] \times [-10, 10].$ For our method, the SCRs can be automatically constructed using the function $\mathtt{auto.SCR}$ from our $\mathtt{R}$ package $\mathtt{SIMle}.$}
\label{table_cp}
\end{table}

\begin{table}[!ht]
\begin{center}
\setlength\arrayrulewidth{1pt}
\renewcommand{\arraystretch}{1.5}
{\fontsize{10}{10}\selectfont 
\begin{tabular}{|c|ccc|ccc|}
\hline
      & \multicolumn{3}{c|}{$n=500$}                                                                                                                       & \multicolumn{3}{c|}{$n=800$}                                                                                                                        \\ \hline
Testing/$\epsilon_i$ & \multicolumn{1}{c|}{(a)} & \multicolumn{1}{c|}{(b)} & \multicolumn{1}{c|}{(c)} &  \multicolumn{1}{c|}{(a)} & \multicolumn{1}{c|}{(b)} & \multicolumn{1}{c|}{(c)}  \\ 
\hline
     & \multicolumn{6}{c|}{Fourier basis}                                                                                                                                                                                                                                                                                          \\
   \hline
   Exact form   &  0.113  &  0.103   &  0.102                        & 0.105                   &             0.093             &          0.098          \\
Stationarity     &          0.114  & 0.091 &  0.089                         &             0.11              &      0.087                     &       0.108                                  \\
Separability    &   0.118         &                        0.107   &       0.113                    & 0.105                           &     0.092                      &  0.094                 \\
\hline
      & \multicolumn{6}{c|}{Legendre basis}                                                                                                                                                                                                                                                                                         \\
       \hline
       Exact form    &  0.113  &  0.108   & 0.11                        &        0.112                    & 0.098                          &      0.105                            \\
Stationarity       & 0.112                  & 0.11                  &                   0.088        &  0.11                         &0.107    &  0.106                                          \\
Separability     & 0.09        & 0.092    & 0.111                                           &      0.108                     &  0.109                       &  0.094  \\
 \hline
  & \multicolumn{6}{c|}{Daubechies-9 basis}                                                                                                                                                                                                                                                                                         \\
       \hline
       Exact form     &  0.096 &  0.105  & 0.109                       &            0.096               &  0.099                        &       0.104    \\
Stationarity       &   0.112              & 0.087 & 0.091                          & 0.092                         &  0.093  &  0.107                                         \\
Separability      &  0.117      &   0.085  & 0.092                                                  & 0.109                          &  0.09                       &   0.092  \\
 \hline
\end{tabular}
}
\end{center}
\caption{ Simulated type one error rates under nominal level $0.1$. The results are reported based on 5,000 simulations. As mentioned earlier,  exact form refers to the test on $m(t,x)=\mathfrak{m}_0(t,x)$ for some pre-given function $\mathfrak{m}_0(t,x),$  stationarity refers to the test on $m(t,x) \equiv m(x)$ and separability refers to the test on $m(t,x)=\rho(t) g(x)$ for some smooth functions $\rho(\cdot)$ and $g(\cdot)$ All these tests can be automatically implemented respectively using the functions $\mathtt{auto.SCR, auto.homo.test}$ and $\mathtt{auto.sep.test}$ in our $\mathtt{R}$ package $\mathtt{SIMle}.$   
}
\label{table_typeoneerror}
\end{table}

\begin{figure}[!ht]
\begin{subfigure}{0.32\textwidth}
\includegraphics[width=6.4cm,height=5cm]{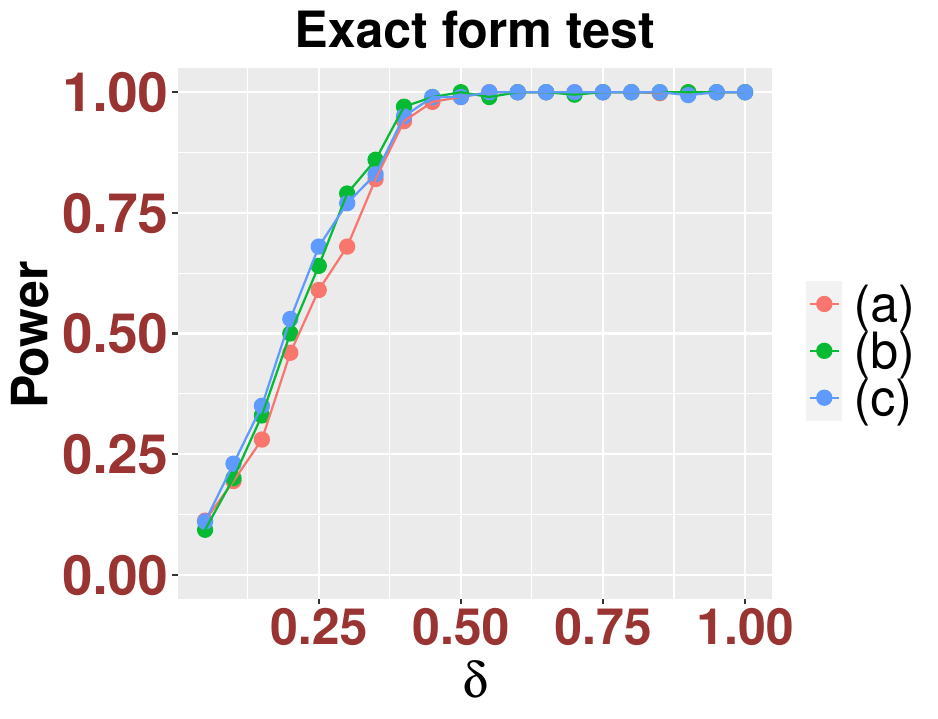}
\end{subfigure}
\begin{subfigure}{0.32\textwidth}
\includegraphics[width=6.4cm,height=5cm]{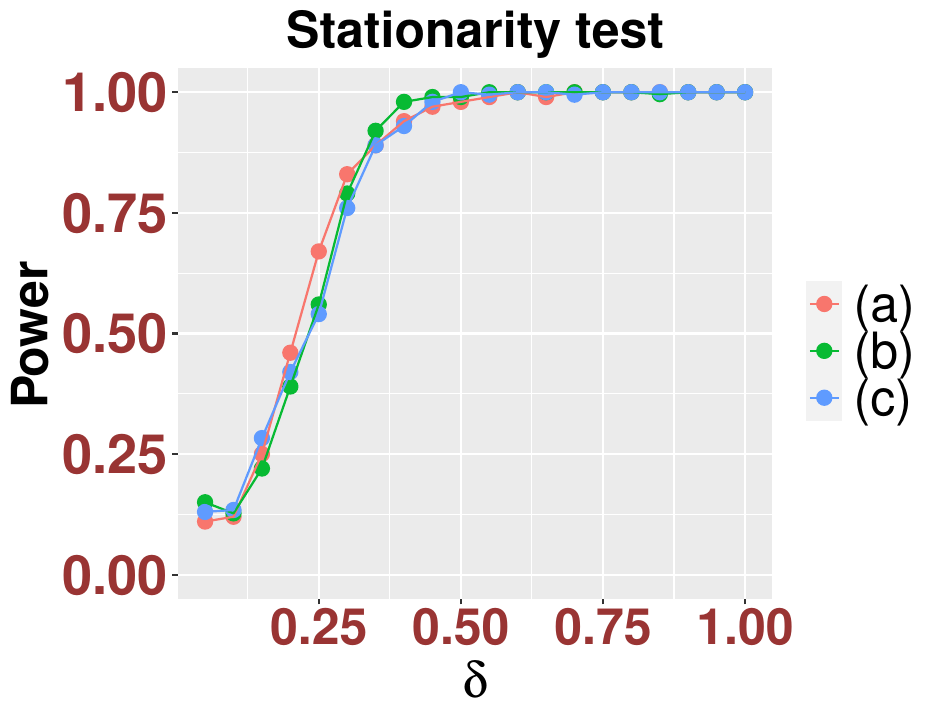}
\end{subfigure}
\begin{subfigure}{0.32\textwidth}
\includegraphics[width=6.4cm,height=5cm]{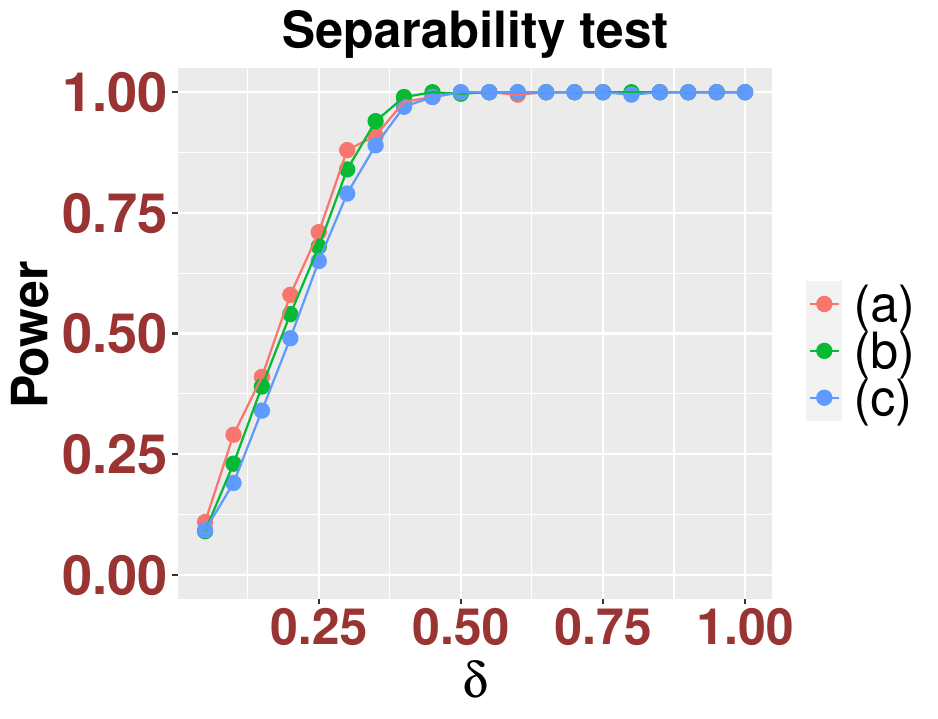}
\end{subfigure}
\vspace*{-0.4cm}
\caption{{ \footnotesize Simulated power for our proposed multiplier bootstrap method under the nominal level $0.1$. Here we used Daubechies-9 basis, $n=800$ and $(a), (b), (c)$ refer to the models for $\epsilon_i$ as in Section \ref{sec_simulationsettup}.  Our results are based on 5,000 simulations.  }  }
\label{fig_selfcomparison}
\end{figure}

\begin{figure}[!ht]
\begin{subfigure}{0.32\textwidth}
\includegraphics[width=6.4cm,height=5cm]{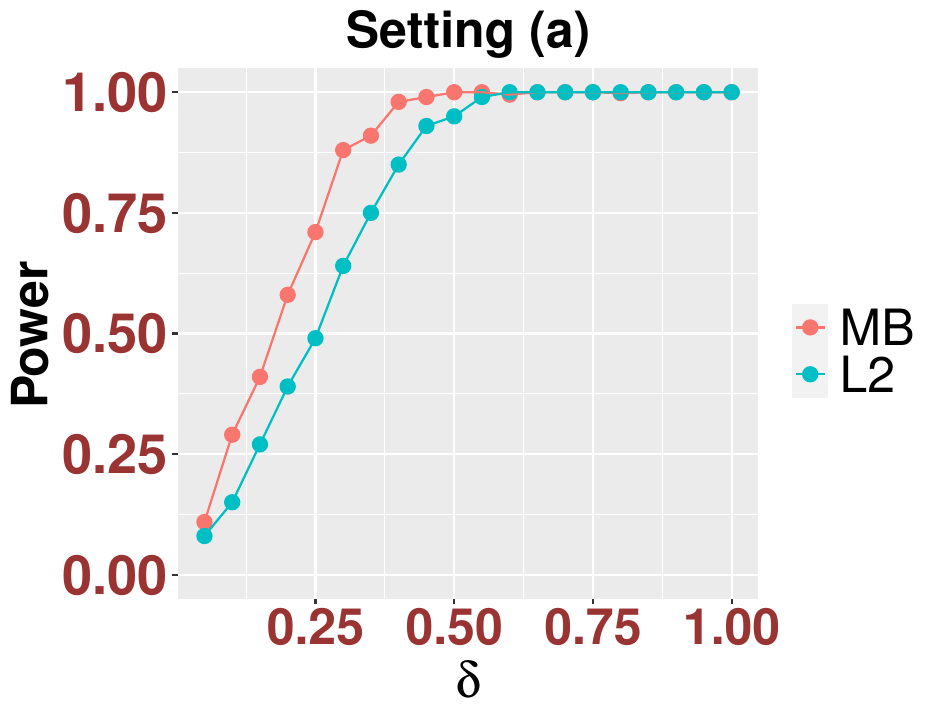}
\end{subfigure}
\begin{subfigure}{0.32\textwidth}
\includegraphics[width=6.4cm,height=5cm]{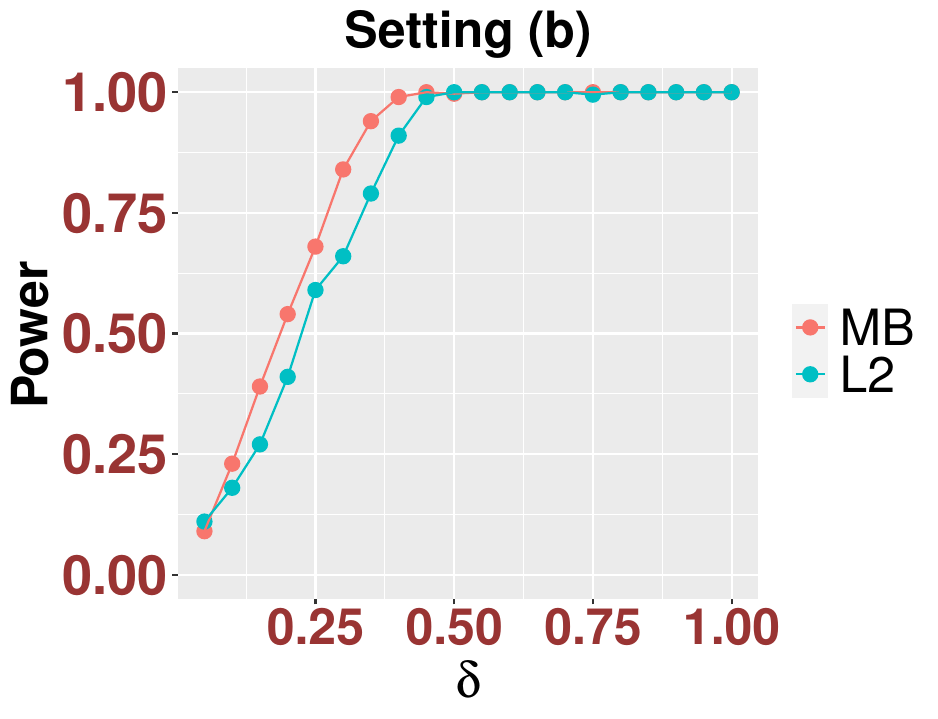}
\end{subfigure}
\begin{subfigure}{0.32\textwidth}
\includegraphics[width=6.4cm,height=5cm]{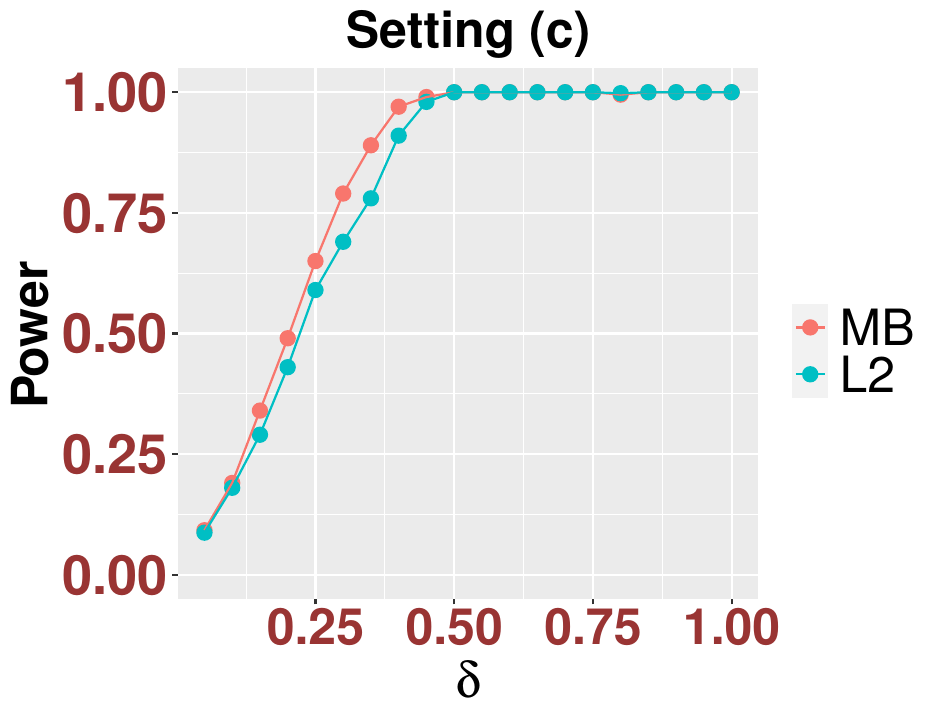}
\end{subfigure}
\vspace*{-0.4cm}
\caption{{ \footnotesize Comparison of the simulated power for our Algorithm \ref{alg:boostrapping} (MB) and the weighted $L_2$-distance based statistic (L2) in \cite{HHY}.  Our results are based on 5,000 simulations.  }  }
\label{fig_jasacomparison}
\end{figure} 

\subsection{Additional simulation results}\label{simu_addtional_identiyissue} 
In this subsection, we conduct some additional simulations when $r=2$. Especially, we consider the following nonlinear AR(2) type model that 
\begin{equation*}
X_i=m_0^*(t_i)+m_1^*(t_i,X_{i-1})+m_2^*(t_i,X_{i-2})+\sigma(t_i, X_{i-1}, X_{i-2}) \epsilon_i.
\end{equation*}  
Regarding the above functions, we consider the following setups
\begin{itemize}
\item[(I).] $m_0^*(t)=2t \cos(2 \pi t), \ m_1^*(t,x)=(2+x^2)^{-4}, \ m_2^*(t,x)=(1+\delta \sin (2 \pi t)) \exp(-x^2).$ For $\sigma(t,x),$ we use the same one as in setup (1) in Section \ref{sec_simulationsettup}. 
\item[(II).] $m_0^*(t)\equiv 2, \ m_1^*(t,x)=\cos(2 \pi t) \exp(-x^2), \ m_2^*(t,x)=(1+\delta \sin (2 \pi t)) \exp(-x^2).$ For $\sigma(t,x),$ we use the same one as in setup (2) in Section \ref{sec_simulationsettup}.
\item[(III).] $m_0^*(t)=2t, \ m_1^*(t,x)= \exp(-x^2), \ m_2^*(t,x)=2t( \delta \exp(-2tx^2)+ \pi^{-1/2}\exp(-x^2/2)).$ For $\sigma(t,x),$ we use the same one as in setup (3) in Section \ref{sec_simulationsettup}.  
\end{itemize}
Note that when $\delta=0,$ $m_2^*(t,x)$ in the above setups (1)-(3) correspond respectively to the null hypotheses in Examples \ref{exam_exactform}--\ref{exam_seperabletest}. 
Moreover, for the locally stationary time series, we use the same setup as in (a)-(c) of Section \ref{sec_simulationsetupepsilon}.

As mentioned in Section  \ref{sec_removenonzeromeanassumption}, for identifiability, we shall work with 
\begin{equation}\label{eq_reallyfocus}
X_i=m_0(t_i)+m_1(t_i,X_{i-1})+m_2(t_i,X_{i-2})+\sigma(t_i, X_{i-1}, X_{i-2}) \epsilon_i,
\end{equation}
where $m_0(t_i)=m_0^*(t_i)+\mathbb{E}m^*_1(t_i,X_{i-1})+\mathbb{E} m^*_2(t_i,X_{i-2}), \ m_1(t_i, X_{i-1})=m_1^*(t_i,X_{i-1})-\mathbb{E} m_1^*(t_i,X_{i-1}), \ m_2(t_i,X_{i-2})=m^*_2(t_i,X_{i-2})-\mathbb{E}m^*_2(t_i,X_{i-2}).$ 

Then we examine the  performance of our proposed sieve estimators for $m_2(t,x)$ in (\ref{eq_reallyfocus}) and their associated SCRs constructed from Algorithm \ref{alg:boostrapping}  using the simulated coverage probabilities.  We also compare our methods with the kernel based estimation methods \cite{MV, ZWW}. Since model (2) is separable, we utilize the method in \cite{CSW} which shows better performance for separable functions. Moreover, since the support of the time series can span over $\mathbb{R}$, to facilitate the comparison with the kernel methods, we focus on the region that $(t,x) \in [0,1] \times [-10, 10].$ For the kernel method, we use the Epanechnikov kernel and the cross-validation approach as in \cite{10.1214/18-AOS1743} to select the bandwidth. The results are reported in Table \ref{table_cpsupp}. We conclude that our estimators achieve reasonably high accuracy and outperform kernel estimators across all commonly used sieve basis functions, as kernel methods may suffer from boundary issues for both $t$ and $x$. 

\begin{table}[!ht]
\begin{center}
\setlength\arrayrulewidth{1pt}
\renewcommand{\arraystretch}{1.5}
{\fontsize{10}{10}\selectfont 
\begin{tabular}{|c|ccc|ccc|ccc|ccc|}
\hline
& \multicolumn{6}{c|}{\large nominal level: $90\%$} & \multicolumn{6}{c|}{\large nominal level: $95\%$} \\
\hline
      & \multicolumn{3}{c|}{$n=500$}                                                                                                                       & \multicolumn{3}{c|}{$n=800$}                                                                                                                        & \multicolumn{3}{c|}{$n=500$}                                                                                                                       & \multicolumn{3}{c|}{$n=800$}                                                                                                                        \\ \hline
Model/$\epsilon_i$ & \multicolumn{1}{c|}{(a)} & \multicolumn{1}{c|}{(b)} & \multicolumn{1}{c|}{(c)} &  \multicolumn{1}{c|}{(a)} & \multicolumn{1}{c|}{(b)} & \multicolumn{1}{c|}{(c)}  & \multicolumn{1}{c|}{(a)} & \multicolumn{1}{c|}{(b)} & \multicolumn{1}{c|}{(c)} &  \multicolumn{1}{c|}{(a)} & \multicolumn{1}{c|}{(b)} & \multicolumn{1}{c|}{(c)} \\ 
\hline
     & \multicolumn{6}{c|}{Sieve estimators (Fourier basis)}                                                                                                                                                                                                                                                                                           & \multicolumn{6}{c|}{Sieve estimators (Fourier basis)}                                                                                                                                                                                                                                                                                          \\
   \hline
(I)     & 0.813             &  0.868                         &                0.83           &  0.86                          &                0.854           &  0.88                                     & 0.93 & 0.942 & 0.919 & 0.948 & 0.944 & 0.92 \\
(II)    & 0.842         &  0.89 &   0.91                    &         0.88                & 0.89                         & 0.92 & 0.918 & 0.95 & 0.958 & 0.941 & 0.938 & 0.961                \\
(III)    & 0.84  &  0.838  &   0.834                       &        0.912           &    0.884                      & 0.873 & 0.951 & 0.946 & 0.964 & 0.958 & 0.944 & 0.948   \\
\hline
      & \multicolumn{6}{c|}{Sieve estimators (Legendre basis)}                                                                                                                                                                                                                                                                                          & \multicolumn{6}{c|}{Sieve estimators (Legendre basis)}                                                                                                                                                                                                                                                                                         \\
       \hline
(I)     &   0.866            &   0.855              &         0.868                  &     0.873                      &    0.889 &                               0.864  & 0.944 & 0.97 & 0.962 & 0.938 & 0.943 & 0.957          \\
(II)     &  0.885       &  0.873   &  0.907                  &          0.918                & 0.884                          &  0.912 & 0.957 & 0.933 & 0.94 & 0.961 & 0.947 & 0.95 \\
(III)     & 0.868   &  0.883    & 0.918                        &         0.881                   & 0.867                           &       0.885   & 0.938 & 0.943 & 0.968 & 0.945 & 0.952 & 0.948                        \\
 \hline
  & \multicolumn{6}{c|}{Sieve estimators (Daubechies-9 basis)}                                                                                                                                                                                                                                                                                         & \multicolumn{6}{c|}{Sieve estimators (Daubechies-9 basis)}                                                                                                                                                                                                                                                                                         \\
       \hline
(I)     &     0.865             &   0.861               &                           0.865 & 0.882                           &   0.887 &     0.907 & 0.944 & 0.961 & 0.963 & 0.947 & 0.942 & 0.957                                   \\
(II)     & 0.857        &   0.862  &  0.872                 &       0.874                    &  0.876                         &  0.878 & 0.964 & 0.949 & 0.932 & 0.951& 0.95 & 0.948    \\
(III)     & 0.885   & 0.914   & 0.893                        &       0.911                    &    0.906                    &     0.894 & 0.934 & 0.954 & 0.959 & 0.947 & 0.952 & 0.952                              \\
 \hline
       & \multicolumn{6}{c|}{Kernel estimators}                                                                                                                                                                                                                                                                                         & \multicolumn{6}{c|}{Kernel estimators}                                                                                                                                                                                                                                                                                         \\
       \hline
(I)     &        0.668        &   0.714                &                           0.754 &            0.753           &   0.811 & 0.81                                    & 0.9 & 0.912 & 0. 87 & 0.89 & 0.923 & 0.917    \\
(II)    & 0.718      &  0.712   &  0.745                                                &        0.8                   &                       0.813 &  0.82  & 0.821 & 0.884 & 0.91 & 0.862 & 0.884 & 0.865  \\
(III)     & 0.7   &  0.694   &    0.714                      &         0.812                   &  0.823                          &           0.732   & 0.84 & 0.891 & 0.911 & 0.875 & 0.865 & 0.886                     \\
 \hline
\end{tabular}
}
\end{center}
\caption{Simulated coverage probabilities at $90\%$ and  $95\%$ nominal levels. For models (I)--(III), we set $\delta=1.$ The coverage probabilities are based on 5,000 simulations on the region $(t,x) \in [0,1] \times [-10, 10].$ For our method, the SCRs can be automatically constructed using the function $\mathtt{auto.SCR}$ from our $\mathtt{R}$ package $\mathtt{SIMle}.$}
\label{table_cpsupp}
\end{table}

Then in Table \ref{table_typeoneerrorsupp}, we report the simulated type one error rates of the above three tests under the null that $\delta=0$ for $m_2(t,x)$ in (\ref{eq_reallyfocus}). For the ease of statements, we call the above three tests as exam form, stationarity and separability tests, respectively. We find that our Algorithm \ref{alg:boostrapping} is reasonably accurate for all these tests. Similar to what have been reported in Figures \ref{fig_selfcomparison} and \ref{fig_jasacomparison},  we can examine the power of our methodologies as $\delta$ increases away from zero. It can also be concluded that once $\delta$ deviates away from $0$ a little bit, our method will be able to reject the null hypothesis. Since the plots and patterns are similar to the aforementioned figures, we will omit them here to save space.
\begin{table}[!ht]
\begin{center}
\setlength\arrayrulewidth{1pt}
\renewcommand{\arraystretch}{1.5}
{\fontsize{10}{10}\selectfont 
\begin{tabular}{|c|ccc|ccc|}
\hline
      & \multicolumn{3}{c|}{$n=500$}                                                                                                                       & \multicolumn{3}{c|}{$n=800$}                                                                                                                        \\ \hline
Testing/$\epsilon_i$ & \multicolumn{1}{c|}{(a)} & \multicolumn{1}{c|}{(b)} & \multicolumn{1}{c|}{(c)} &  \multicolumn{1}{c|}{(a)} & \multicolumn{1}{c|}{(b)} & \multicolumn{1}{c|}{(c)}  \\ 
\hline
     & \multicolumn{6}{c|}{Fourier basis}                                                                                                                                                                                                                                                                                          \\
   \hline
   Exact form   &  0.094  &  0.098   &  0.105                        & 0.094                   &             0.097             &          0.104          \\
Stationarity     &          0.098  & 0.094 &  0.104                         &             0.105              &      0.094                     &       0.106                                  \\
Separability    &   0.11         &                        0.089   &       0.112                    & 0.091                           &     0.094                      &  0.114                 \\
\hline
      & \multicolumn{6}{c|}{Legendre basis}                                                                                                                                                                                                                                                                                         \\
       \hline
       Exact form    &  0.113  &  0.097   & 0.109                        &        0.09                    & 0.089                          &      0.108                            \\
Stationarity       & 0.114                  & 0.096                  &                   0.091        &  0.108                         &0.104    &  0.108                                          \\
Separability     & 0.112        & 0.091    & 0.106                                           &      0.092                     &  0.098                       &  0.093  \\
 \hline
  & \multicolumn{6}{c|}{Daubechies-9 basis}                                                                                                                                                                                                                                                                                         \\
       \hline
       Exact form     &  0.103 &  0.102  & 0.105                       &            0.106               &  0.093                        &       0.094    \\
Stationarity       &   0.105              & 0.093 & 0.094                          & 0.089                         &  0.091  &  0.097                                         \\
Separability      &  0.107      &   0.095  & 0.102                                                  & 0.098                          &  0.093                       &   0.095  \\
 \hline
\end{tabular}
}
\end{center}
\caption{ Simulated type one error rates under nominal level $0.1$. }
\label{table_typeoneerrorsupp}
\end{table}

\subsection{Some plots}\label{sec_someplotsfigures} 
First, we provide some plots in Figures \ref{fig_exact}, \ref{fig_homogene} and \ref{fig_sepa} on how the SCRs can be used to conduct the tests outlined in Examples \ref{exam_exactform}, \ref{exam_stationarytest} and \ref{exam_seperabletest}.

\begin{figure}[h]
\hspace*{-1cm}
\begin{subfigure}{0.4\textwidth}
\includegraphics[width=6.4cm,height=5.5cm]{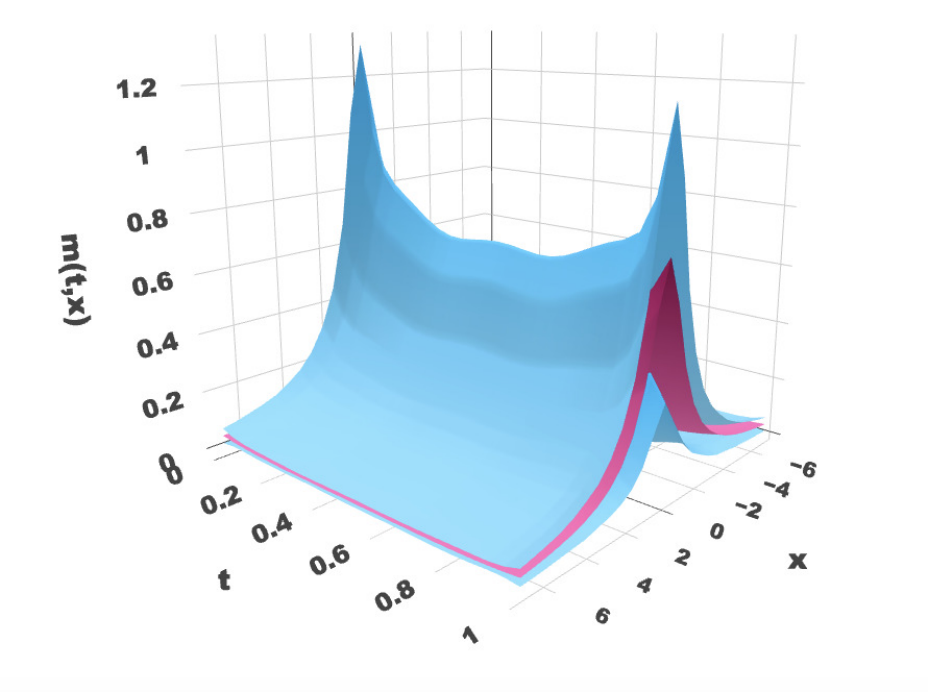}
\end{subfigure}
\hspace*{0.7cm}
\begin{subfigure}{0.4\textwidth}
\includegraphics[width=6.4cm,height=5.5cm]{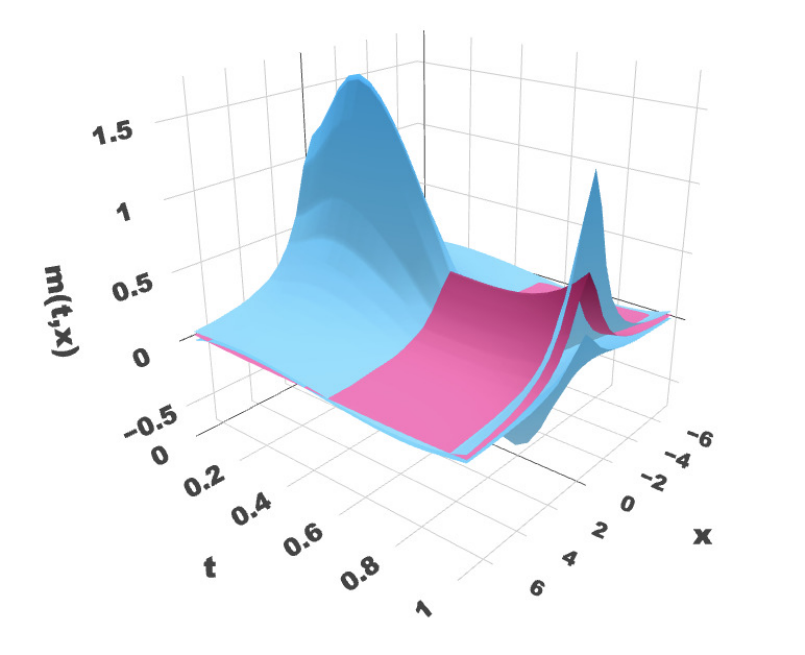}
\end{subfigure}
\vspace*{-0.4cm}
\caption{{ \footnotesize Exact form test using SCR. We use simulation setting (1) in Section \ref{sec_simulationsettup}, where $\delta=0$ corresponds to the null hypothesis (\ref{eq_nullhypotheis}) and $\delta>0$ corresponds to the alternative. The red surface is the exact form to be tested as discussed in Example \ref{exam_exactform} and the blue surfaces are the SCR using (\ref{eq_scrdefinition}) whose estimation and construction will be discussed in Section \ref{sec_practicalimplementation} and summarized in Algorithm \ref{alg:boostrapping}.  The left panel corresponds to the null setting $\delta=0$ where the pre-given function lies inside the SCRs. The right panel corresponds to the alternative setting $\delta=1$ where the pre-given function does not lie inside the SCRs so we need to reject (\ref{eq_nullhypotheis}). The implementation can be done using the function $\mathtt{auto.SCR}$ and the plots can be generated using $\mathtt{test.plot}$ using our package $\mathtt{SIMle}.$ Here $n=800$ and we used the Legendre and mapped Legendre basis functions. All the parameters will be chosen automatically via our $\mathtt{R}$ package as discussed in Section \ref{sec_parameterchoice} of our supplement \cite{suppl}. } }
\label{fig_exact}
\end{figure}


\begin{figure}[!h]
\hspace*{-1cm}
\begin{subfigure}{0.5\textwidth}
\includegraphics[width=6.4cm,height=5.5cm]{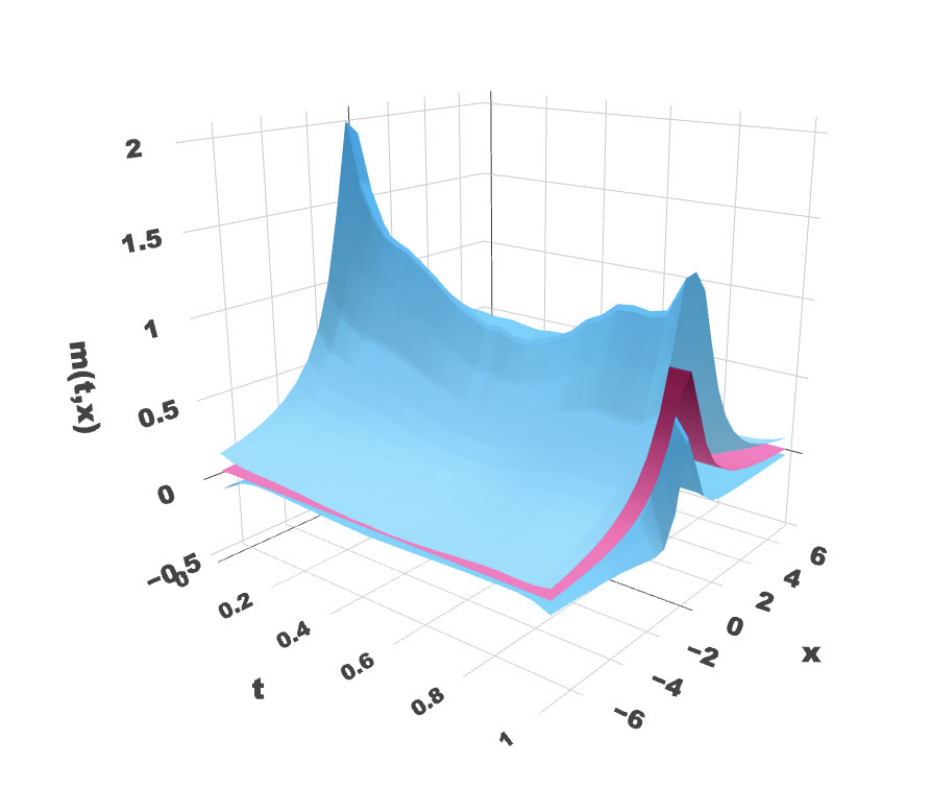}
\end{subfigure}
\hspace*{0.7cm}
\begin{subfigure}{0.5\textwidth}
\includegraphics[width=6.4cm,height=5.5cm]{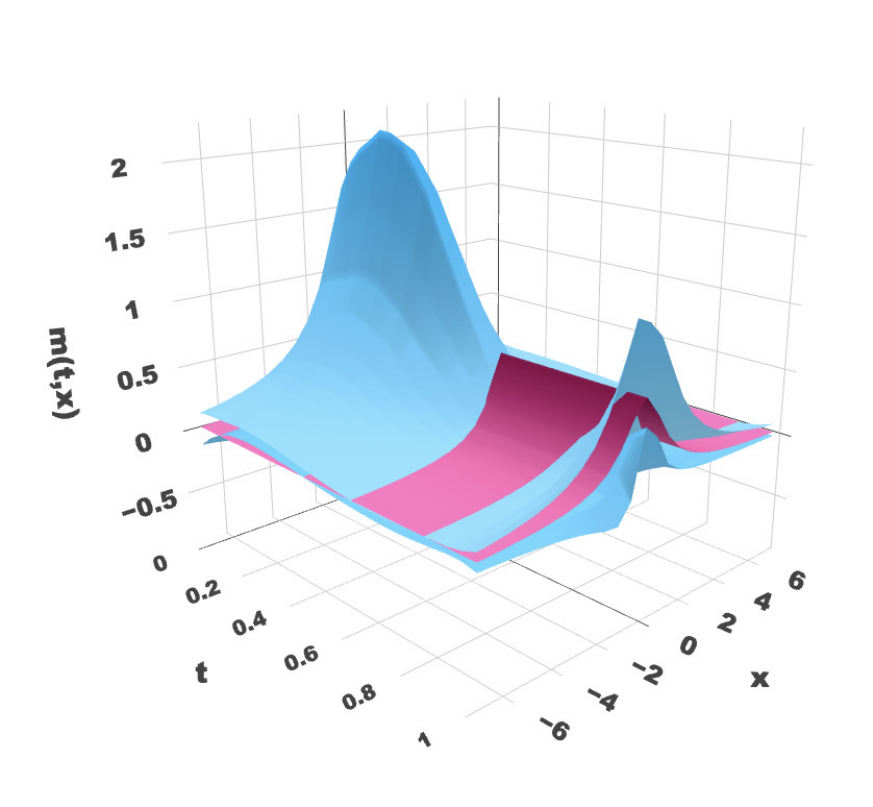}
\end{subfigure}
\caption{{ \footnotesize Time-homogeneity test using SCR. We use simulation setting (2) in Section \ref{sec_simulationsettup}, where $\delta=0$ corresponds to the null hypothesis (\ref{eq_hotestinghomogeneity}) and $\delta>0$ corresponds to the alternative. The red surface is the estimation under (\ref{eq_hotestinghomogeneity}) as discussed in Example \ref{exam_stationarytest} and the blue surfaces are the SCRs using (\ref{eq_scrdefinition}) whose estimation and construction are discussed in Section \ref{sec_practicalimplementation} and summarized in Algorithm \ref{alg:boostrapping}. The left panel corresponds to the null setting $\delta=0$ where the estimated function lies inside the SCRs since it converges faster. The right panel corresponds to the alternative setting $\delta=1$ where the estimated function does not lie inside the SCRs so we need to reject (\ref{eq_hotestinghomogeneity}). The implementation can be done using the function $\mathtt{auto.homo.test}$ and the plots can be generated using $\mathtt{test.plot}$ using our package $\mathtt{SIMle}.$ Here $n=800$ and we used the Legendre and mapped Legendre basis functions. All the parameters will be chosen automatically via our $\mathtt{R}$ as discussed in Section \ref{sec_parameterchoice}. } }
\label{fig_homogene}
\end{figure}

\begin{figure}[!h]
\hspace*{-1cm}
\begin{subfigure}{0.5\textwidth}
\includegraphics[width=6.4cm,height=5.5cm]{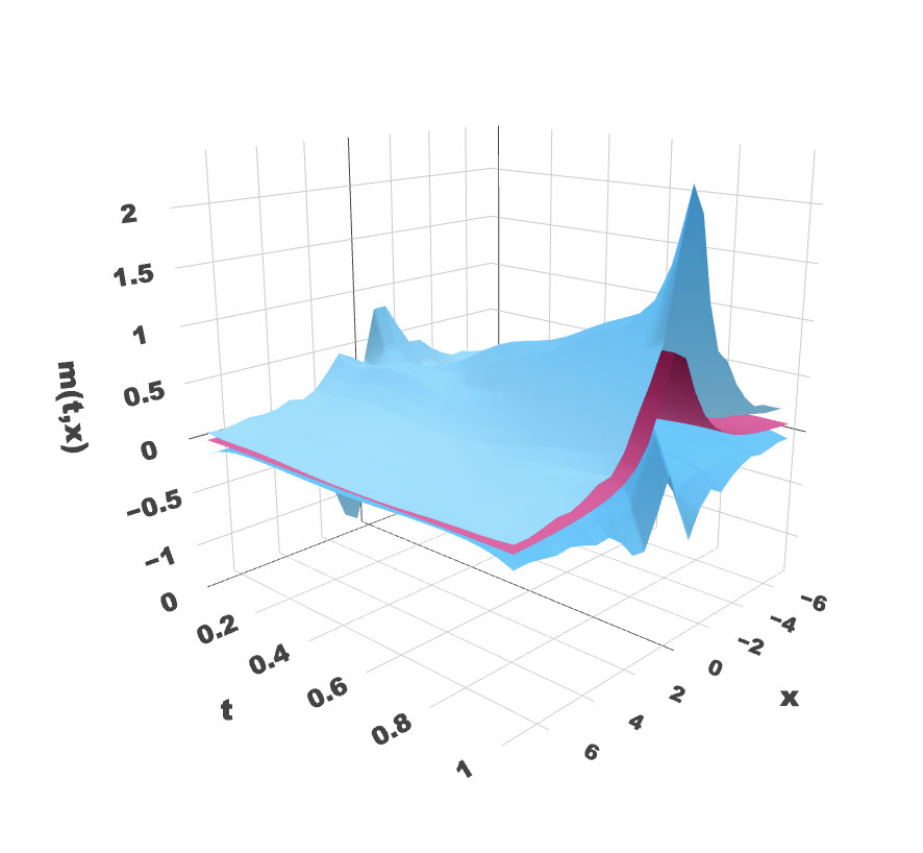}
\end{subfigure}
\hspace*{0.7cm}
\begin{subfigure}{0.5\textwidth}
\includegraphics[width=6.4cm,height=5.5cm]{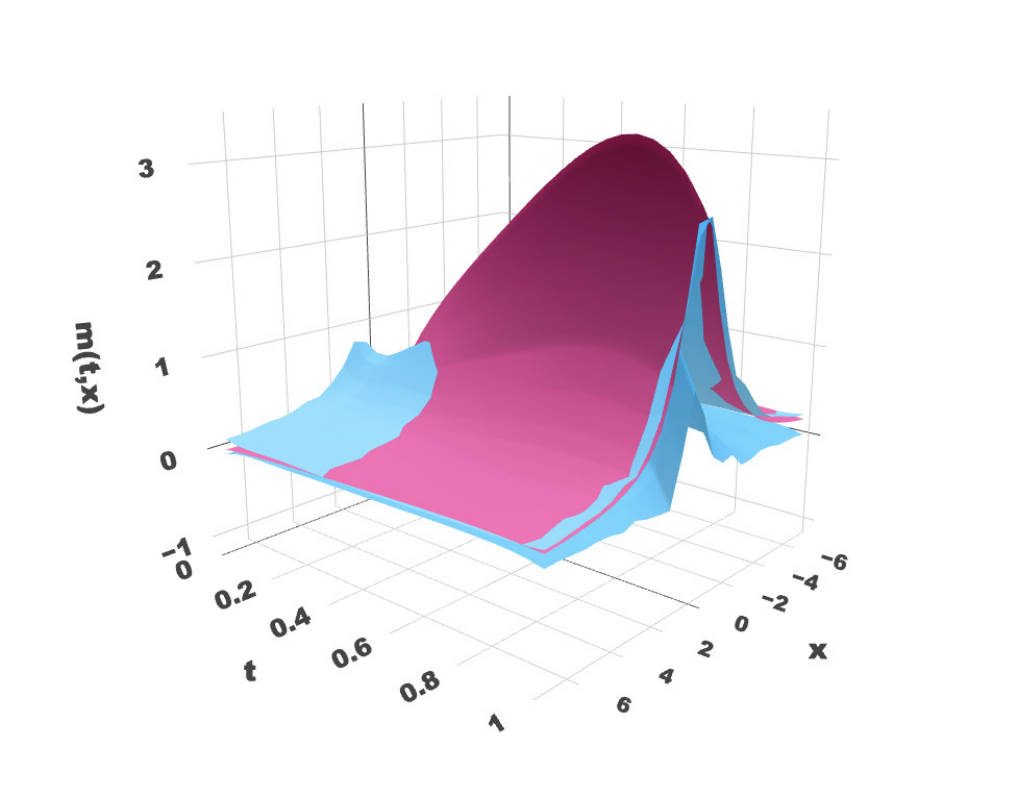}
\end{subfigure}
\caption{{ \footnotesize Separability test using SCR. We use simulation setting (3) in Section \ref{sec_simulationsettup}, where $\delta=0$ corresponds to the null hypothesis (\ref{eq_hotestseparability}) and $\delta>0$ corresponds to the alternative. The red surface is the estimation under (\ref{eq_hotestseparability}) as discussed in Example \ref{exam_seperabletest} and the blue surfaces are the SCRs using (\ref{eq_scrdefinition}) whose estimation and construction are discussed in Section \ref{sec_practicalimplementation} and summarized in Algorithm \ref{alg:boostrapping}. The left panel corresponds to the null setting $\delta=0$ where the estimated function lies inside the SCRs since it converges faster. The right panel corresponds to the alternative setting $\delta=1$ where the estimated function does not lie inside the SCRs so we need to reject (\ref{eq_hotestseparability}). The implementation can be done using the function $\mathtt{auto.sep.test}$ and the plots can be generated using $\mathtt{test.plot}$ using our package $\mathtt{SIMle}.$ Here $n=800$ and we used the Legendre and mapped Legendre basis functions. All the parameters will be chosen automatically via our $\mathtt{R}$ as discussed in Section \ref{sec_parameterchoice}.  } }
\label{fig_sepa}
\end{figure}

Second, we provide some plots (Figures \ref{fig_timevaryingrealdata}--\ref{fig_ourestimation}) to support our real data analysis in Section \ref{sec_realdata}.

\begin{figure}[!ht]
\hspace*{-0.8cm}
\begin{subfigure}{0.34\textwidth}
\includegraphics[width=6cm,height=5.4cm]{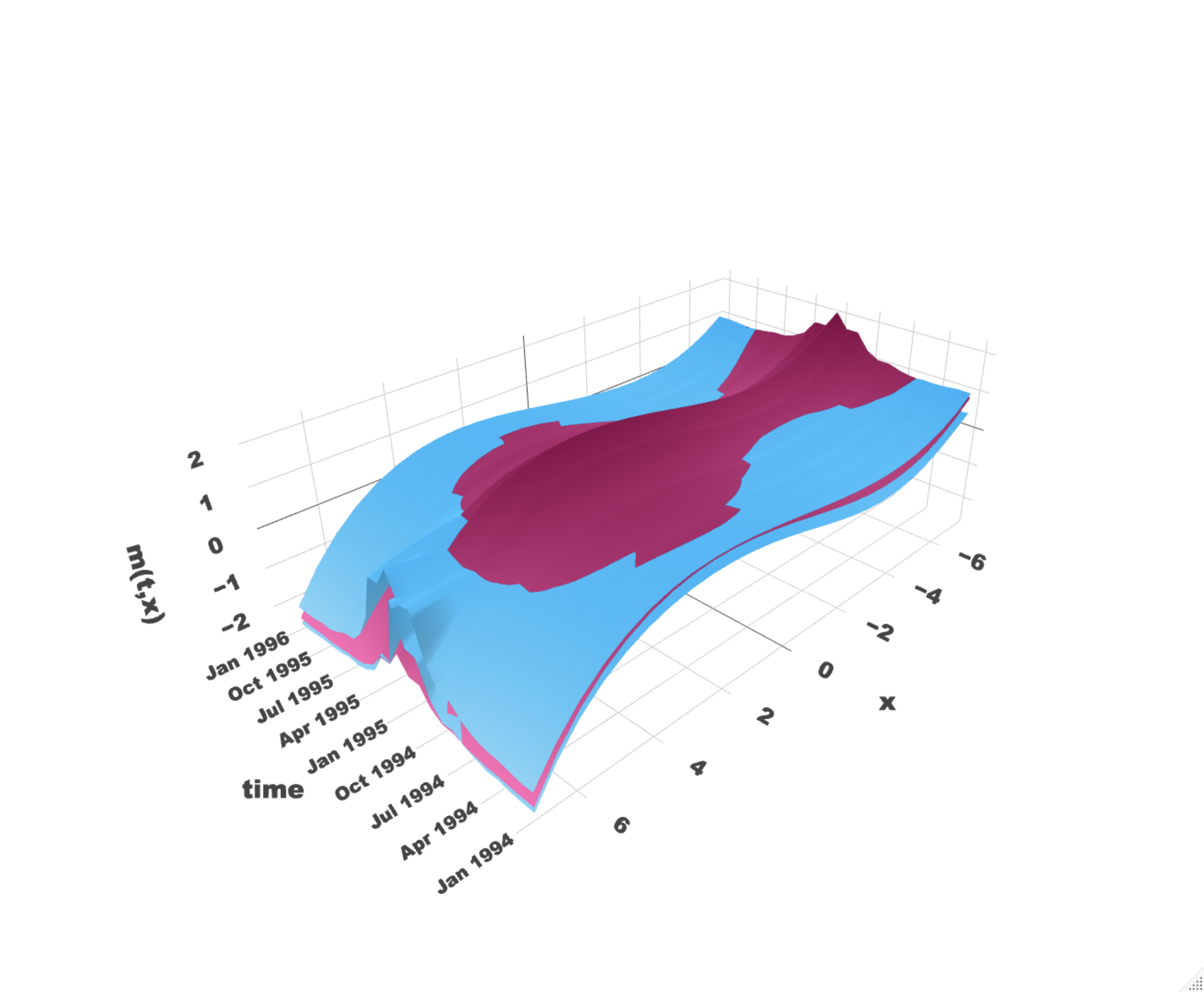}
\end{subfigure}
\begin{subfigure}{0.34\textwidth}
\includegraphics[width=6cm,height=5.4cm]{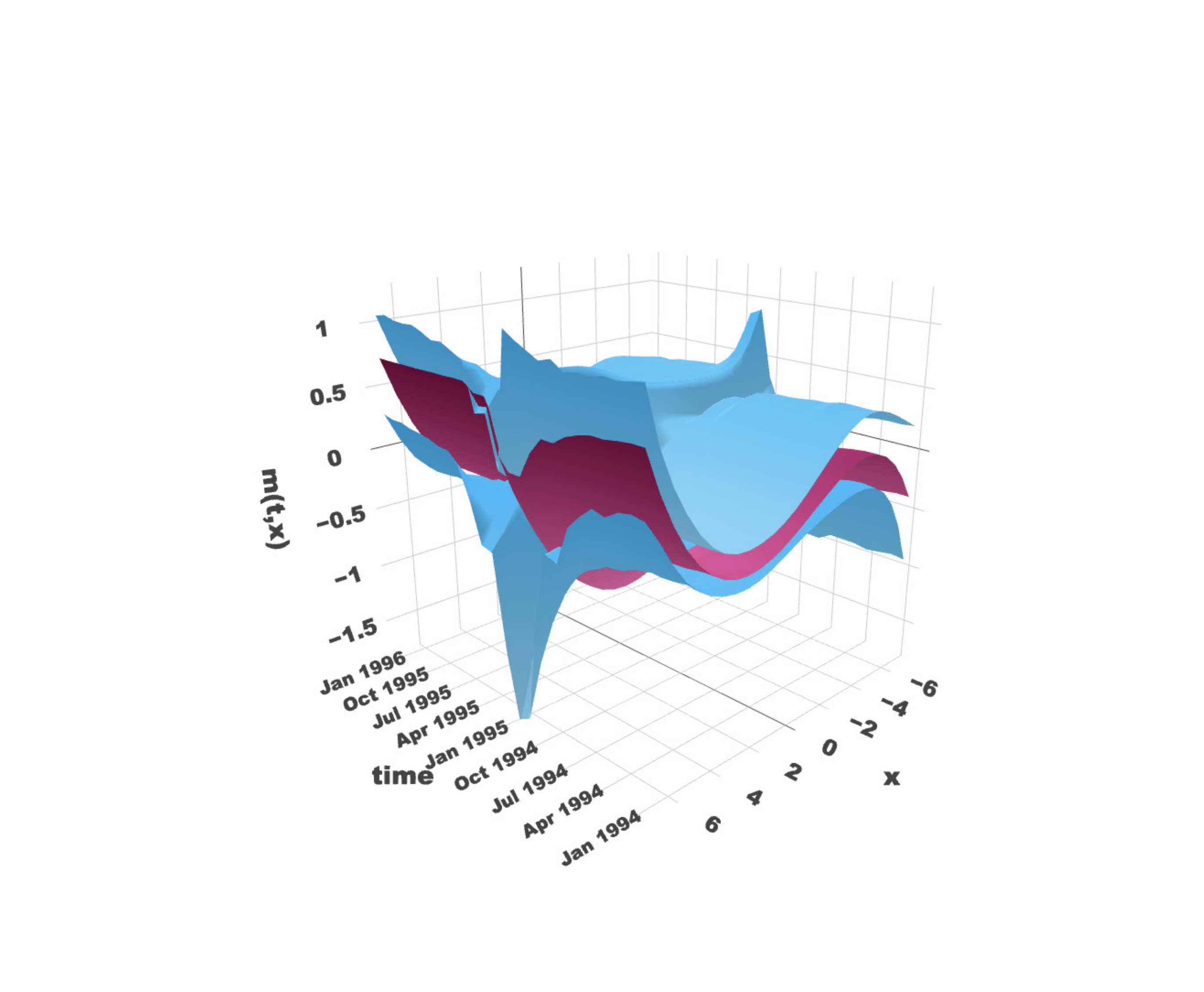}
\end{subfigure}
\begin{subfigure}{0.34\textwidth}
\includegraphics[width=6cm,height=5.4cm]{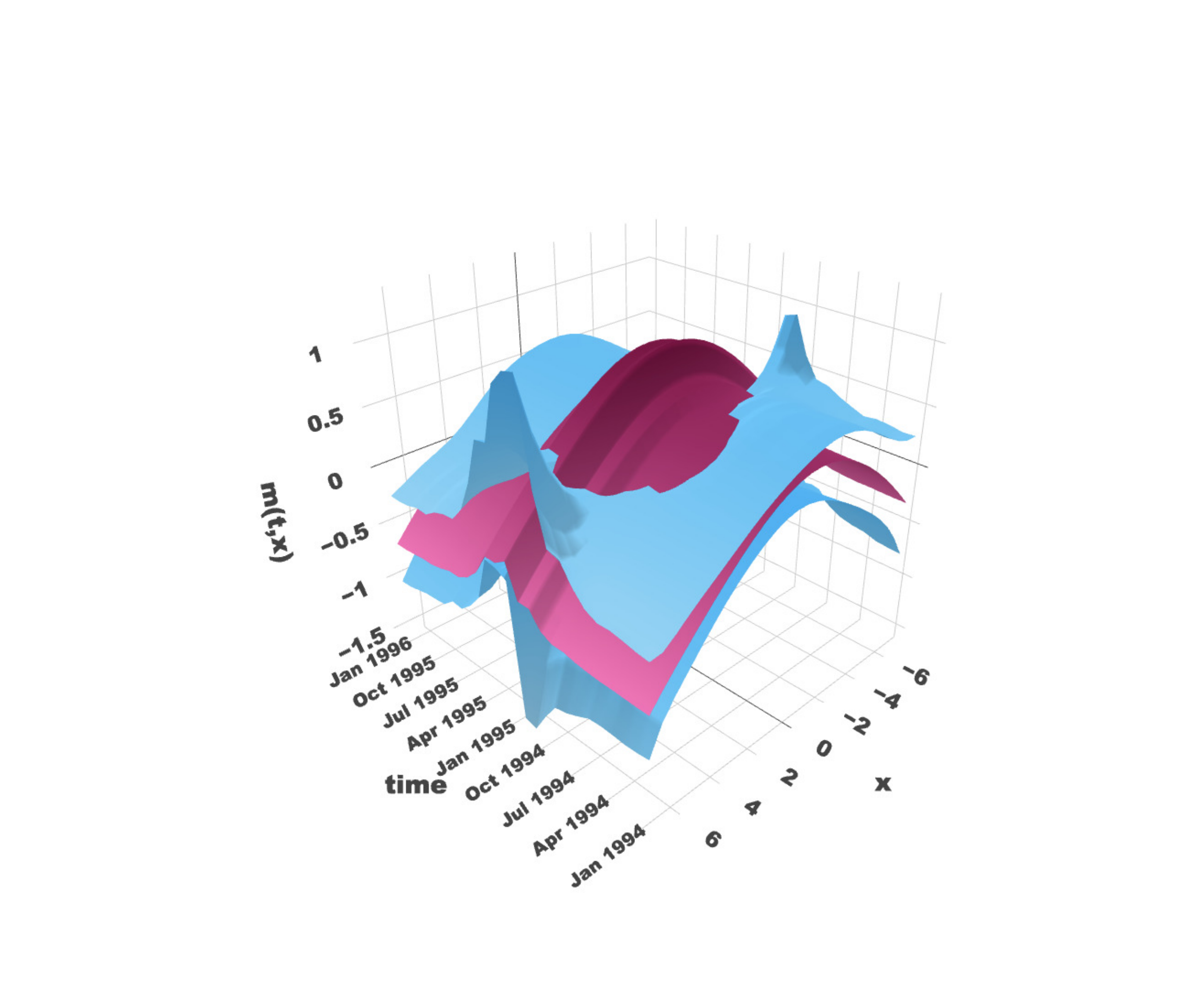}
\end{subfigure}
\vspace*{-0.4cm}
\caption{{ \footnotesize Time-homogeneity tests for all the covariates. From left to right, we have $j = 1, 2, 3$ (i.e., $ \operatorname{SO}_2 $, $ \operatorname{NO}_2 $, and dust). For each case, the estimated $m_j(t,x)$ (red surfaces) under the null hypothesis does not fit within the SCRs (blue surfaces), leading us to reject the null assumption of time-homogeneity.  
}  }
\label{fig_timevaryingrealdata}
\end{figure}

\begin{figure}[!ht]
\hspace*{-0.8cm}
\begin{subfigure}{0.34\textwidth}
\includegraphics[width=6cm,height=5.4cm]{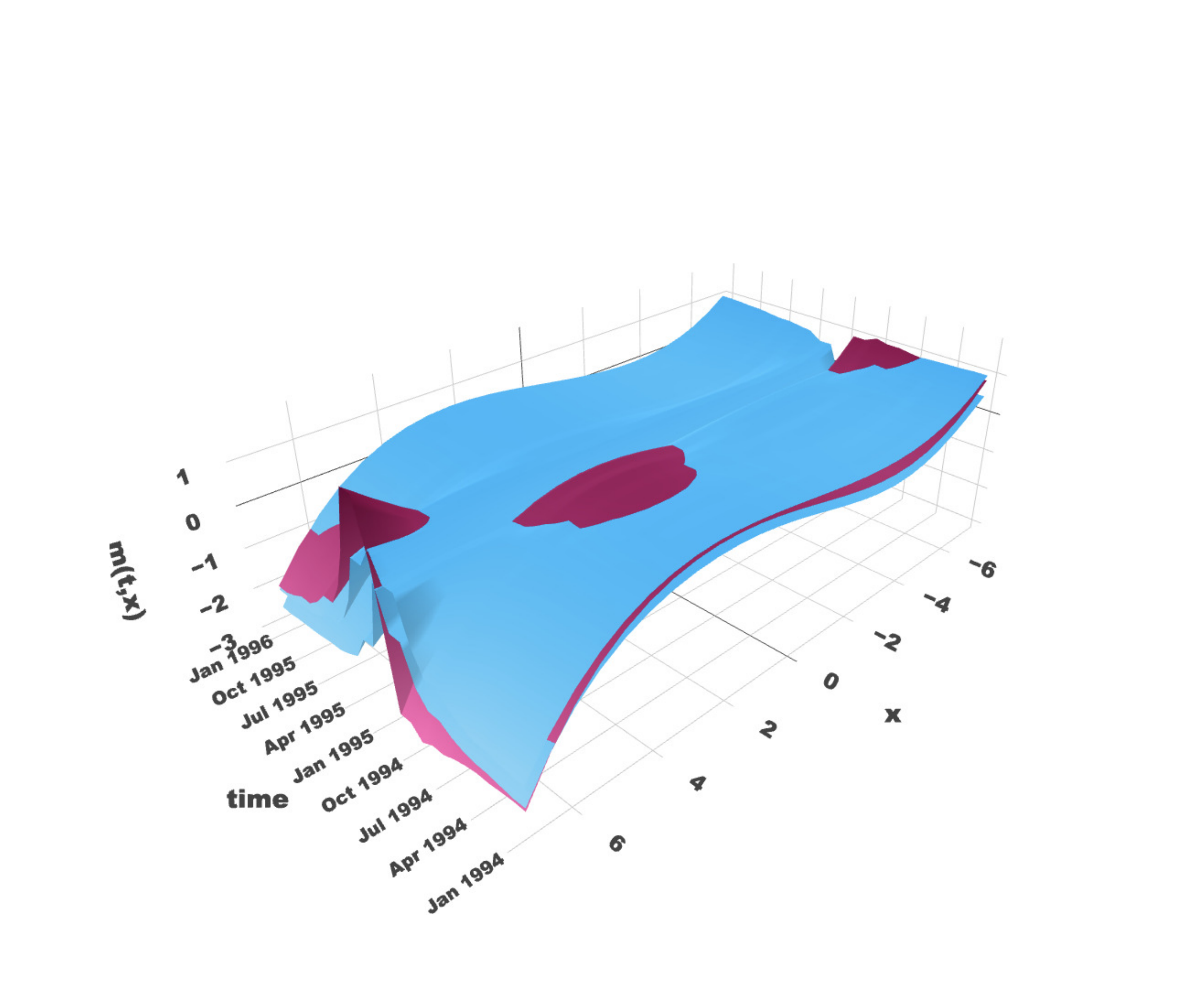}
\end{subfigure}
\begin{subfigure}{0.34\textwidth}
\includegraphics[width=6cm,height=5.4cm]{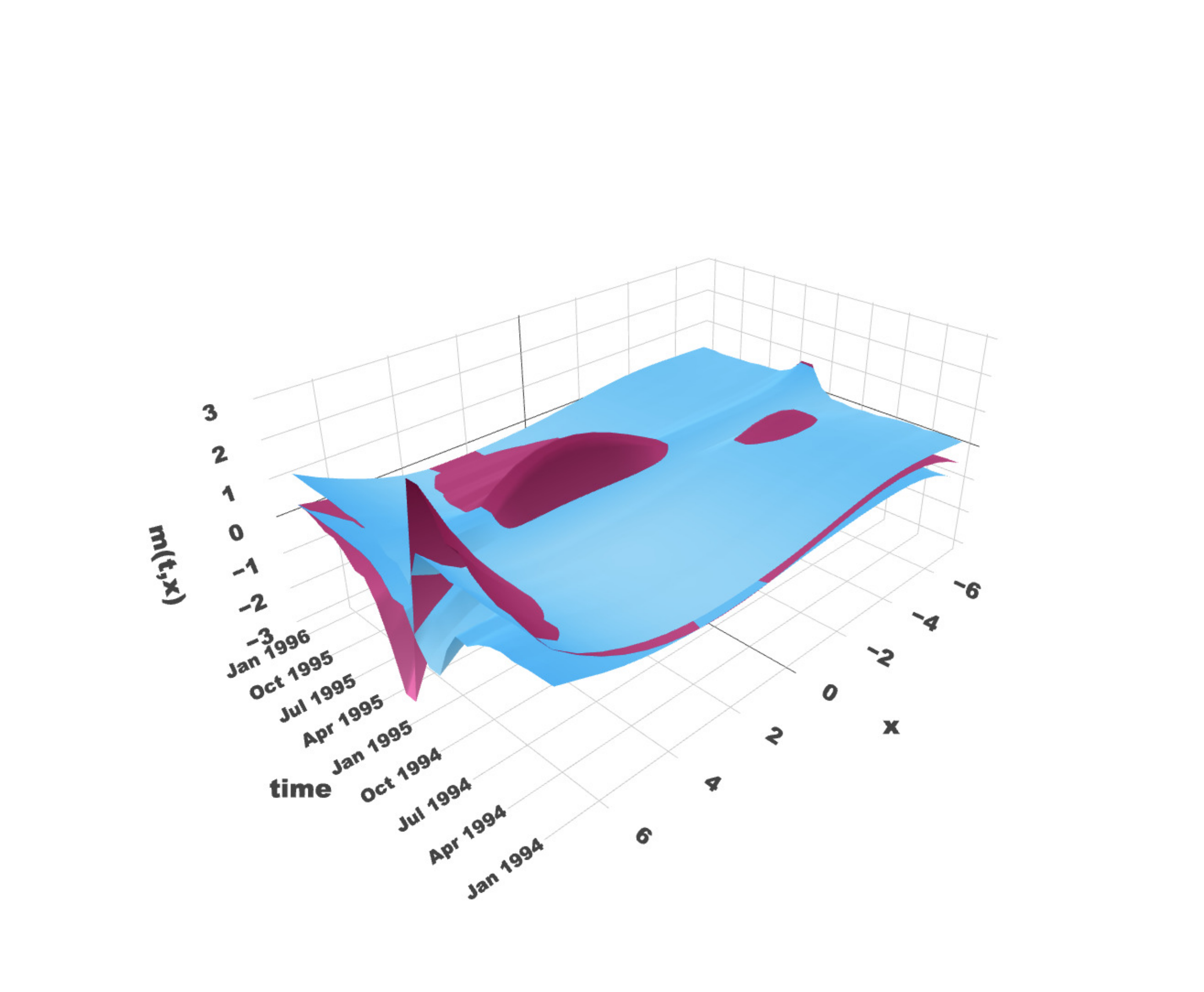}
\end{subfigure}
\begin{subfigure}{0.34\textwidth}
\includegraphics[width=6cm,height=5.4cm]{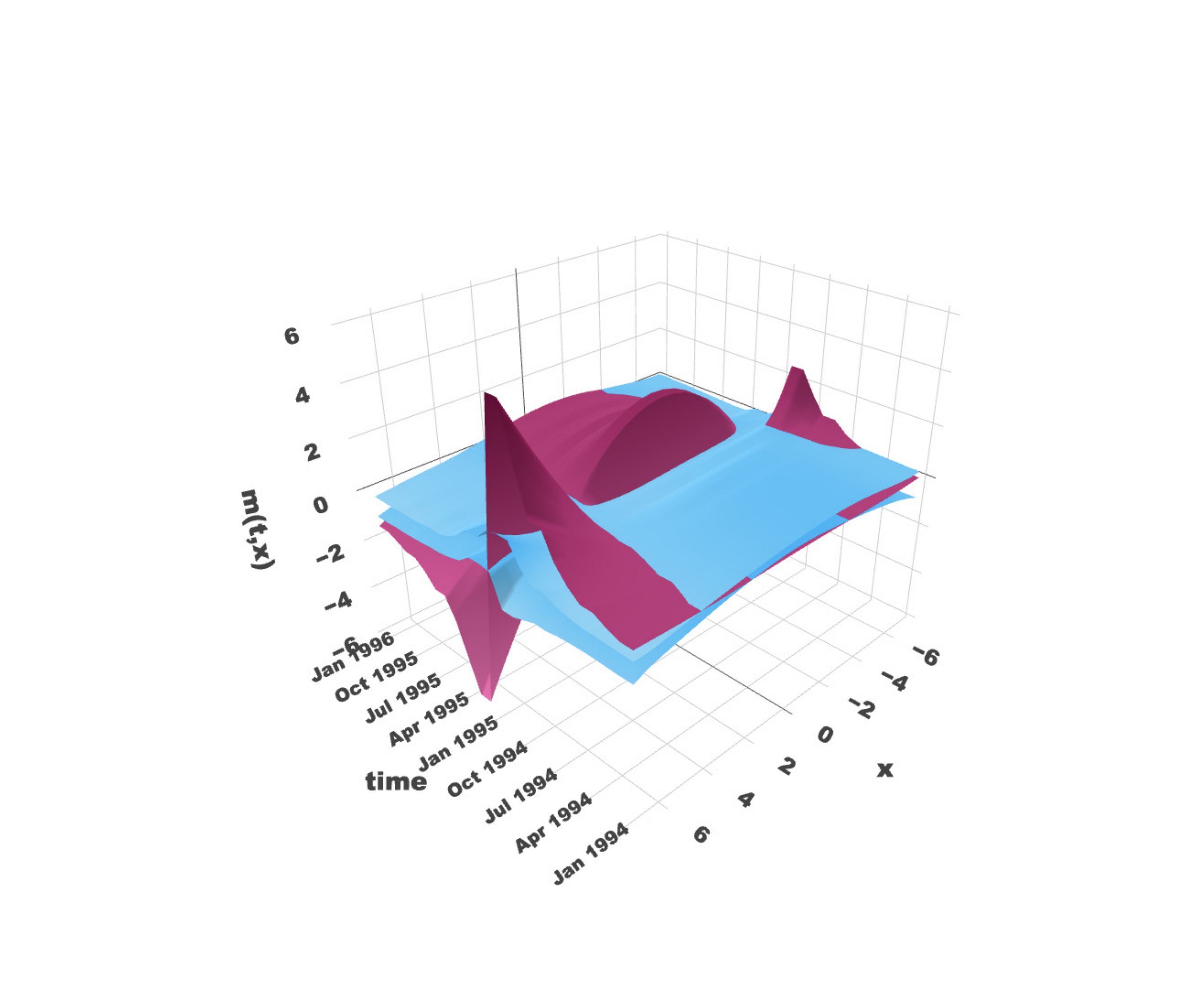}
\end{subfigure}
\vspace*{-0.4cm}
\caption{{ \footnotesize Separability tests for all the covariates. From left to right, we have $j = 1, 2, 3$ (i.e., $ \operatorname{SO}_2 $, $ \operatorname{NO}_2 $, and dust). For each case, the estimated $m_j(t,x)$ (red surfaces) under the null hypothesis does not fit within the SCRs (blue surfaces), leading us to reject the null assumption of separability.  
}  }
\label{fig_separablerealdata}
\end{figure}

\begin{figure}[!ht]
\hspace*{-0.8cm}
\begin{subfigure}{0.34\textwidth}
\includegraphics[width=6cm,height=5.4cm]{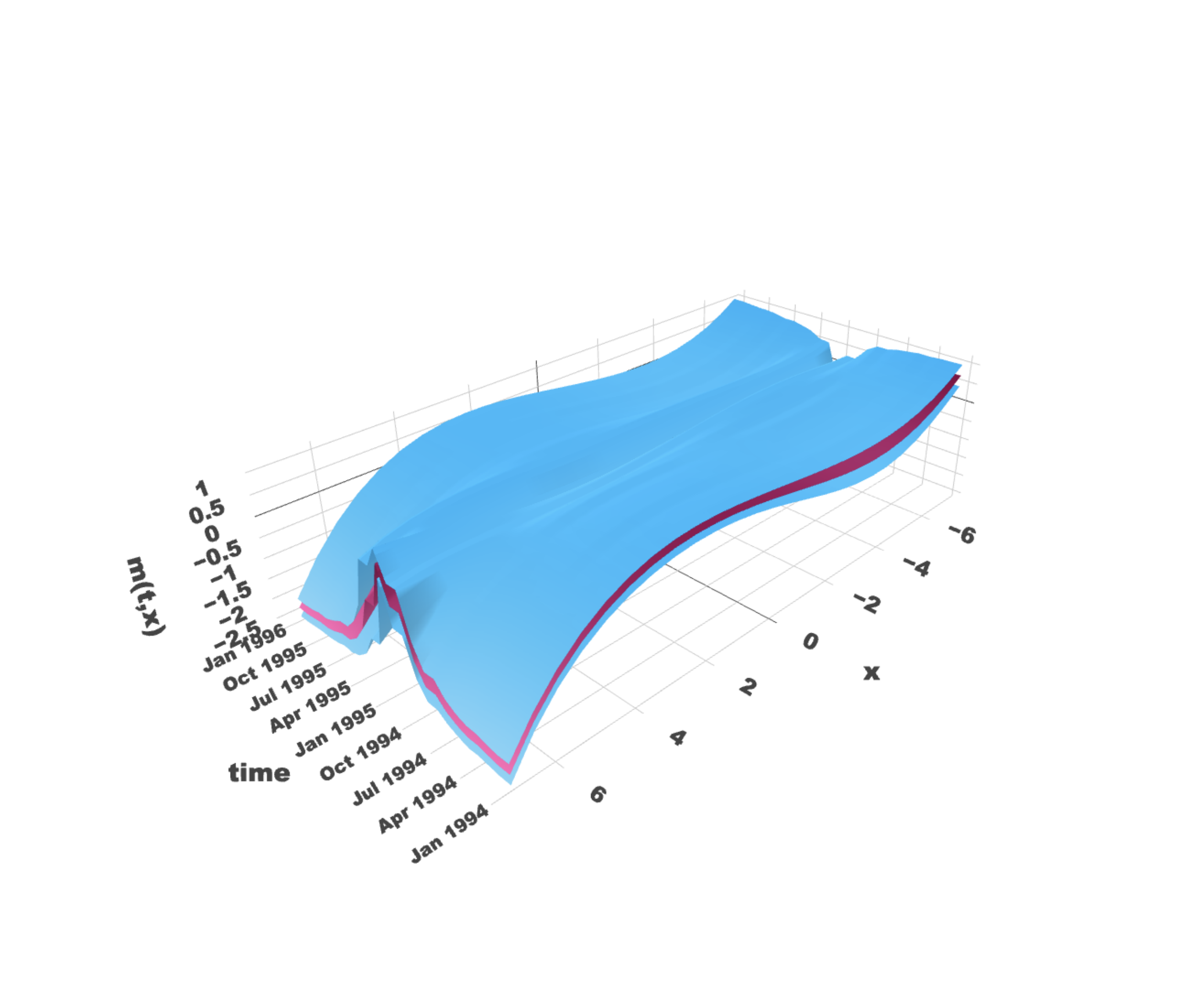}
\end{subfigure}
\begin{subfigure}{0.34\textwidth}
\includegraphics[width=6cm,height=5.4cm]{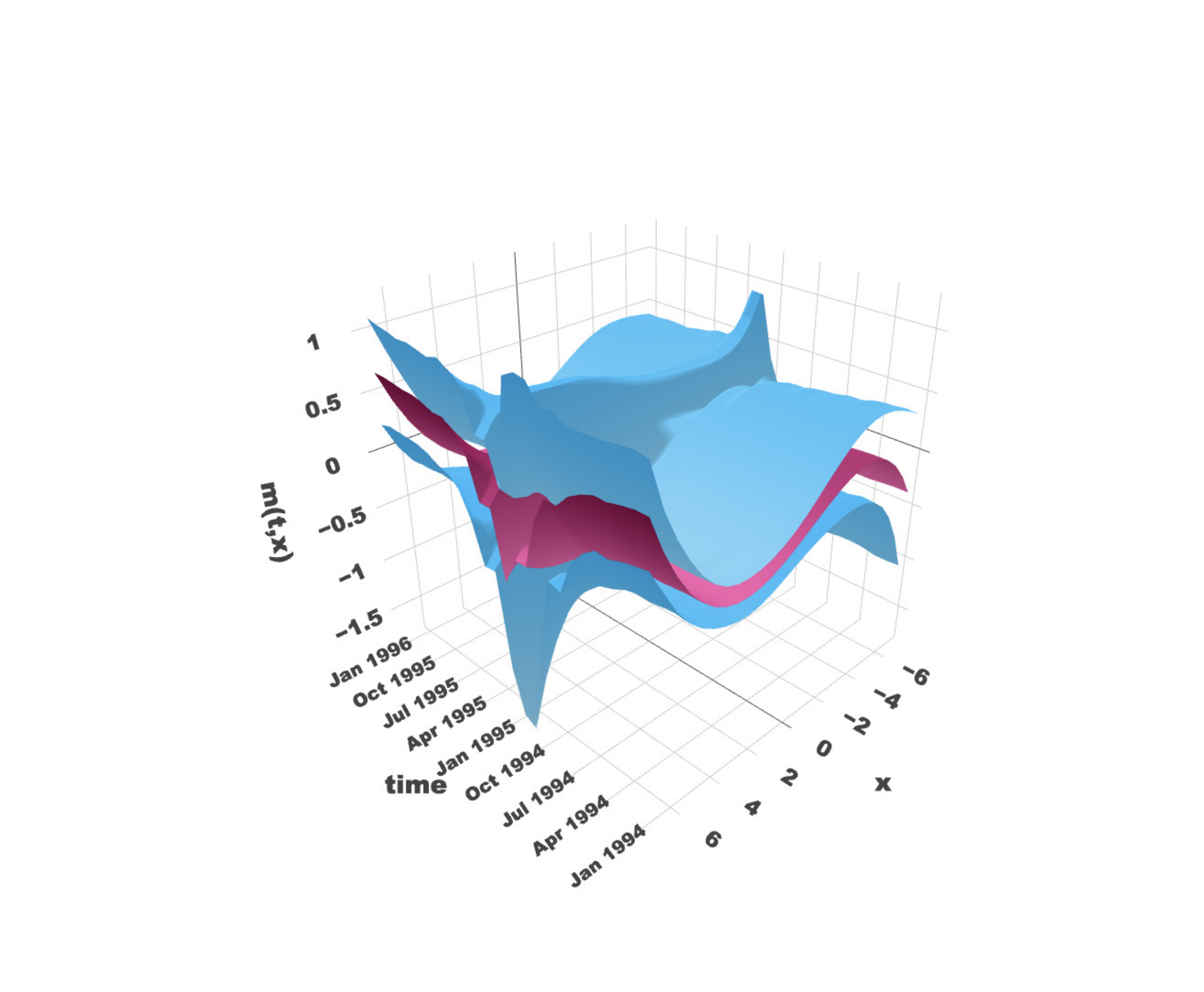}
\end{subfigure}
\begin{subfigure}{0.34\textwidth}
\includegraphics[width=6cm,height=5.4cm]{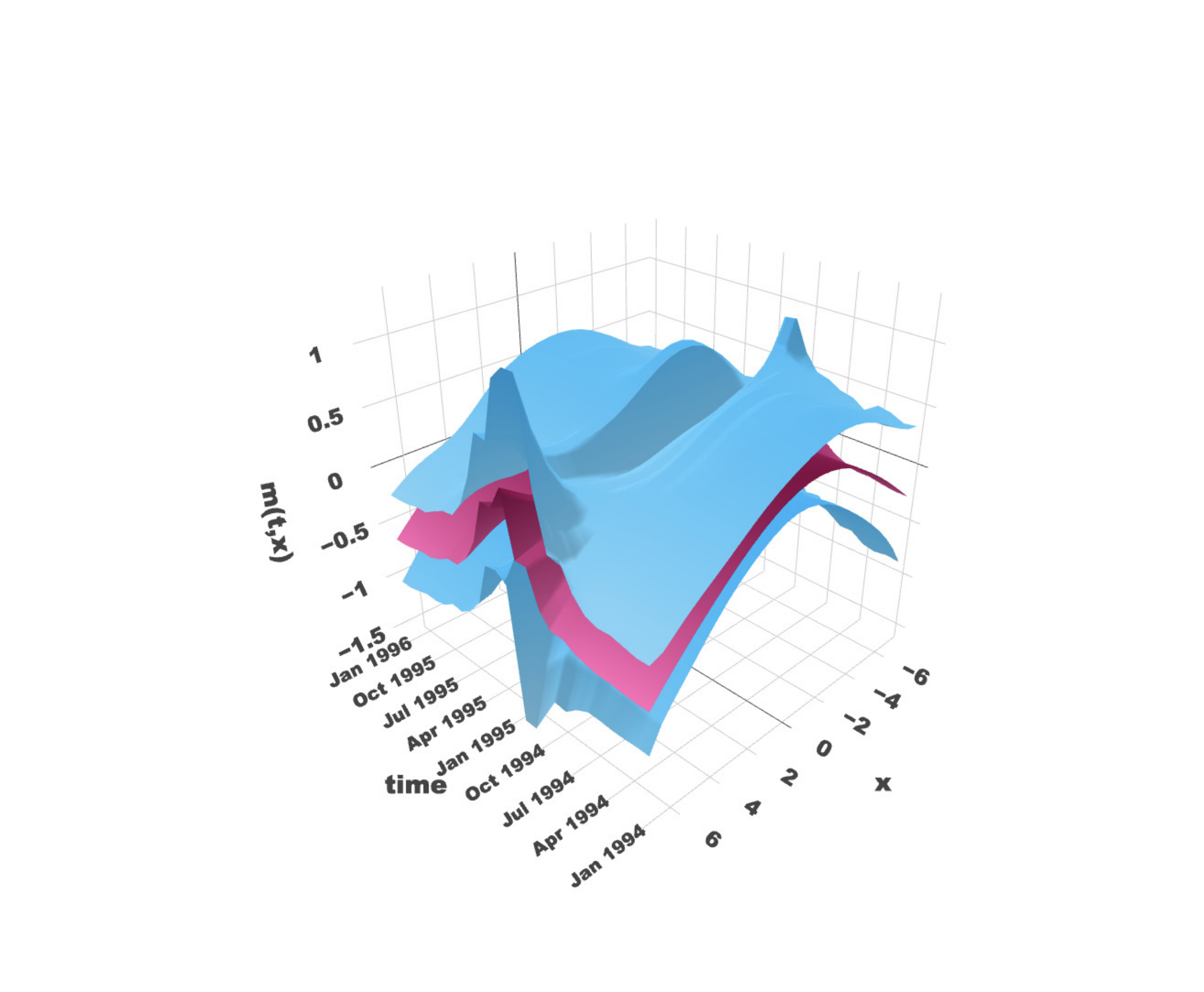}
\end{subfigure}
\vspace*{-0.4cm}
\caption{{ \footnotesize SCR estimation using our proposed model and method. } }
\label{fig_ourestimation}
\end{figure}

%

\subsection{Additional real data analysis}\label{supp_additionalrealdata}
In this subsection, we apply our method to study the shape of the monthly risk premium for S$\&$P 500 index as in \cite[Section 6.2]{CSW}. Let $\mu_t$ and $\sigma_t$ be the conditional mean and conditional volatility of the excess return of the market portfolio, respectively. Our goal is to uncover how $\mu_t$ relates to $\sigma_t.$ Especially, we want to understand the functional relation that $\mu_t=m(t, \sigma_t^2)$ for some unknown function $m.$ Such a problem has been studied extensively in the literature of financial economics. In the seminal work \cite{MR441271}, Merton modeled that $\mu_t=\gamma \sigma_t^2$ for some constant $\gamma$ representing the risk aversion of the agent. Recently, it has been argued in \cite{BOLLERSLEV2013409,DUKE, CHAIEB2021669, GHYSELS2014118, 10.1093/rfs/hhaa009} that a general nonlinear and time-varying function will facilitate the modeling and interpretation of the shape of the market risk premium. 

Very recently, in \cite{CSW}, the authors modeled the relationship between $\mu_t$ and $\sigma_t$ using $\mu_t=\rho(t) g(\sigma_t)$ for some unknown functions $\rho$ and $g.$ They estimated these functions using kernel methods. They found their estimation were consistent with real observations empirically and statistically analysis theoretically. Moreover, they provided some insights for the shape of the monthly risk premium for S$\&$P 500 Index based on their analysis. However, they did not justify why a separable structure was valid for modeling the risk premium. In what follows, we use our Algorithm \ref{alg:boostrapping} to justify the separability assumption, i.e., testing (\ref{eq_hotestseparability}). 

We follow the setting of \cite[Section 6.2]{CSW}. In the notation of our model (\ref{eq:model}), $m_0(t)=0$ and $r=1$ so that we do not have the issue of identifiability. We set $Y_t$ to be $\mu_t=r_{mt}-r_{ft}$ which is the excess return on the S$\&$P 500 Index calculated as the difference between the monthly continuously compounded cumulative return on the index minus the monthly return on 30-day Treasury Bills. For the conditional volatility   $\sigma_t,$ we used the realized volatility (RV) measure from Oxford Man Realized Library (\url{https://realized.oxford-man.ox.ac.uk/}). To be more precise, let $RV_t$ denote the daily annualized RV during the $t$th month. Then we obtain one-month-ahead RV forecast $\mathbb{E}_{t-1}(RV_t)$ using the HAR-RV model as in \cite{10.1093/jjfinec/nbp001}, which is our $X_t.$ Our data covers the period between 31 January 2001 and 31 December 2018. Under the nominal level $0.05,$ using the Daubechies-9 basis functions, we find that the $p$-value is $0.38$ so we can conclude that the separable structural assumption is reasonable for modeling the risk premium. This supports the analysis of \cite{CSW}.   

Finally, we estimate the functions  $\rho(\cdot)$ and $g(\cdot)$ using our sieve estimators. Due to separability, as discussed in Example \ref{exam_seperabletest}, we can estimate them separability instead of using the hierarchical sieve basis functions. We also provide the simultaneous confidence bands for these functions. In Figure \ref{fig_functionsplot} blow, we report the results. We find that when the RV is fixed, the risk premium changes in a nonlinear way of time which suggests the existence of a time-varying risk aversion. Moreover, for the $g(\cdot)$ function, it is more flat and seems to have a monotone pattern.  Our findings are consistent with those in \cite{CSW}; see Figure 7 and the discussion therein. 

\begin{figure}[!ht]
\hspace*{-1cm}
\begin{subfigure}{0.5\textwidth}
\includegraphics[width=6.4cm,height=5.5cm]{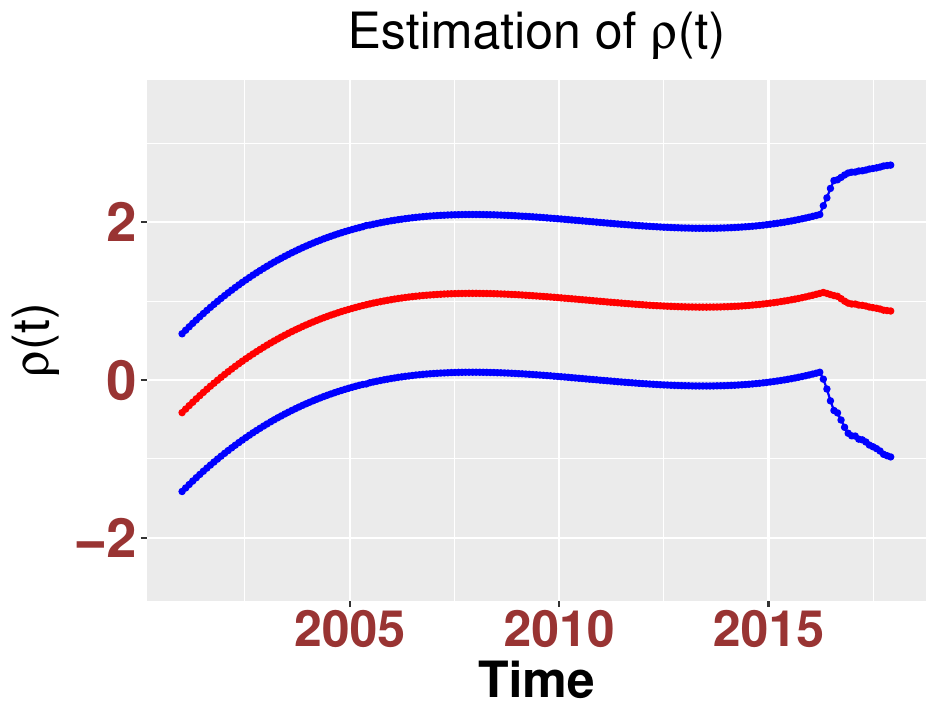}
\end{subfigure}
\hspace*{0.7cm}
\begin{subfigure}{0.5\textwidth}
\includegraphics[width=6.4cm,height=5.5cm]{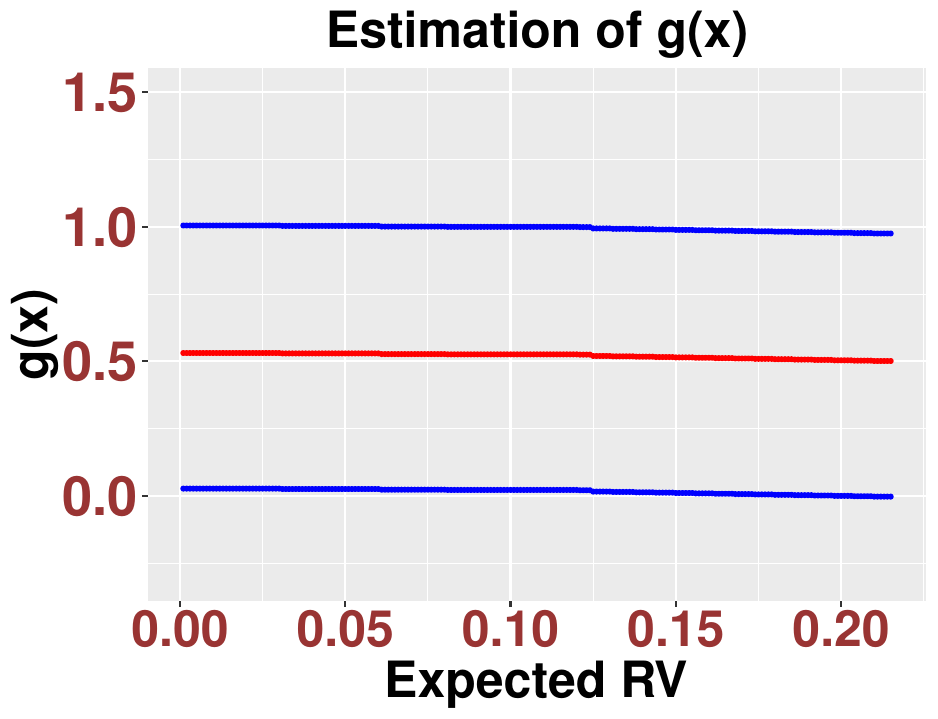}
\end{subfigure}
\caption{{ \footnotesize Estimated $\rho(t)$ and $g(x)$. The red lines are the estimated $\rho$ and $g$ and the blue lines the corresponding  simultaneous 95 $\%$ confidence bands constructed using the multiplier bootstrap. }  }
\label{fig_functionsplot}
\end{figure}

\section{Technical proofs}\label{sec_techinicalproof}

In this section, we provide the main technical proofs. 

\subsection{Uniform consistency of the proposed estimators: proof of Theorem \ref{thm_consistency}}\label{sec_consistencyproof}

In this subsection, we prove Theorem \ref{thm_consistency}. Till the end of the paper, for a positive definite $H,$ we denote $\|H \|_{\op}$ as the operator norm of $H,$ i.e., the largest eigenvalue of $H.$ 

\begin{proof}[\bf Proof of Theorem \ref{thm_consistency}] Since $r$ is finite, without loss of generality, we assume $r=1.$ For general $r,$ the proof can be modified verbatim with only additional notional complicatedness. Consequently, we can omit the subscript $j,$ so that $m \equiv m_{1}, X_i \equiv X_{1,i}, c \equiv c_1, d \equiv d_1, \chi(t) \equiv \chi_1(t)$ and $\widehat{m} \equiv \widehat{m}_{1}.$ In light of the identifiability assumption (\ref{eq_identiassum}) and the treatment above (\ref{eq_meanchisquare}),   the starting point of our proof is the following decomposition {
\begin{align}\label{eq_initialdecomposition}
|\widehat{m}(t,x)-m(t,x)| & \leq |\widehat{m}^*_{c,d}(t,x)-m_{c,d}(t,x)|+|\chi(t)-\widehat{\chi}(t)| \\
&+|m_{c,d}(t,x)-m(t,x)|.  \nonumber
\end{align} }  
Since the last term of the right-hand side of (\ref{eq_initialdecomposition}) can be bounded using Proposition \ref{thm_approximation}, it suffices to control the first two terms. We start with the first one.

Using the definitions of  (\ref{eq_firststeptruncation1}), (\ref{eq_firststeptruncation}) and (\ref{eq_proposedestimator}), by Cauchy-Schwarz inequality, we have that {
\begin{equation}\label{eq_ll1}
|\widehat{m}^*_{c,d}(t,x)-m_{c,d}(t,x)| \leq  | \widehat{\bm{\beta}}_1-\bm{\beta}_1 | \zeta,  
\end{equation}}
where $\zeta$ is defined in (\ref{eq_defnxic}).
Since $| \widehat{\bm{\beta}}_1-\bm{\beta}_1| \leq | \widehat{\bm{\beta}}-\bm{\beta} |,$ the rest of the proof leaves to control $\widehat{\bm{\beta}}-\bm{\beta}.$ Note that we denote $\bm{Y}=(Y_i)_{1 \leq i \leq n} \in \mathbb{R}^{n}$ and $\bm{\epsilon}=(\epsilon_i)_{1 \leq i \leq n} \in \mathbb{R}^{n}.$  Recall the notations around (\ref{eq_designmatrix}) and (\ref{eq_betaolsform}) and that
\begin{equation}\label{eq_betadifferenceexpression}
\widehat{\bm{\beta}}-\bm{\beta}=(W^\top W)^{-1} W^\top \bm{\epsilon}. 
\end{equation}
Therefore, we need to control $| \widehat{\bm{\beta}}-\bm{\beta} |$ which satisfies {
\begin{equation}\label{eq_betabound1}
| \widehat{\bm{\beta}}-\bm{\beta} | \leq \left| \frac{1}{\lambda_{\mathsf{p}}((n^{-1} W^\top W))} \right| | n^{-1} W^\top \bm{\epsilon} |, 
\end{equation} 
where we recall the notation (\ref{eq_defnp}). }


First, we establish the convergence results for the matrix $W^\top W. $ Recall that $W$ contains two parts $W_0$ and $W_1.$ Denote $W_{0(1)}(i)$ as the $i$th column of $W_{0(1)}.$ It suffices to analyze the terms $W_0(i)^\top W_0(j), W_0(i)^\top W_1(j)$ and $W_1(i)^\top W_1(j).$  Without loss of generality, we focus our discussion on the terms $$ W_1(1)^\top W_1(1) ,W_1(1)^\top W_2(1), W_0(1)^\top W_0(1), W_0(1)^\top W_0(2), W_0(1)^\top W_1(1). $$ 
The general cases of $i,j$ can be handled similarly.

We start with the term $W_1^\top(1) W_1(1).$ Let $W_1(1)=(\mathsf{w}_{12}, \cdots, \mathsf{w}_{1n})^\top.$ Recall (\ref{eq_designmatrix}). We have 
\begin{equation*}
\mathsf{w}_{1k}=\phi_1(t_k) \varphi_1(X_{k}), \ 1 \leq k \leq n.
\end{equation*}
Consequently, we have that
\begin{equation*}
\frac{1}{n} W_1(1)^\top W_1(1)=\frac{1}{n} \sum_{k=1}^n \phi_1(t_k)^2 \varphi_1(X_{k})^2.  
\end{equation*}
Note that by a discussion similar to Lemma \ref{lem_locallystationaryform}, we can show that $\mathfrak{h}_k:=\phi_1(t_k)^2 \varphi_1(X_{k})^2, 1 \leq k \leq n,$ is a sequence of locally stationary time series whose physical dependence measure satisfies that 
\begin{equation*}
\delta_{\mathfrak{h}}(j,q) \leq C \xi^2 \varsigma^2 j^{-\tau}, \ \text{for some constant} \ C>0,
\end{equation*}
where we again recall the notations in (\ref{eq_defnxic}).

Together with (1) of Lemma \ref{lem_concentration}, we obtain that
\begin{equation}\label{eq_firstusedimportantequation}
\left\|\frac{1}{n} W_1^\top(1) W_1(1)-\frac{1}{n}\sum_{k=1}^n \phi_1(t_k)^2 \mathbb{E} \varphi_1(X_{k})^2  \right\|_q \leq \frac{C \xi^2 \varsigma^2}{\sqrt{n}},
\end{equation}
where we used the assumption that $q>2.$ 
Moreover, by (2) of Lemma \ref{lem_concentration}, we have that 
\begin{equation}\label{eq_cccc}
\left| \frac{1}{n}\sum_{k=1}^n \phi_1(t_k)^2 \mathbb{E} \varphi_1(X_{k})^2 -\frac{1}{n} \sum_{k=1}^n \Pi_{11}(t_k) \right| \leq C \xi^2 \varsigma^2 n^{-1+\frac{2}{\tau+1}},
\end{equation}
where $\Pi_{11}(t)$ is the first entry of $\Pi(t)$ defined in (\ref{eq_longrunwitht}).
Further, by Lemma \ref{lem_intergralappoximation}, we 
have that
\begin{equation}\label{eq_ccc2}
\left|\frac{1}{n} \sum_{k=1}^n  \Pi_{11}(t_k)-\int_0^1  \Pi_{11}(t) \dd t \right|=\OO(n^{-2}). 
\end{equation}
We point out that $\int_0^1 \Pi_{11}(t)\dd t=\Pi_{11}$ which is the first entry of $\Pi$ as in (\ref{eq_longruncovariancematrix}). Combining with the above arguments, we conclude that 
\begin{equation*}
\left\| \frac{1}{n} W_1^\top(1) W_1(1)- \Pi_{11} \right\|_q \leq C \left( \frac{\varsigma^2\xi^2}{\sqrt{n}}+\frac{\varsigma^2 \xi^2 n^{\frac{2}{\tau+1}}}{n}\right).
\end{equation*}
Similarly, we can show that 
\begin{equation*}
\left\| \frac{1}{n} W_1^\top(1) W_2(1)- \Pi_{12} \right\|_q \leq C \left( \frac{\varsigma^2 \xi^2}{\sqrt{n}}+\frac{\varsigma^2 \xi^2 n^{\frac{2}{\tau+1}}}{n}\right).
\end{equation*}

Moreover, for $W_0(1)^\top W_0(1),$ we recall that $W_0(1)=(\phi_1(t_1), \cdots, \phi_1(t_n)).$ Consequently, we have that 
\begin{equation*}
\frac{1}{n} W_0(1)^\top W_0(1)=\frac{1}{n} \sum_{k=1}^n \phi_1^2(t_k).
\end{equation*}
By the orthonormality of the basis functions and Lemma \ref{lem_intergralappoximation}, we obtain that 
\begin{equation*}
\left|\frac{1}{n} W_0(1)^\top W_0(1)-1 \right|=\mathrm{O}(n^{-2}). 
\end{equation*}
Similarly, we can prove that 
\begin{equation*}
\left|\frac{1}{n} W_0(1)^\top W_0(2) \right|=\mathrm{O}(n^{-2}). 
\end{equation*}

Finally, we discuss $W_0(1)^\top W_1(1).$ By definition, we have that
\begin{equation*}
\frac{1}{n} W_0(1)^\top W_1(1)=\frac{1}{n} \sum_{k=1}^n \phi_1^2(t_k) \varphi_1(X_k). 
\end{equation*}
Then according to a discussion similar to (\ref{eq_firstusedimportantequation})--(\ref{eq_ccc2}), we readily obtain that 
\begin{equation*}
\left\| \frac{1}{n} W_0(1)^\top W_1(1)-\overline{\Pi}_{1 (c_0+1)} \right\|_q \leq C \left(\frac{ \xi^2 \varsigma}{\sqrt{n}}+\xi^2 \varsigma n^{-1+\frac{2}{\tau+1}} \right). 
\end{equation*}

Recall (\ref{eq_Pibar}). 
Together the above results with Lemma \ref{lem_circle}, we conclude that 
\begin{equation}\label{eq_consistencyconvergency}
\left\| \frac{1}{n} W^\top W-\overline{\Pi} \right\|_{\op} =\OO_{\mathbb{P}} \left( \mathsf{p}\left( \frac{\varsigma^2 \xi^2}{\sqrt{n}}+\frac{\xi^2 \varsigma^2 n^{\frac{2}{\tau+1}}}{n}\right) \right),
\end{equation}
where we used the fact that $q>2$ and recall that $\mathsf{p}$ is defined in (\ref{eq_defnp}). Together with Assumption \ref{assum_updc} and (\ref{eq_parameterassumption}), we find that {
\begin{equation}\label{eq_boundone}
\left | \frac{1}{\lambda_{\min}(n(W^\top W)^{-1})} \right |=\OO_{\mathbb{P}}(1). 
\end{equation}}

Second, we control the error term $\frac{W^\top \bm{\epsilon}}{n}.$ {Without loss of generality, we focus on the first entry of the matrix $\frac{W_1^\top \bm{\epsilon}}{n}$ as follows
\begin{equation}\label{eq_entrywiseexpansion}
\frac{\left[W_1^\top \bm{\epsilon} \right]_{1}}{n}=\frac{1}{n}\sum_{k=1}^n \phi_1(t_k) \varphi_1(X_{k}) \epsilon_{k}.
\end{equation}} 
By Lemma \ref{lem_locallystationaryform}, (1) of Lemma \ref{lem_concentration} and $q>2$, we conclude that 
\begin{equation*}
\left \|\frac{\left[W_1^\top \bm{\epsilon} \right]_{1}}{n} \right \|^2 \leq \frac{C\xi^2 \varsigma^2 }{n}.
\end{equation*} 
Consequently, for some constant $C>0,$ we have that 
\begin{equation}\label{eq_boundtwo}
\left \|\frac{W^\top \bm{\epsilon}}{n} \right \| \leq  \frac{C \xi \varsigma \sqrt{\mathsf{p}}}{\sqrt{n}}.   
\end{equation}
In summary,  by (\ref{eq_boundone}) and (\ref{eq_boundtwo}), in view of (\ref{eq_betabound1}), we obtain that {
\begin{equation}\label{eq_betabound}
|\widehat{\bm{\beta}}-\bm{\beta} |=\OO_{\mathbb{P}}(\xi \varsigma \sqrt{\frac{\mathsf{p}}{n}}). 
\end{equation} }
Together with (\ref{eq_ll1}), we have completed the proof of the control of the first part of the right-hand side of (\ref{eq_initialdecomposition}) . 

Then we proceed to control $|\chi(t)-\widehat{\chi}(t) |.$  Recall (\ref{eq_definitionvartheta}) and  $\bm{f}=\{\phi_{\ell_1}(t)\}_{\ell_1 \leq i \leq c} \otimes \{\vartheta_{\ell_2}(t)\}_{1 \leq \ell_2 \leq d}  \in \mathbb{R}^{cd}.$ Moreover, for the estimators $\{\widehat{\vartheta}_{\ell_2}(t)\},$ we denote
$\widehat{\bm{f}}=\{\phi_{\ell_1}(t)\}_{1 \leq \ell_1 \leq c} \otimes \{\widehat{\vartheta}_{\ell_2}(t)\}_{1 \leq \ell_2 \leq d}  \in \mathbb{R}^{cd}.$ By definition, we have that 
\begin{equation}\label{eq_widehatfnotation}
\chi(t)-\widehat{\chi}(t)=\bm{f}^\top (\bm{\beta}_1-\widehat{\bm{\beta}}_1)+\widehat{\bm{\beta}}_1^\top(\bm{f}-\widehat{\bm{f}}). 
\end{equation}

On the one hand, for some constant $C>0,$ using (\ref{eq_betabound}), we have that {
\begin{equation*}
| \bm{f}^\top (\bm{\beta}_1-\widehat{\bm{\beta}}_1) | \leq |\bm{f}| | \bm{\beta}-\widehat{\bm{\beta}}|=\OO_{\mathbb{P}}(|\bm{f}| \xi \varsigma \sqrt{\frac{\mathsf{p}}{n}}).  
\end{equation*}}
Recall $\phib(t)=\{\phi_{\ell_1}(t)\} \in \mathbb{R}^c.$ For $|\bm{f}|,$ using the property of Kronecker product, we have that
\begin{align}\label{eq_bbbbbb}
|\bm{f}|=|\phib| |\bm{\vartheta}(t)  | \leq \gamma  \iota,
\end{align}
where $\bm{\vartheta}(t) \in \mathbb{R}^d$ is the collection of $\{\vartheta_{\ell_2}\}$ and we recall (\ref{eq_defnxic}) again. Consequently, we have that{ 
\begin{equation*}
| \bm{f}^\top (\bm{\beta}_1-\widehat{\bm{\beta}}_1) | =\OO_{\mathbb{P}}(\xi \varsigma \gamma \iota \sqrt{\frac{\mathsf{p}}{n}}).  
\end{equation*}}

On the other hand, when $n$ is sufficiently large, together with (\ref{eq_betabound}),   we have that {
\begin{equation*}
|\widehat{\bm{\beta}}_1^\top(\bm{f}-\widehat{\bm{f}})|=\OO_{\mathbb{P}}(| \bm{\beta}_1| | \bm{f}-\widehat{\bm{f}}|).   
\end{equation*} }
Using the property of Kronecker product, we have that 
\begin{equation*}
| \bm{f}-\widehat{\bm{f}} |=|\phib| | \bm{\vartheta}-\widehat{\bm{\vartheta}} | \leq \gamma \sqrt{\sum_{\ell=1}^d \|\vartheta_\ell-\widehat{\vartheta}_\ell\|^2}. 
\end{equation*}
Using Lemma \ref{lem_locallystationaryform}, by a discussion similar to (\ref{eq_betabound}) (see \cite{DZ, DZ2} for more details), we readily have that { 
\begin{equation*}
|\vartheta_{\ell}(t)-\widehat{\vartheta}_{\ell}(t)|=\OO_{\mathbb{P}}\left(c_\ell^{-\mathsf{n}_{\ell}}+\gamma_\ell \xi \varsigma \sqrt{\frac{c_\ell}{n}}  \right). 
\end{equation*}}
This yields that {
\begin{equation}\label{eq_bbbbbb11111}
| \bm{f}-\widehat{\bm{f}} |=\OO_{\mathbb{P}}\left(\gamma \sqrt{\sum_{\ell=1}^d \left( c_\ell^{-n_\ell}+\gamma_\ell \xi \varsigma \sqrt{\frac{c_\ell}{n}} \right)^2} \right). 
\end{equation} }

Combining all the above results, one can readily obtain that {
\begin{equation}\label{eq_boundnnnnnn}
| \chi(t)-\widehat{\chi}(t) |=\mathrm{O}_{\mathbb{P}} \left( \xi \varsigma \gamma \iota \sqrt{\frac{\mathsf{p}}{n}} +\gamma \sqrt{\sum_{\ell=1}^d \left( c_\ell^{-n_\ell}+\gamma_\ell \xi \varsigma \sqrt{\frac{c_\ell}{n}} \right)^2}\right).
\end{equation}}
This concludes our proof of (\ref{eq_rate}). 

Finally, we deal with the time-varying intercept part. Similar to (\ref{eq_initialdecomposition}), we have that {
\begin{align*}
| \widehat{m}_0(t)-m_0(t) | \leq | \widehat{m}_{0,c}^*(t)-m_{0,c}(t) |+\left|  \widehat{\chi}(t)- \chi(t) \right|+|m_0(t)-m_{0,c}(t)|. 
\end{align*}}
Using (\ref{eq_betabound}) and a discussion similar to (\ref{eq_ll1}), we readily have that {
\begin{equation}\label{eq_adddadada}
 | \widehat{m}_{0,c}^*(t)-m_{0,c}(t) |=\mathrm{O}_{\mathbb{P}}\left( \gamma \xi \varsigma \sqrt{\frac{\mathsf{p}}{n}} \right). 
\end{equation} }
Together with (\ref{eq_boundnnnnnn}), we can then conclude the proof. 

\end{proof}

\subsection{{High dimensional point-wise Gaussian approximation and proof of Theorem \ref{thm_asymptoticdistribution}}}\label{sec_gassuianapproximation}
In this subsection, we prove Theorem \ref{thm_asymptoticdistribution}. The key ingredient is to establish the Gaussian approximation for the statistics $\mathsf{T}_{j}$ . The starting point is to rewrite the statistics $\mathsf{T}_{j}$ more explicitly in an affine form. We first prepare some notations. Without loss of generality, we only explain the proof for $r=1$ as  discussed in the beginning of the proof of Theorem \ref{thm_consistency} and omit the subscript $j$ in the sequel. Additionally, since our focus is on the regression functions, we assume $m_0(t) \equiv 0$. The general case can be handled in a similar manner; see Remark \ref{rem_meanzerodiscussions} below for more discussions. 

Recall (\ref{eq_widehatfnotation}). Denote 
\begin{equation*}
\widehat{\bm{r}}=\bm{b}-\widehat{\bm{f}}. 
\end{equation*} 
According to (\ref{eq_initialdecomposition}), under the assumption of (\ref{eq_parameterassumption}), Assumption \ref{assum_debiasassumption} and mean-zero assumption,
\begin{align} \label{eq_fundementalexpression}
\widehat{m}(t,x)-m(t,x)&= \bm{r}^\top (\widehat{\bm{\beta}}-\bm{\beta})+(\widehat{\bm{r}}^\top-\bm{r}^\top) \widehat{\bm{\beta}} \nonumber \\
& =\frac{1}{\sqrt{n}}\bm{r}^\top \Pi^{-1} (\frac{1}{\sqrt{n}} W^\top \bm{\epsilon})\left(1+\oo_{\mathbb{P}}(1)  \right),
\end{align}
where we recall (\ref{eq_verctorconstruction}), and we used (\ref{eq_betadifferenceexpression}) and  (\ref{eq_consistencyconvergency}). 
\begin{remark}\label{rem_meanzerodiscussions}
We note that if the mean is nonzero, (\ref{eq_fundementalexpression}) can be easily modified such that the arguments and discussions remain valid. More specifically, we have that  
\begin{align*} 
\widehat{m}(t,x)-m(t,x)&= \bm{r}^\top (\widehat{\bm{\beta}}_1-\bm{\beta}_1) \nonumber \\
& =\frac{1}{\sqrt{n}}\overline{\bm{r}}^\top \overline{\Pi}^{-1} (\frac{1}{\sqrt{n}} W^\top \bm{\epsilon})\left(1+\oo_{\mathbb{P}}(1)  \right),
\end{align*}
where $\overline{\Pi}$ is defined in (\ref{eq_Pibar}) and $\overline{\bm{r}}=(\bm{0}^\top_c \ \bm{r}^\top)^\top.$  Consequently, we can further write  
\begin{align*} 
\widehat{m}(t,x)-m(t,x)& =\frac{1}{\sqrt{n}}\overline{\bm{r}}^\top \overline{\Pi}^{-1} (\frac{1}{\sqrt{n}} W^\top \bm{\epsilon})\left(1+\oo_{\mathbb{P}}(1)  \right). 
\end{align*}
Then all the discussions below still apply. 
\end{remark}

Based on Remark \ref{rem_meanzerodiscussions}, in what follows, without loss of generality, we assume the time-varying intercept is always zero; see Remark \ref{rem_meannotzero} for more discussions on this point. Recall that the deterministic vector $\bm{l}$ and random vector $\bm{z}$ are defined as
\begin{equation}\label{eq_originaldefinition}
\bm{l}=\Pi^{-1} \bm{r}, \ \bm{z}=\frac{1}{\sqrt{n}} W^\top \bm{\epsilon}.
\end{equation}
In what follows, recall (\ref{eq_realU}), we denote $\bm{y}_i \in \mathbb{R}^p$ such that 
\begin{equation}\label{eq_decompositionkronecker}
\bm{y}_i=\widetilde{\mathbf{U}}(t_i, \mathcal{F}_i).
\end{equation}
Using (\ref{eq_designmatrix}), it is easy to see that  
\begin{equation}\label{eq_bmzgreatform}
\bm{z}=\frac{1}{\sqrt{n}}\sum_{i=1}^n \bm{y}_i.
\end{equation}
Therefore, in view of (\ref{eq_originaldefinition}),  the analysis of $\mathsf{T}$ reduces to study $\bm{l}^\top \bm{z}$ using (\ref{eq_bmzgreatform}) and (\ref{eq_decompositionkronecker}).

Based on the above discussion, we have seen that it suffices to establish the Gaussian approximation theory for the affine form $\bm{l}^\top \bm{z}.$ Before stating the Gaussian approximation results,  we first pause to record the covariance structure of $\bm{z}$ in Lemma \ref{lem_covarianceofz} below. It indicates that $\operatorname{Cov}(\bm{z})$ is close to $\Omega$ defined in (\ref{eq_longruncovariancematrix}). 
\begin{lemma}\label{lem_covarianceofz}
Let $\operatorname{Cov}(\bm{z})$ be the covariance matrix of $\bm{z}.$ Suppose the assumptions of Theorem \ref{thm_consistency} hold. Then we have that
\begin{equation*}
\|\operatorname{Cov}(\bm{z})-\Omega \|_{\op}=\OO\left(\frac{p\xi^2 \varsigma^2 n^{2/\tau}}{\sqrt{n}} +p\xi^2 \varsigma^2 n^{-1+\frac{2}{\tau+1}}\right).
\end{equation*}
\end{lemma}
\begin{proof}
Denote $\bm{z}=(z_1, \cdots, z_p)^\top.$ We control the error entrywisely and focus on $[\operatorname{Cov}(\bm{z})]_{11}=\operatorname{Var}(z_1)$. Recall (\ref{eq_ddd}).  Note that 
\begin{equation}\label{eq_z1form}
z_1=\frac{1}{\sqrt{n}}\sum_{k=1}^n \phi_1(t_k) \varphi_1(X_{k}) \epsilon_{k}.
\end{equation}
Using the notation (\ref{eq_ddd}) and Lemma \ref{lem_locallystationaryform}, we conclude that 
\begin{align}\label{eq_decompositionvariance}
\operatorname{Var}(z_1) &=\frac{1}{n} \sum_{k_1, k_2=1}^n \phi_1(t_{k_1}) \phi_1(t_{k_2}) \mathbb{E} u_{k_1 1} u_{k_2 1} \nonumber \\
&=\frac{1}{n}\sum_{k=1}^n \phi_1(t_{k})^2 \mathbb{E} u^2_{k 1}+\frac{1}{n} \sum_{k_1 \neq k_2}^n \phi_1(t_{k_1}) \phi_1(t_{k_2}) \mathbb{E} u_{k_1 1} u_{k_2 1}:=\mathsf{E}_1+\mathsf{E}_2,
\end{align}

First, by a discussion similar to (\ref{eq_cccc}) and (\ref{eq_ccc2}), we  find that
\begin{equation}\label{eq_e1finalerror}
\left| \mathsf{E}_1-\Omega_{11}\right|=\OO\left(\xi^2 \varsigma^2 n^{-1+\frac{2}{\tau+1}}\right).
\end{equation}
Second, for $\mathsf{E}_2,$ by Lemma \ref{lem_locallystationaryform} and (3) of Lemma \ref{lem_concentration}, we find that for some constant $C>0$ 
\begin{equation}\label{eq_hhahahahahhaha}
|\mathbb{E} u_{k_1 1} u_{k_2 1}| \leq C \varsigma |k_1-k_2|^{-\tau}. 
\end{equation}
Consequently, we obtain that for some constant $C>0$ {
\begin{align}\label{eq_e2part}
\left| \mathsf{E}_2 \right| &= \frac{1}{n} \left(  \sum_{k_1=1}^n \phi_1(t_{k_1}) \sum_{|k_1-k_2| > n^{2/\tau}}  \phi_1(t_{k_2})  \mathbb{E} u_{k_1 1} u_{k_2 1} + \sum_{k_1=1}^n \phi_1(t_{k_1}) \sum_{|k_1-k_2| \leq n^{2/\tau}, k_1 \neq k_2}  \phi_1(t_{k_2})  \mathbb{E} u_{k_1 1} u_{k_2 1}\right) \nonumber \\
&\leq C\frac{ \varsigma \xi^2}{n}+C\left|\frac{1}{n} \sum_{k_1=1}^n \phi_1(t_{k_1}) \sum_{|k_1-k_2| \leq n^{2/\tau}, k_1 \neq k_2}  \phi_1(t_{k_2})  \mathbb{E} u_{k_1 1} u_{k_2 1} \right|,
\end{align}
where in the second step we used (\ref{eq_hhahahahahhaha}).} For each fixed $k_2,$ by a discussion similar to Lemma \ref{lem_locallystationaryform}, we see that $\{u_{k_1 1} u_{k_2 1}\}$ is a locally stationary time series whose physical dependence measure satisfies that $\delta(i,q) \leq C \varsigma^2 i^{-\tau}.$ {Recall that $\phi_1(t_{k_2})=\mathrm{O}(\xi).$ Similarly, for fixed $k_1,$ $\{\sum_{|k_1-k_2| \leq n^{2/\tau}, k_1 \neq k_2}  \phi_1(t_{k_2}) u_{k_1 1} u_{k_2 1} \}$ is also a locally stationary time series whose dependence measure $\delta'(i,q)$ can be bounded by  $C \xi n^{2/\tau} \varsigma^2 i^{-\tau}.$} As $q>2,$ applying (1) of Lemma \ref{lem_concentration}, we readily see that 
{ 
\begin{align*}
\left|\frac{1}{n} \sum_{k_1=1}^n \phi_1(t_{k_1}) \sum_{|k_1-k_2| \leq n^{2/\tau}, k_1 \neq k_2}  \phi_1(t_{k_2})  \mathbb{E} u_{k_1 1} u_{k_2 1} \right| \leq \frac{C \varsigma^2 \xi^2 n^{2/\tau}}{\sqrt{n}}.
\end{align*}
}
Together with (\ref{eq_e2part}), we arrive at
\begin{equation}\label{eq_e2finalerror}
|\mathsf{E}_2| \leq C \left( \varsigma \frac{\xi^2}{n} +\frac{ \varsigma^2 \xi^2 n^{2/\tau}}{\sqrt{n}}\right).
\end{equation} 
By (\ref{eq_decompositionvariance}), (\ref{eq_e1finalerror}) and (\ref{eq_e2finalerror}), we obtain that 
\begin{equation*}
|\operatorname{Var}(z_1)-\Omega_{11}|=\OO\left( \frac{\varsigma \xi^2}{n} +\frac{ \varsigma^2 \xi^2 n^{2/\tau}}{\sqrt{n}} +\xi^2 \varsigma^2 n^{-1+\frac{2}{\tau+1}}\right).
\end{equation*}
The general term $\operatorname{Cov}(z_i, z_j)$ can be analyzed similarly. We can therefore conclude our proof using Lemma \ref{lem_circle}. 
\end{proof} 

\begin{remark}\label{rem_meannotzero}
When the intercept is nonzero, we should construct $\bm{y}_i \in \mathbb{R}^{c_0+p}$ as follows 
\begin{equation*}
\bm{y}_i=(D(t_i,\mathcal{F}_i) \otimes \phib_0(t_i)^\top, \widetilde{\mathbf{U}}(t_i, \mathcal{F}_i)^\top)^\top. 
\end{equation*}
Then all the discussion will follow. Especially, one has $\bm{z}=(z_{01}, \cdots, z_{0c_0}, z_1, \cdots, z_p)^\top.$ Note that 
\begin{equation*}
z_{01}=\frac{1}{\sqrt{n}}\sum_{k=1}^n \phi_1(t_k) \epsilon_k.
\end{equation*}  
Based on this, we can obtain a similar similar to Lemma \ref{lem_covarianceofz} by replacing $\Omega$ with $\overline{\Omega},$ and $p$ with $\mathsf{p}$ in the error bound, as well as the fact that 
\begin{equation*}
\operatorname{Cov}(z_{01}, z_1)=\frac{1}{n} \sum_{k_1=1}^n \sum_{k_2=1}^n \phi_1(t_{k_1}) \phi_1(t_{k_2}) \operatorname{Cov}(\varphi_1(X_k),\epsilon_k)=0,
\end{equation*}  
where we used the assumption below (\ref{eq:model}) that $\mathbb{E}(\epsilon_k|X_k)=0.$
\end{remark}

Next, we state the Gaussian approximation results for the affine form of $\bm{z}$. Consider a sequence of centered Gaussian random vectors $\{\bm{g}_i\}_{i=1}^n$ in $\mathbb{R}^p$ which preserve the covariance structure of $\{\widetilde{\mathbf{U}}(t_i, \mathcal{F}_i)\}$ and denote
\begin{equation}\label{eq_wgaussiandefinition}
\bm{v}=\frac{1}{\sqrt{n}} \sum_{i=1}^n \bm{g}_i. 
\end{equation}
Recall $\bm{l}$ in (\ref{eq_originaldefinition}) which involves both $t \in [0,1]$ and $x \in \mathbb{R}$. The Gaussian approximation result, Theorem \ref{thm_gaussianapproximationcase}, provides controls on the following Kolmogorov distance {for any given fixed $t$ and $x$
\begin{equation*}
\mathcal{K}(\bm{z}, \bm{v}):=\sup_{\mathsf{x} \in \mathbb{R}} \left| \mathbb{P}( \bm{z}^\top \bm{l} \leq \mathsf{x})-\mathbb{P}(  \bm{v}^\top \bm{l} \leq \mathsf{x}) \right|,
\end{equation*}  
where we recall from (\ref{eq_originaldefinition}) that $\bm{l} \equiv \bm{l}(t,x).$
}

\begin{theorem}\label{thm_gaussianapproximationcase}
Suppose Assumptions \ref{assum_models}--\ref{assum_updc} hold. Moreover, we assume that 
for some large values $m$ and $\hd,$ (\ref{eq_onebound}) holds. Then we have that 
\begin{equation*}
\mathcal{K}(\zb, \bm{v})=\OO(\Theta). 
\end{equation*}
\end{theorem}

Theorem \ref{thm_gaussianapproximationcase} establishes the {Gaussian approximation} for the affine form by controlling its Kolmogorov distance with its Gaussian counterpart. We point that (\ref{eq_onebound}) is a mild assumption and can be easily satisfied as discussed below Theorem \ref{thm_asymptoticdistribution}.  Before proving Theorem \ref{thm_gaussianapproximationcase}, we first show how it implies Theorem \ref{thm_asymptoticdistribution}.

\begin{proof}[\bf Proof of Theorem \ref{thm_asymptoticdistribution}] As before, without loss of generality, we assume $m_0(t) \equiv 0$ and focus on the case $r=1$ and omit the subscript $j$ for the statistic $\mathsf{T}$ and other related quantities. Under the assumption of (\ref{eq_onebound}), we find that Theorem \ref{thm_gaussianapproximationcase} holds for $\mathcal{K}(\bm{z}, \bm{v}).$ Together with   (\ref{eq_fundementalexpression}) and (\ref{eq_originaldefinition}), by Theorem \ref{thm_consistency} and the assumption of (\ref{eq_assumptionerrorreduce}), we find that for any fixed $t$ and $x$
\begin{equation}\label{eq_t1c}
\frac{\sqrt{n}\mathsf{T}}{\sqrt{\operatorname{Var}(\bm{v}^\top \bm{l})}} \simeq \mathcal{N}(0,1). 
\end{equation}
Since $\bm{v}$ is Gaussian, we have that
\begin{equation}\label{eq_t1cc}
\bm{l}^\top \bm{v} \sim \mathcal{N}(0, \bm{l}^\top \operatorname{Cov}(\bm{v}) \bm{l}). 
\end{equation}
By the construction of $\bm{v},$ the assumption (\ref{eq_onebound}) and  Lemma \ref{lem_covarianceofz}, we conclude that
\begin{equation}\label{eq_t1ccc}
\| \operatorname{Cov}(\bm{v})-\Omega \|_{\op}=\oo(1). 
\end{equation}
Recall (\ref{eq_defhtx}). We can conclude our proof of (\ref{eq_thmonepartone}) using (\ref{eq_t1c}), (\ref{eq_t1cc}) and (\ref{eq_t1ccc}).  
\end{proof}

Finally, we proceed to prove Theorem \ref{thm_gaussianapproximationcase}. Its proof relies on the device of $\fm$-dependent approximation, the technique of suitable truncation and Lemma \ref{lem_mvnapp} which states a Gaussian approximation result on the convex set.  We prepare  some notations in the beginning. For $\bm{y}_i, i \geq 1,$ in (\ref{eq_decompositionkronecker}) and any nonnegative integer $\fm \geq 0,$ we define the so-called $\fm$-approximation of $\bm{y}_i$ as
\begin{equation}\label{eq_dependencesequence}
\bm{y}_i^\Mt=\mathbb{E}(\widetilde{\mathbf{U}}(t_i, \mathcal{F}_i)|\sigma(\eta_{i-\fm}, \cdots, \eta_i)),  \ \fm \geq 0,
\end{equation}
where $\sigma(\eta_{i-\fm}, \cdots, \eta_i)$ is the natural sigma-algebra generated by the sequence of random variables.  
Corresponding to (\ref{eq_bmzgreatform}), we denote
\begin{equation*}
\bm{z}^\Mt=\frac{1}{\sqrt{n}}\sum_{i=1}^n \bm{y}_i^\Mt. 
\end{equation*}
Next, for a given truncation level $\hd>0,$ we define the truncated version of $\bm{y}_i^\Mt$ and $\bm{z}^\Mt$ following 
\begin{equation}\label{eq_truncation}
\overline{\mathbf{U}}^\Mt(t_i, \mathcal{F}_i)=\mathbf{U}^\Mt(t_i, \mathcal{F}_i) \mathbf{1}(\mathbf{U}^\Mt(t_i, \mathcal{F}_i) \leq \hd), \ \bar{\bm{y}}^\Mt_i=\overline{\widetilde{\mathbf{U}}}^\Mt(t_i, \mathcal{F}_i), \ \bar{\bm{z}}^\Mt=\frac{1}{\sqrt{n}}\sum_{i=1}^n \bar{\bm{y}}^\Mt_i, 
\end{equation}
where the operation of truncation is applied entrywisely and $\overline{\widetilde{\mathbf{U}}}^\Mt(t_i, \mathcal{F}_i)$ is constructed in the same way as (\ref{eq_realU}) using $\overline{\mathbf{U}}^\Mt(t_i, \mathcal{F}_i)$. Moreover, define $\{\bar{\bm{g}}^{\Mt}_i\}$ as the sequence of Gaussian random vectors which preserve the covariance structure of $\{\bar{\bm{y}}_i^\Mt\}$ whose Gaussian part is the same as $\bm{v}$ as in (\ref{eq_wgaussiandefinition}) and let
\begin{equation*}
\bar{\bm{v}}^\Mt=\frac{1}{\sqrt{n}}\sum_{i=1}^n \bar{\gb}_i^\Mt.
\end{equation*}


%
%
%

\begin{proof}[\bf Proof of Theorem \ref{thm_gaussianapproximationcase}] The starting point is the following triangle inequality 
{
\begin{align}\label{eq_decomposition}
\mathcal{K}(\zb, \bm{v})& =\sup_{ \mathsf{x} \in \mathbb{R}} \left|   \mathbb{P}(  \bm{z}^\top \bm{l} \leq \mathsf{x})-\mathbb{P}(  \bm{v}^\top \bm{l} \leq \mathsf{x})  \right| \nonumber \\
& \leq \mathcal{K}(\zb, \zb^\Mt)+\mathcal{K}(\zb^\Mt, \bar{\zb}^\Mt)+\mathcal{K}(\bar{\zb}^\Mt, \bar{\bm{v}}^\Mt)+\mathcal{K}(\bar{\bm{v}}^\Mt, \bm{v}).
\end{align}}
It suffices to control every term of the right-hand side of (\ref{eq_decomposition}). 
In what follows, we provide the detailed  arguments for the controls following a five-step strategy.

\vspace{3pt}
\noindent{\bf Step one. \ \underline{Show the closeness of $\operatorname{Cov}(\bm{z})$ and $\operatorname{Cov}(\bar{\zb}^\Mt).$}} Let $\zb=(z_1, \cdots, z_p)$, $\bar{\zb}^\Mt=(\bar{z}_1, \cdots, \bar{z}_p)$ and $\zb^\Mt=(z_1^\Mt, \cdots, z_p^\Mt).$ Note that
\begin{equation}\label{eq_tridecomposition}
\left\|\operatorname{Cov}(\bm{z})-\operatorname{Cov}(\bar{\zb}^\Mt) \right\|_{\op} \leq \left\| \operatorname{Cov}(\zb)-\operatorname{Cov}(\zb^\Mt) \right\|_{\op}+ \left\| \operatorname{Cov}(\bar{\zb}^\Mt)-\operatorname{Cov}(\zb^\Mt) \right\|_{\op},
\end{equation}
where we recall again that $\|\cdot \|_{\op}$ is the operator norm of the given positive definite matrix. 

First, we control the first term of the right-hand side of (\ref{eq_tridecomposition}). Observe that for all $1 \leq i \leq p,$
\begin{equation}\label{eq_variancedecomposition}
\operatorname{Var}(z_i)-\operatorname{Var}(z^\Mt_i)=\mathbb{E}(z_i)^2-\mathbb{E}(z^\Mt_i)^2+\mathbb{E}(z_i-z^\Mt_i) \mathbb{E}(z_i+z^\Mt_i). 
\end{equation} 
Without loss of generality, we focus on the case $i=1$ and still keep the subscript $i$ without causing any further confusion. By (\ref{eq_z1form}), Lemmas \ref{lem_locallystationaryform} and  \ref{lem_mdependent}, 
we readily obtain that for some constant $C>0,$
\begin{equation*}
\mathbb{E}(|z_i-z_i^\Mt|^q)^{2/q} \leq C \xi^2 \Theta^2_{\fm,q} \leq C \xi^2 \varsigma^2 \fm^{-2\tau+2},  
\end{equation*} 
where $\Theta_{\fm,q}=\sum_{k=\fm}^{\infty} \delta_u(k,q)$ (recall (\ref{eq_deltaxdefinition})). Consequently, as $q>2,$ by (1) of Lemma \ref{lem_collectionprobineq}, we readily obtain that for some constant $C_1>0$
\begin{equation}\label{eq_control}
\mathbb{E}|z_i-z^\Mt_i| \leq C_1 \xi \Theta_{\fm,q}, \  \mathbb{E}|z_i-z^\Mt_i|^2 \leq C_1 \xi^2 \Theta^2_{\fm,q}.
\end{equation} 
Note that (\ref{eq_variancedecomposition}) implies that 
\begin{equation}\label{eq_decompositionv}
\left| \operatorname{Var}(z_i)-\operatorname{Var}(z_i^\Mt)\right| \leq \sqrt{\mathbb{E}|z_i-z_i^\Mt|^2 \mathbb{E}|z_i+z_i^\Mt|^2}+\mathbb{E}|z_i-z_i^\Mt| \mathbb{E}|z_i+z_i^\Mt|,
\end{equation}
where we used (2) of Lemma \ref{lem_collectionprobineq}. Moreover, by Lemmas \ref{lem_locallystationaryform} and \ref{lem_concentration}, using the definition (\ref{eq_z1form}), we find that for some constant $C_2>0$
\begin{equation}\label{eq_bound}
\mathbb{E}|z_i|< C_2 \xi \varsigma, \ \mathbb{E}|z_i|^2< C_2 \xi^2 \varsigma^2.
\end{equation} 
Together with (\ref{eq_control}), we see that
\begin{equation*}
\mathbb{E}|z_i+z_i^\Mt| \leq C_2 \xi \varsigma+ C_1 \xi \Theta_{\fm,q}.
\end{equation*}
Furthermore, using (\ref{eq_bound}), we see that
\begin{align*}
 \mathbb{E}|z_i+z_i^\Mt|^2 & \leq 2 \mathbb{E}|z_i|^2+2 \mathbb{E}|z_i^\Mt|^2 \\
 & \leq 2C_2 \xi^2 \varsigma^2 +2\mathbb{E}|z_i|^2+2\mathbb{E}|z_i-z_i^\Mt| \mathbb{E} |z_i+z_i^\Mt| \\
 & \leq 4 C_2 \xi^2 \varsigma^2+2 C_1 \xi\Theta_{\fm,q}(C_2\xi \varsigma+C_1 \Theta_{\fm,q}).
\end{align*}
Together (\ref{eq_control}) and (\ref{eq_decompositionv}), we conclude that for some constant $C>0,$
\begin{equation*}
\left| \operatorname{Var}(z_i)-\operatorname{Var}(z_i^\Mt)\right| \leq C \xi^2 \varsigma \Theta_{\fm,q}. 
\end{equation*}
Similarly, we can show that for all $1 \leq i,j \leq p$
\begin{equation*}
\left| \operatorname{Cov}(z_i, z_j)-\operatorname{Cov}(z_i^\Mt, z_j^\Mt) \right| \leq C \xi^2 \varsigma \Theta_{\fm,q}. 
\end{equation*}
Consequently, by Lemma \ref{lem_circle}, we obtain that for some constant $C_1>0$
\begin{equation}\label{eq_partonecontrol}
\left\| \operatorname{Cov}(\zb)-\operatorname{Cov}(\zb^\Mt)  \right\|_{\op} \leq C p \xi^2 \varsigma \Theta_{\fm,q} \leq C_1 p \xi^2 \varsigma^2 \fm^{-\tau+1}, 
\end{equation}
where in the second inequality we used (\ref{eq_transferbound}). 

Second, we control the second term on the right-hand side of (\ref{eq_tridecomposition}). Recall (\ref{eq_truncation}). We point out that $\bar{z}^\Mt_i \neq z_i^\Mt \mathbf{1}(z_i \leq \hd)$ in general. For notational simplicity, we denote $\mathbf{U}^\Mt(t_i, \mathcal{F}_i)=(u_{i1}^\Mt, \cdots, u_{i,d}^\Mt).$ Using the construction (\ref{eq_dependencesequence}), with an argument similar to Lemma \ref{lem_locallystationaryform}, we can show that $\{\mathbf{U}^\Mt(t_i, \mathcal{F}_i)\}$ is a locally stationary time series whose physical dependence measure also satisfies $\delta(j,q) \leq C \varsigma j^{-\tau},$ for some constant $C>0$. Analogous to the discussion for the first term, we focus on the case $i=1$ and keep the subscript $i.$ Observe that
\begin{equation}\label{eq_firstorder}
|\mathbb{E}z_i^\Mt-\mathbb{E}\bar{z}_i^\Mt|=\left|\frac{1}{\sqrt{n}} \sum_{k=1}^n \phi_1(t_k) \mathbb{E}(u^\Mt_{k,i} \mathbf{1}(|u_{k,i}^\Mt|>\hd)) \right|.
\end{equation} 
Moreover, using Chebyshev's inequality (c.f. (3) of Lemma \ref{lem_collectionprobineq}), a discussion similar to (\ref{eq_bound}),  and the fact that 
\begin{equation*}
\mathbf{1}(|u_{k,i}^\Mt|>\hd) \leq \frac{|u_{k,i}^\Mt|^{q-1}}{\hd^{q-1}},
\end{equation*}
we find that for some constant $C>0,$
\begin{equation*}
|\mathbb{E}(u^\Mt_{k,i} \mathbf{1}(|u_{k,i}^\Mt|>\hd))| \leq C \xi \varsigma \hd^{-(q-1)}.  
\end{equation*}
Consequently, we see that 
\begin{equation}\label{eq_expectationcontrolone}
|\mathbb{E}z_i^\Mt-\mathbb{E}\bar{z}_i^\Mt| \leq C \sqrt{n} \xi^2 \varsigma  \hd^{-(q-1)}.
\end{equation}
Similarly, we can show that 
\begin{equation}\label{eq_secondorder}
\left| \mathbb{E}(z_1^\Mt)^2-\mathbb{E} (\bar{z}_1^\Mt)^2 \right| \leq C n \xi^4 \varsigma^2  \hd^{-(q-2)}.
\end{equation}
By (\ref{eq_firstorder}), (\ref{eq_secondorder}) and the fact $|\bar{u}_{k,i}^\Mt| \leq \hd,$ using a decomposition similar to (\ref{eq_variancedecomposition}), we conclude that for some
constant $C>0,$
\begin{equation*}
\left| \operatorname{Var}(\bar{z}^\Mt_i)-\operatorname{Var}(z_i^\Mt)\right| \leq C n \xi^4 \varsigma^2 \hd^{-(q-2)}. 
\end{equation*}
Similarly, we can show that for all $1 \leq i,j \leq p$
\begin{equation*}
\left| \operatorname{Cov}(\bar{z}^\Mt_i, \bar{z}^\Mt_j)-\operatorname{Cov}(z_i^\Mt, z_j^\Mt) \right| \leq  C n \xi^4 \varsigma^2 \hd^{-(q-2)}. 
\end{equation*}
Consequently, by Lemma \ref{lem_circle}, we obtain that
\begin{equation}\label{eq_parttwocontrol}
\left\| \operatorname{Cov}(\bar{\zb}^\Mt)-\operatorname{Cov}(\zb^\Mt)  \right\|_{\op} \leq C pn \xi^4 \varsigma^2 \hd^{-(q-2)}. 
\end{equation}

In summary, by (\ref{eq_tridecomposition}), (\ref{eq_partonecontrol}) and (\ref{eq_parttwocontrol}), we find that for some constant $C>0$
\begin{equation}\label{eq_covarianceclose}
\left\|\operatorname{Cov}(\bm{z})-\operatorname{Cov}(\bar{\zb}^\Mt) \right\|_{\op}  \leq C \varsigma^2 \left( p \xi^2 \fm^{-\tau+1}+ pn \xi^4 \hd^{-(q-2)} \right). 
\end{equation}
This completes Step one. Furthermore, since the assumption (\ref{eq_onebound}) ensures that the right-hand side of (\ref{eq_covarianceclose}) is of order $\mathrm{o}(1),$ we can conclude that the covariance matrices are close. 

\vspace{3pt}
\noindent {\bf Step two. \underline{Control the third term  $\mathcal{K}(\bar{\zb}^\Mt, \bar{\bm{v}}^\Mt).$}} In this step, we apply Lemma \ref{lem_mvnapp} to control the associated term $\bar{\bm{z}}^{\mathtt{M}}$. {Note that $\bm{l} \equiv \bm{l}(t,x)$.} For {$\mathsf{x} \in \mathbb{R},$} denote {
\begin{equation}\label{eq_defnset}
\mathsf{A}_{\mathsf{x}} \equiv \mathsf{A}_{\mathsf{x}}(t,x):=\left\{ \bm{q} \in \mathbb{R}^p: \bm{q}^\top \bm{l} \leq \mathsf{x}   \right\}. 
\end{equation} }
It is easy to see that {$\mathsf{A}_{\mathsf{x}}$} is a convex set in $\mathbb{R}^p$. Denote $\mathcal{A}$ as the collection of all the convex sets in $\mathbb{R}^p.$ Then we have that {
\begin{equation*}
\mathcal{K}(\bar{\zb}^\Mt, \bar{\bm{v}}^\Mt) \leq \sup_{\mathsf{A} \in \mathcal{A}} \left| \mathbb{P}(\bar{\zb}^\Mt \in \mathsf{A})-\mathbb{P}(\bar{\bm{v}}^\Mt \in \mathsf{A} ) \right|.
\end{equation*} }
Therefore, it suffices to bound the right-hand side of the above equation using Lemma \ref{lem_mvnapp}. 
We verify conditions of Lemma \ref{lem_mvnapp}. {Let $\mathcal{N}_{i, \fm}:=\{i-\fm, i-\fm+1, \cdots, i, i+1, \cdots, i+\fm\}$} and denote {
\begin{equation*}
N_i:=\mathcal{N}_{i,\fm}, \  N_{ij}:= \mathcal{N}_{i, \fm} \cup \mathcal{N}_{j, \fm}, \ N_{ijk}:=\mathcal{N}_{i, \fm} \cup \mathcal{N}_{j, \fm} \cup \mathcal{N}_{k, \fm}. 
\end{equation*}}
By the constructions (\ref{eq_dependencesequence}) and (\ref{eq_truncation}), using the property of conditional expectation (c.f. \cite[Example 4.1.7]{probbook}), it is easy to see that the conditions of Lemma \ref{lem_mvnapp} are satisfied such that{
\begin{equation}\label{eq_setup}
n_1=\fm+1, \ n_2=2(\fm+1), \ n_3=3(\fm+1), \ \beta=\sqrt{p} \mathrm{h}. 
\end{equation} }
By Lemma \ref{lem_mvnapp}, we have that for some constant $C>0$
 \begin{equation*}
\mathcal{K}(\bar{\zb}^\Mt, \bar{\bm{v}}^\Mt)  \leq C p^{7/4} n^{-1/2} \| \Sigma^{-1/2} \|^3 \mathrm{h}^3\fm\left(\fm+\frac{\fm}{p}\right),
\end{equation*}  
where $\Sigma=\operatorname{Cov}(\bar{\zb}^\Mt).$  Together with Assumption \ref{assum_updc}, Lemma \ref{lem_covarianceofz}, (\ref{eq_covarianceclose}) and the assumption of (\ref{eq_onebound}),  we conclude that for some constant $C>0,$ 
 \begin{equation}\label{eq_steptworesults}
\mathcal{K}(\bar{\zb}^\Mt, \bar{\bm{v}}^\Mt)  \leq C p^{7/4} n^{-1/2}  \mathrm{h}^3 \fm^2.  
\end{equation}
Moreover, since the assumption of  (\ref{eq_onebound}) ensures that $p^{7/4} n^{-1/2}  \mathrm{h}^3 \fm^2=\mathrm{o}(1),$ we conclude that $\bar{\zb}^\Mt$ is asymptotically Gaussian. This completes Step two. 

\vspace{3pt}
\noindent {\bf Step three. \underline{ Control the fourth term  $\mathcal{K}(\bar{\bm{v}}^\Mt, \bm{v}).$}}  Decompose that {
\begin{equation}\label{eq_probabilitydecomposition}
\mathcal{K}(\bar{\bm{v}}^\Mt, \bm{v})=\sup_{\mathsf{x} \in \mathbb{R}} \left| \mathbb{P}( \bm{v}^\top \bm{l} \leq \mathsf{x} )-\mathbb{P} ( \bm{v}^\top \bm{l} \leq \mathsf{x}+\mathtt{D}(\bar{\bm{v}}^\Mt, \bm{v})) \right|,
\end{equation}}
where $\mathtt{D}(\bar{\bm{v}}^\Mt, \bm{v}))$ is defined as {
\begin{equation}\label{eq_differencedefinition}
\mathtt{D}(\bar{\bm{v}}^\Mt, \bm{v})= (\bm{v}-\bar{\bm{v}}^\Mt)^\top \bm{l}. 
\end{equation}}
Let $\Sigma_0=\operatorname{Cov}(\bm{z})$ and $\Sigma=\operatorname{Cov}(\bar{\bm{z}}^\Mt).$ Moreover, let $\mathbf{f}$ be some $p$-dimensional standard Gaussian random vector.  By construction, we can further write { 
\begin{equation}\label{eq_ddddd}
\mathtt{D}(\bar{\bm{v}}^\Mt, \bm{v})=\mathbf{f}_1^\top \bm{l}, \ \mathbf{f}_1:=((\Sigma_0^{1/2}-\Sigma^{1/2})\mathbf{f}). 
\end{equation}}
Using Cauchy-Schwarz inequality, we find that for some constant $C>0$ 
\begin{equation*}
\left\| \mathtt{D}(\bar{\wb}^\Mt, \wb) \right\| \leq C \| \Sigma_0-\Sigma \|_{\op} \left(\mathbb{E} \| \mathbf{f} \|^2_2  \right)^{1/2}  | \bm{l} |,
\end{equation*}
where {$| \bm{l} |$ is the $L_2$ norm of the vector $\bm{l}.$ } By Bernstein’s concentration inequality (c.f. (4) of Lemma \ref{lem_collectionprobineq}), we find that  for some constant $C>0,$
\begin{equation}\label{oldb1}
(\mathbb{E}\| \mathbf{f} \|_2^2)^{1/2} \leq C \sqrt{p}. 
\end{equation}
{Denote {
\begin{equation}\label{oldb212121}
\mathrm{b}:=\sup_{x \in \mathbb{R}, t \in [0,1]} | \bm{l} |. 
\end{equation}}
Moreover, by the definition (\ref{eq_defnxic}), Assumption \ref{assum_updc} as well as (\ref{eq_bbbbbb}), we find that for some constant $C>0$, we have that 
 \begin{equation}\label{oldb2asadadas}
\mathrm{b}\leq C (\zeta+\gamma \iota). 
\end{equation}
}
Together with (\ref{eq_covarianceclose}), we conclude that {
\begin{equation}\label{eq_bounddbarww}
\|\mathtt{D}(\bar{\bm{v}}^\Mt, \bm{v})\| \leq   C\sqrt{p} \left( p \xi^2 \varsigma^2 \fm^{-\tau+1}+ pn \xi^4 \varsigma^2 \hd^{-(q-2)} \right)\mathrm{b}:=\mathrm{ab}.
\end{equation}}
Under the assumption of (\ref{eq_onebound}), we have that $\mathrm{ab}=\mathrm{o}(1).$ {Since $\bm{v}$ is a Gaussian random vector, we have
\begin{equation*}
 \bm{v}^\top \bm{l} \sim \mathcal{N}(0, \bm{l}^\top \operatorname{Cov}(\bm{v}) \bm{l}). 
\end{equation*}
For notional simplicity, we denote 
\begin{equation}\label{eq_simplifiednotationone}
\mathsf{v}:= \bm{v}^\top \bm{l}, \ \text{and} \ \mathrm{w}:=\sqrt{\bm{l}^\top \operatorname{Cov}(\bm{v}) \bm{l}}. 
\end{equation}}

Moreover, given some large constant $C_1>0,$ we denote the following events  {
\begin{align}\label{eq_allprobabilityevents}
\mathcal{A}_1:=\{ \bm{v}^\top \bm{l} \leq \mathsf{x}\}, \  \mathcal{A}_2:=\{ \bm{v}^\top \bm{l} \leq \mathsf{x}+\mathtt{D}(\bar{\bm{v}}^\Mt, \bm{v})\}, \ \mathcal{A}_3:=\{ |\mathtt{D}(\bar{\bm{v}}^\Mt, \bm{v})| \leq  (\mathrm{ab})^{2/3}\}, 
\end{align} 
}
We also denote $\mathcal{A}_3^c$ as the complement of $\mathcal{A}_3.$ We can then write { 
\begin{align}\label{eq_decompositionhh}
|\mathbb{P}(\bm{v}^\top \bm{l} \leq \mathsf{x} ) & -\mathbb{P} ( \bm{v}^\top \bm{l} \leq \mathsf{x}+\mathtt{D}(\bar{\bm{v}}^\Mt, \bm{v}))|=|\left[\mathbb{E}(\mathbf{1}(\mathcal{A}_1)\mathbf{1}(\mathcal{A}_3))-\mathbb{E}(\mathbf{1}(\mathcal{A}_2)\mathbf{1}(\mathcal{A}_3)) \right] \nonumber \\
&+\left[\mathbb{E}(\mathbf{1}(\mathcal{A}_1)\mathbf{1}(\mathcal{A}^c_3))-\mathbb{E}(\mathbf{1}(\mathcal{A}_2)\mathbf{1}(\mathcal{A}^c_3)) \right]|:=|\mathsf{P}_1+\mathsf{P}_2| \leq |\mathsf{P}_1|+|\mathsf{P}_2|. 
\end{align}}
For $\mathsf{P}_2,$ by Markov inequality, combining with (\ref{eq_bounddbarww}),  we have that {
\begin{equation}\label{eq_p2control}
|\mathsf{P}_2| \leq 2 \mathbb{P}(\mathcal{A}_3^c)=\mathrm{O}((\mathrm{ab})^{2/3})=\mathrm{O} \left( \left[\sqrt{p} (\zeta+\gamma \iota) \left( p \xi^2 \varsigma^2 \fm^{-\tau+1}+ pn \xi^4 \varsigma^2 \hd^{-(q-2)} \right)  \right]^{2/3} \right),   
\end{equation} 
where we used (\ref{oldb2asadadas}).

For $\mathsf{P}_1,$ by definition, we have that 
\begin{align}\label{eq_totallawofprobabilityexpansion}
|\mathsf{P}_1|=|\mathbb{E}(\mathbf{1}(\mathcal{A}_1^c)\mathbf{1}(\mathcal{A}_2) \mathbf{1}(\mathcal{A}_3))|.
\end{align}
Clearly, if $\mathtt{D}(\bar{\bm{v}}^\Mt, \bm{v}) \leq 0,$ $\mathsf{P}_1=0.$ It suffices to consider that $\mathtt{D}(\bar{\bm{v}}^\Mt, \bm{v})>0.$ In this setting, using the convention (\ref{eq_simplifiednotationone}),  we denote 
\begin{equation}\label{eq_oooooo1212121}
\mathcal{A}_{123,o}:=\{\mathsf{x}< \bm{v}^\top \bm{l} \leq \mathsf{x}+(\mathrm{ab})^{2/3}\},
\end{equation}
and
\begin{equation}\label{eq_a123}
\mathcal{A}_{123}:=\{\mathsf{x}-(\mathrm{ab})^{2/3}< \bm{v}^\top \bm{l} \leq \mathsf{x}+(\mathrm{ab})^{2/3}\} \equiv \left\{\frac{\mathsf{x}-(\mathrm{ab})^{2/3}}{\mathrm{w}}< \frac{\mathsf{v}}{\mathrm{w}} \leq \frac{\mathsf{x}+(\mathrm{ab})^{2/3}}{\mathrm{w}}\right\}. 
\end{equation}
Moreover, we construct that 
\begin{equation*}
\mathsf{S}_{\mathsf{x}}:=\left\{\frac{\mathsf{v}}{\mathrm{w}} \leq \mathsf{x}\right\}. 
\end{equation*}
Clearly, we have that $\mathcal{A}_1^c \cap \mathcal{A}_2 \cap \mathcal{A}_3 \subset \mathcal{A}_{123,o} \subset \mathcal{A}_{123} \backslash \mathsf{S}_{\mathsf{x}}.$ Consequently, the rest of the proof leaves to control 
the probability $\mathbb{P}(\mathcal{A}_{123} \backslash \mathsf{S}_{\mathsf{x}})$ using (5) of Lemma \ref{lem_intergralappoximation} {for the standard Gaussian random variable $\frac{\mathsf{v}}{\mathrm{w}}$}.
%
Then we can use (5) of Lemma \ref{lem_intergralappoximation} to get 
\begin{equation}\label{eq_probabilitya123control}
\mathbb{P}(\mathcal{A}_{123} \backslash \mathsf{S}_{\mathsf{x}} ) = \mathbb{P}( \mathcal{A}_{123} \backslash \mathsf{S}_{\mathsf{x}} )=\mathrm{O}((\mathrm{ab})^{2/3}/\mathrm{w}).
\end{equation}
Together with Assumption \ref{assum_updc} and the fact that $|\bm{l}|$ is bounded from below, 
}
combining the above discussions,  we conclude that for some constant $C>0$ {
\begin{equation}\label{eq_stepthreeconclusion}
\mathcal{K}(\bar{\bm{v}}^\Mt, \bm{v}) \leq C \left[\sqrt{p} (\zeta+\gamma \iota) \left( p \xi^2 \varsigma^2 \fm^{-\tau+1}+ pn \xi^4 \varsigma^2 \hd^{-(q-2)} \right)  \right]^{2/3}. 
\end{equation} }
This completes the proof of Step three.

\vspace{3pt}
\noindent {\bf Step four. \underline{Control the second term  $\mathcal{K}(\zb^\Mt, \bar{\zb}^\Mt).$}} In this step, we apply a discussion similar to Step three except that we utilize the asymptotic normality of $\bar{\zb}^\Mt$ which has been established in Step two under the assumption of (\ref{eq_onebound}). {We elaborate this in more details as follows.} Analogously to (\ref{eq_probabilitydecomposition}), we decompose that {
\begin{equation}\label{eq_dd1}
\mathcal{K}(\bar{\zb}^\Mt, \bm{z}^\Mt)=\sup_{\mathsf{x} \in \mathbb{R}} \left| \mathbb{P}( (\bar{\zb}^\Mt)^\top \bm{l} \leq \mathsf{x} )-\mathbb{P} ((\bar{\zb}^\Mt)^\top \bm{l} \leq \mathsf{x}+\mathtt{D}(\bar{\zb}^\Mt, \zb^\Mt)) \right|,
\end{equation} }
{where $\mathtt{D}(\bar{\zb}^\Mt, \zb^\Mt)$ is defined similarly as in (\ref{eq_differencedefinition}) {using $\bar{\zb}^\Mt$ and $\bm{z}^\Mt$}. Recall (\ref{eq_firstorder}). By a discussion similar to (\ref{eq_expectationcontrolone}), we have that for some constant $C>0,$
\begin{equation*}
\mathbb{E} | \bar{z}_i^\Mt-z^\Mt_i | \leq C \sqrt{n} \xi^2 \varsigma \hd^{-(q-1)}.
\end{equation*}
Similarly, we have that 
{
\begin{equation}\label{eq_needboundtwo}
(\mathbb{E} \| \bar{\zb}^\Mt-\zb^\Mt \|_2^2)^{1/2} \leq C \sqrt{p n} \xi^2 \varsigma \hd^{-(q-1)}:=\mathrm{a}_1. 
\end{equation} }
Together with (\ref{oldb212121}), using Cauchy-Schwarz inequality, 
we have that {
\begin{equation}\label{eq_dd2}
\|\mathtt{D}(\bar{\zb}^\Mt, \zb^\Mt)\| \leq \mathrm{a}_1 \mathrm{b}. 
\end{equation}
Based on the above quantities, we can define some probability events $\mathcal{B}_{1}, \mathcal{B}_2, \mathcal{B}_3$ analogous to (\ref{eq_allprobabilityevents}) using $\bar{\zb}^\Mt, \mathtt{D}(\bar{\zb}^\Mt, \zb^\Mt)$ and $\mathrm{a}_1.$ Then by a discussion similar to (\ref{eq_decompositionhh}), we can bound (\ref{eq_dd1}) using the modified counterparts of $\mathsf{P}_1$ and $\mathsf{P}_2.$ Similar to (\ref{eq_p2control}), we can bound $\mathsf{P}_2=\mathrm{O}(\mathrm{a}_1 \mathrm{b}).$ For $\mathsf{P}_1,$ similar to (\ref{eq_a123}), it suffices to bound the probability of the following set
\begin{equation*}
\mathcal{B}_{123,o}:=\{\mathsf{x}<(\bar{\zb}^\Mt)^\top \bm{l} \leq \mathsf{x}+(\mathrm{a_1 b})^{2/3}\},
\end{equation*}
and 
\begin{equation*}
\mathcal{A}'_{123,o}:=\{\mathsf{x}< \bm{v}^\top \bm{l} \leq \mathsf{x}+(\mathrm{a_1 b})^{2/3}\}.
\end{equation*}
According to the results of step two, we have that 
\begin{equation*}
\mathbb{P}(\mathcal{B}_{123,o})=\mathbb{P}(\mathcal{A}'_{123,o})+\mathrm{O}\left(p^{7/4} n^{-1/2}  \mathrm{h}^3 \fm^2 \right). 
\end{equation*} 
Finally, similar to the discussion of (\ref{eq_probabilitya123control}), we have that 
\begin{equation*}
\mathbb{P}(\mathcal{A}'_{123,o}) \leq \mathbb{P}(\mathcal{A}'_{123} \backslash \mathsf{S}_{\mathsf{x}}) =\mathrm{O}\left(\left[ \sqrt{p n} (\zeta+\gamma \iota) \xi^2 \varsigma \hd^{-(q-1)}\right]^{2/3}\right).
\end{equation*}

Combining all the above discussions with (\ref{oldb2asadadas}), we can obtain that for some constant $C>0$
\begin{align}\label{eq_stepfourconclusion}
\mathcal{K}(\bar{\zb}^\Mt, \bm{z}^\Mt)  \leq Cp^{7/4} n^{-1/2}  \mathrm{h}^3 \fm^2   + C \left[ \sqrt{p n} (\zeta+\gamma \iota) \xi^2 \varsigma \hd^{-(q-1)}\right]^{2/3}. 
\end{align}}
This completes the proof of Step four. 

\vspace{3pt}
\noindent{\bf Step five. \underline{Control the first term  $\mathcal{K}(\zb, \zb^\Mt).$}} In the last step, we apply a discussion similar to Step four utilizing the {Gaussian approximation results} of $\zb^\Mt$ as established in Step four under the assumption of (\ref{eq_onebound}). By (\ref{eq_control}), we have that for some constant $C>0$
\begin{equation}\label{eq_boundneedone}
(\mathbb{E} \| \zb-\zb^\Mt \|^2_2)^{1/2} \leq C \sqrt{p} \sum_{i=1}^p \Theta_{\fm,q} \leq C p^{3/2} \xi \varsigma \fm^{-\tau+1},
\end{equation}
where in the second inequality we again used (\ref{eq_transferbound}). {Similar to the discussion between (\ref{eq_dd1}) and (\ref{eq_stepfourconclusion}),} we conclude that for some constant $C>0$ {
\begin{align}\label{eq_stepfiveconclusion}
\mathcal{K}(\zb^\Mt, \zb)  \leq Cp^{7/4} n^{-1/2}  \mathrm{h}^3 \fm^2   + C \left[ \sqrt{p n} (\zeta+\gamma \iota) \xi^2 \varsigma \hd^{-(q-1)}\right]^{2/3}+C\left[p^{3/2} \xi (\zeta+\gamma \iota) \varsigma \fm^{-\tau+1} \right]^{2/3}.
\end{align}}
This completes the proof of Step five. 

In summary, by (\ref{eq_steptworesults}), (\ref{eq_stepthreeconclusion}), (\ref{eq_stepfourconclusion}), (\ref{eq_stepfiveconclusion}) and (\ref{eq_decomposition}), we can conclude the proof of $\mathcal{K}. $
}
\end{proof}

{\subsection{High dimensional uniform Gaussian approximation and proof of Theorem \ref{thm_uniformdistribution}} In this subsection, we will establish a uniform Gaussian approximation result beyond fixed values of $t$ and $x.$ {The results and techniques will be used to prove Theorem \ref{thm_uniformdistribution} and the results in the next subsection.}  {Denote 
\begin{equation}\label{eq_uniformGaussianneedtobebounded}
\mathcal{K}^*(\bm{z}, \bm{v}):=\sup_{\mathsf{x} \geq 0} \left| \mathbb{P}( \sup_{t \in [0,1], x \in \mathbb{R}} |\bm{z}^\top \bm{l}| \leq \mathsf{x})-\mathbb{P}( \sup_{t \in [0,1], x \in \mathbb{R}}  |\bm{v}^\top \bm{l}| \leq \mathsf{x}) \right|,
\end{equation}  
where we again recall from (\ref{eq_originaldefinition}) that $\bm{l} \equiv \bm{l}(t,x).$ Recall $\Theta^*$ in (\ref{eq_thetastar}).

\begin{theorem}\label{thm_uniformconvergencegaussian}
Suppose the assumptions of Theorem \ref{thm_consistency} holds. Moreover, we assume that (\ref{eq_boothstrappingextraassumption}) holds. Then we have that 
\begin{equation*}
\mathcal{K}^*(\bm{z}, \bm{v})=\mathrm{O}\left( \Theta^* \right). 
\end{equation*} 
\end{theorem}

{Before proceeding to the proof of Theorem \ref{thm_uniformconvergencegaussian}, we first use it to prove Theorem \ref{thm_uniformdistribution}.}
\begin{proof}[\bf Proof of Theorem \ref{thm_uniformdistribution}]
{Our proof relies on Theorem \ref{thm_uniformconvergencegaussian} and Lemma \ref{lem_volumeoftube}.} Since if $t$ and $x$ vary, $\sqrt{n} \mathsf{T}$ is asymptotically a Gaussian process whose convergence rate can be controlled in Theorem \ref{thm_uniformconvergencegaussian}, it suffices to focus on the Gaussian case. Recall (\ref{eq_ljdefinition}).
By (\ref{eq_betaolsform}), (\ref{eq_proposedestimator}) and a discussion similar to (\ref{eq_fundementalexpression}), we can write
\begin{align*}
\widehat{m}_{c,d}(t,x)&= \bm{r}^\top (W^\top W)^{-1} W^\top \bm{Y}+\left(\widehat{\bm{r}}^\top-\bm{r} \right) (W^\top W)^{-1} W^\top \bm{Y}\\
&=l(t,x)^\top W^\top \bm{Y}+\bm{r}^\top \left( (n^{-1}W^\top W)^{-1}-\Pi^{-1}) \right)\left[ n^{-1} W^\top \bm{Y} \right]+\left(\widehat{\bm{r}}^\top-\bm{r} \right) (W^\top W)^{-1} W^\top \bm{Y}. 
\end{align*}
Using (\ref{eq_consistencyconvergency}) and Assumption \ref{assum_updc}, we have that{
\begin{equation*}
\left| \bm{r}^\top \left( (n^{-1}W^\top W)^{-1}-\Pi^{-1}) \right)\left[ n^{-1} W^\top \bm{Y} \right]\right |=\OO_{\mathbb{P}}\left(\zeta \sqrt{p} \left[  p\left( \frac{\varsigma^2 \xi^2}{\sqrt{n}}+\frac{\xi^2 \varsigma^2 n^{\frac{2}{\tau+1}}}{n}\right) \right] \right),
\end{equation*}}
where in the second step we used Cauchy-Schwartz inequality and a discussion similar to (\ref{eq_boundtwo}). Similarly, using (\ref{eq_boundnnnnnn}), we can show that {
\begin{equation*}
\left| \left(\widehat{\bm{r}}^\top-\bm{r} \right) (W^\top W)^{-1} W^\top \bm{Y} \right |=\mathrm{O}_{\mathbb{P}}\left(\zeta \sqrt{p} \left[ \xi \varsigma \gamma \iota \sqrt{\frac{p}{n}} \right] \right).
\end{equation*} } 
Consequently, under Assumption \ref{assum_debiasassumption} and {the assumptions of Theorem \ref{thm_consistency}}, in view of (\ref{eq_ll}), we can further write the followings {uniformly for $(t,x)$}
\begin{equation}\label{eq_rrr}
\widehat{m}_{c,d}(t,x)=\widetilde{l}(t, \widetilde{x})^\top W^\top \bm{Y}+\oo_{\mathbb{P}}(1).
\end{equation}
Since the expression (\ref{eq_alphatheortrep}) is independent of the temporal relation of $\bm{\epsilon},$ using (\ref{eq_defntdefn}), (\ref{eq_rrr}) {as well as Theorem \ref{thm_uniformconvergencegaussian} with the assumption $\Theta^*=\oo(1)$}, we can write {
\begin{equation*}
\alpha=\p\left( \sup_{t \in [0,1],x \in \mathcal{X}} \left| T(t,\widetilde{x})^\top \widetilde{\bm{\epsilon}} \right| \geq \mathsf{c}_{\alpha}+\oo(1)  \right),
\end{equation*} }
where $\widetilde{\bm{\epsilon}} \sim \mathcal{N}(0, \mathbf{I}_{n})$ and $\bm{\epsilon}=\sqrt{\operatorname{Cov}(W^\top \bm{\epsilon})}\widetilde{\bm{\epsilon}}.$ By Lemma \ref{lem_volumeoftube} and the fact that {$\mathsf{c}_{\alpha}\geq \kappa$ for some constant $\kappa>0$}, we complete the proof of (\ref{eq_expansionformula}).

Then we prove (\ref{eq_quantitiesbound}). According to equation (3.2) of \cite{MR1311978}, we have that
\begin{equation}\label{eq_kappodefinition}
\kappa_0=\int_0^1 \int_0^1 \sqrt{\det(A A^\top )} \dd t \dd x,
\end{equation}
where $A \in \mathbb{R}^{ 2 \times n}$ is defined as $A^\top=\left( \frac{\partial}{\partial t } T,  \frac{\partial}{\partial \widetilde{x} } T \right).$ Here we recall (\ref{eq_defntdefn}) for the definition of $T.$ By an elementary calculation, it is easy to see that 
\begin{equation*}
\frac{\partial T}{\partial t}=\frac{\sqrt{n} l_t \sqrt{\operatorname{Cov}(W^\top\bm{\epsilon})}}{\widetilde{h}(t, \widetilde{x})}-\frac{1}{2} \frac{T^\top(t,\widetilde{x})}{\widetilde{h}^2(t, \widetilde{x})} s_t, \ 
\end{equation*} 
where we use the notations that
\begin{equation*}
l_t= n^{-1} (\nabla_t \bm{\phi}(t) \otimes \bm{\varphi}(\widetilde{x})-\nabla_t \bm{\phi}(t) \otimes \bm{\vartheta}(t)- \bm{\phi}(t) \otimes \nabla_t \bm{\vartheta}(t))^\top \Pi^{-1}, 
\end{equation*}
where
\begin{equation*}
\ \bm{\phi}(t)=(\phi_1(t), \cdots, \phi_c(t))^\top, \  \bm{\vartheta}(t)=(\vartheta_1(t), \cdots, \vartheta_d(t))^\top, \ \bm{\varphi}(\widetilde{x})=(\varphi_1(\widetilde{x}), \cdots, \varphi_d(\widetilde{x}))^\top,
\end{equation*}
and 
\begin{align*}
s_t&=(\nabla_t \bm{\phi}(t) \otimes \bm{\varphi}(\widetilde{x})) \Pi^{-1} \Omega \Pi^{-1} \bm{r}+ \bm{r}^\top \Pi^{-1} \Omega \Pi^{-1} (\nabla_t \bm{\phi}(t) \otimes \bm{\varphi}(\widetilde{x})) \\
&-(\nabla_t \bm{\phi}(t) \otimes \bm{\vartheta}(t)) \Pi^{-1} \Omega \Pi^{-1} \bm{r}- \bm{r}^\top \Pi^{-1} \Omega \Pi^{-1} (\nabla_t \bm{\phi}(t) \otimes \bm{\vartheta}(t)) \\
& -( \bm{\phi}(t) \otimes \nabla_t \bm{\vartheta}(t)) \Pi^{-1} \Omega \Pi^{-1} \bm{r}- \bm{r}^\top \Pi^{-1} \Omega \Pi^{-1} ( \bm{\phi}(t) \otimes \nabla_t \bm{\vartheta}(t)). 
\end{align*}
Similarly, we can calculate $\frac{\partial T}{\partial x}.$
For notational simplicity, we set
\begin{align*}
& \mathbf{b}_1:=\nabla_t \bm{\phi}(t) \otimes \bm{\varphi}(\widetilde{x})-\nabla_t \bm{\phi}(t) \otimes \bm{\vartheta}(t)- \bm{\phi}(t) \otimes \nabla_t \bm{\vartheta}(t), \\
& \mathbf{b}_2:= \bm{\phi}(t) \otimes \nabla_{\widetilde{x}} \bm{\varphi}(\widetilde{x}).
\end{align*}
With the above notations, we have that 
\begin{equation*}
\det(A A^\top )=\mathcal{E}_1\mathcal{E}_3-\mathcal{E}_2^2,
\end{equation*}
where $\mathcal{E}_k, k=1,2,3,$ are defined as 
\begin{equation*}
\mathcal{E}_1=\frac{\partial T}{\partial t} \frac{\partial T}{\partial t}^\top, \ \mathcal{E}_3= \frac{\partial T}{\partial \widetilde{x}} \frac{\partial T}{\partial \widetilde{x}}^\top , \ \mathcal{E}_2= \frac{\partial T}{\partial \widetilde{x}} \frac{\partial T}{\partial t}^\top.
\end{equation*}
It is easy to see that
\begin{equation*}
\widetilde{h}(t, \widetilde{x}) \asymp \|  \widetilde{\bm{r}}\|_2^2, \ n \widetilde{l}(t, \widetilde{x}) \asymp \|  \widetilde{\bm{r}}\|_2^2, \ s_t^2 \asymp  \| \mathbf{b}_1 \|_2^2.  
\end{equation*}
Moreover, by Assumptions \ref{assum_physical} and \ref{assum_updc}, using (\ref{eq_consistencyconvergency}), we have that for all $1 \leq k \leq p,$
\begin{equation*}
\lambda_{k}(\Pi^{-1}), \ \lambda_{k}(\Omega)  \asymp 1. 
\end{equation*}
With the above estimates, for $\mathcal{E}_1,$ by an elementary but tedious calculation,  we can see that   
\begin{equation}\label{eq_boundcontrol}
\frac{\|\mathbf{b}_1 \|_2^2}{\| \widetilde{\bm{r}} \|_2^2}   \lambda_{\min}(\mathbf{T}) \leq  \| \mathcal{E}_1\| \leq   \lambda_{\max}(\mathbf{T}) \frac{\|\mathbf{b}_1 \|_2^2}{\| \widetilde{\bm{r}} \|_2^2}, \ \mathbf{T}:= \operatorname{Cov}(W^\top \bm{\epsilon}). 
\end{equation}
By (\ref{eq_originaldefinition}), Lemma \ref{lem_covarianceofz} and Assumption \ref{assum_updc}, we find that 
\begin{equation*}
 \lambda_{\min}(\mathbf{T}), \lambda_{\max}(\mathbf{T}) \asymp p.
\end{equation*}
Consequently, we have that 
\begin{equation*}
 \mathcal{E}_1 \asymp p \frac{\|\mathbf{b}_1 \|_2^2}{\| \widetilde{\bm{r}} \|_2^2}. 
\end{equation*}
Similarly, we can show 
\begin{equation*}
 \mathcal{E}_3 \asymp p \frac{\|\mathbf{b}_2 \|_2^2}{\| \widetilde{\bm{r}} \|_2^2}.
\end{equation*}
Therefore, under the assumption of (\ref{eq_assumptionbasisderivative}) and using the choice of (\ref{eq_cdchoice}), we have that for some constant $C>0$
\begin{equation*}
\kappa_0 \leq \int \int \sqrt{\mathcal{E}_1 \mathcal{E}_3}  \dd t\dd \widetilde{x} \leq Cp(n^{\alpha_1}+n^{\alpha_2}).
\end{equation*}
To provide a lower bound, since $AA^\top$ is symmetric and positive semi-definite, by (\ref{eq_kappodefinition}), we have that
\begin{equation*}
\kappa_0 \geq \int \int \lambda_{\min}(AA^\top) \dd t \dd \widetilde{x}. 
\end{equation*}
Together with Lemma \ref{lem_circle}, we have  that
\begin{equation*}
\kappa_0 \geq  \int \int \sqrt{\min\left\{(|\mathcal{E}_2|-\mathcal{E}_3)^2, \ (|\mathcal{E}_2|-\mathcal{E}_1)^2 \right\}} \dd t \dd \widetilde{x}. 
\end{equation*}  
By a discussion similar to (\ref{eq_boundcontrol}), using the definitions of $\mathcal{E}_k, k=1,2,3,$ we see that for some constant $C_1>0,$
\begin{equation*}
\kappa_0 \geq C_1 p(n^{\alpha_1}+n^{\alpha_2}). 
\end{equation*}
 This yields that
\begin{equation*}
 \kappa_0 \asymp  p n^{\alpha_1}+p n^{\alpha_2}. 
\end{equation*}  
We point out that $\eta_0$ can be analyzed similarly using the second equation on Page 1335 of \cite{MR1311978}. Together with (\ref{eq_expansionformula}), we can complete the proof of (\ref{eq_quantitiesbound}).
\end{proof}

{The rest of the work focuses on proving Theorem \ref{thm_uniformconvergencegaussian}. }
The key idea for controlling (\ref{eq_uniformGaussianneedtobebounded}) is to employ the smoothing technique proposed in \cite{bentkus2003dependence}, the convex Gaussian approximation trick in \cite{MR3571252,FangRollin2015} and the crucial bounds obtained in the proofs of Theorem \ref{thm_gaussianapproximationcase}.

\begin{proof}[\bf Proof of Theorem \ref{thm_uniformconvergencegaussian}]
For the purpose of smoothing, we introduce the following function as in \cite{bentkus2003dependence}. For the collection of convex sets $\mathcal{A}$ in $\mathbb{R}^p$ and $\mathsf{A} \in \mathcal{A},$ given some small $\epsilon_1>0,$ we denote 
\begin{equation*}
h_{\mathsf{A}, \epsilon_1}(\bm{q})=\eta\left( \frac{\operatorname{dist}(\bm{q},\mathsf{A})}{\epsilon_1} \right),
\end{equation*} 
where $\operatorname{dist}(\bm{q},\mathsf{A})=\inf_{\bm{w} \in \mathsf{A}}|\bm{q}-\bm{w}|$ and 
\begin{equation*}
\eta(x)=
\begin{cases}
1, & x<0, \\
1-2x^2, & 0 \leq x<\frac{1}{2}, \\
2(1-x)^2, & \frac{1}{2} \leq x<1, \\
0, & x \geq 1. 
\end{cases}
\end{equation*}
According to (iv) of Lemma 2.3 of \cite{bentkus2003dependence}, we have that 
\begin{equation}\label{eq_derivativebound}
|\nabla h_{\mathsf{A}, \epsilon_1} | \leq \frac{2}{\epsilon_1}.
\end{equation}

Recall the convention in (\ref{eq_ddddd}), we have that $\bm{v}=\Sigma_0^{1/2}\mathbf{f}$ with $\mathbf{f}$ being a standard Gaussian random vector. Similar, for notional simplify, we also write $\bm{z}=\Sigma_0^{1/2} \mathbf{x},$ where $\mathbf{x}=\Sigma_0^{-1/2} \bm{z}.$ Denote $\bm{l}_1=\Sigma_0^{1/2} \bm{l}.$ We can then rewrite (\ref{eq_uniformGaussianneedtobebounded}) as follows 
\begin{equation}\label{eq_fffafafearaefaef}
\mathcal{K}^*(\bm{z}, \bm{v}) \equiv \mathcal{K}^*(\xb, \fb):=\sup_{\mathsf{x} \geq 0} \left| \mathbb{P}( \sup_{t \in [0,1], x \in \mathbb{R}} |\xb^\top \bm{l}_1| \leq \mathsf{x})-\mathbb{P}( \sup_{t \in [0,1], x \in \mathbb{R}}  |\fb^\top \bm{l}_1| \leq \mathsf{x}) \right|. 
\end{equation}
For $\mathsf{x} \in \mathbb{R},$ denote
\begin{equation*}
\mathsf{A}_{\mathsf{x}} \equiv \mathsf{A}_{\mathsf{x}}(t,x):=\left\{ \bm{q} \in \mathbb{R}^p: |\bm{q}^\top \bm{l}_1| \leq \mathsf{x}   \right\}, \ \ \text{and} \ \ \mathsf{A}_{\mathsf{x}}:=\bigcap_{t \in [0,1] \\ x \in \mathbb{R}}\mathsf{A}_{\mathsf{x}}(t,x). 
\end{equation*}
Similar to the discussions below (\ref{eq_defnset}), we have that 
\begin{equation*}
\mathcal{K}^*(\xb, \fb) \leq \sup_{\mathsf{A} \in \mathcal{A}} |\mathbb{P}(\xb \in \mathsf{A})-\mathbb{P}(\fb \in \mathsf{A})|.  
\end{equation*} 
Consequently, according to Lemma 4.2 of \cite{FangRollin2015}, we can control
\begin{equation}\label{eq_reducednorm}
\mathcal{K}^*(\xb, \fb) \leq 4 p^{1/4} \epsilon_1+\sup_{\mathsf{A} \in \mathcal{A}} \left| \mathbb{E}\left[h_{\mathsf{A}, \epsilon_1}(\xb)- h_{\mathsf{A}, \epsilon_1}(\fb)\right] \right|. 
\end{equation} 
Therefore, the rest of the proof leaves to control the right-hand side of (\ref{eq_reducednorm}). For notional simplicity, for the quantities in (\ref{eq_decomposition}), we denote 
\begin{equation*}
\xb^{\Mt}=\Sigma_0^{-1/2} \bm{z}^{\Mt}, \ \overline{\xb}^{\Mt}=\Sigma_0^{-1/2} \overline{\bm{z}}^{\Mt}, \  \ \overline{\fb}^{\Mt}=\Sigma_0^{-1/2} \overline{\bm{v}}^{\Mt}.
\end{equation*} 
Note that according to Lemma \ref{lem_covarianceofz}, the assumption of (\ref{eq_parameterassumption}), as well as Assumption \ref{assum_updc}, we have that $\Sigma_0^{-1/2}$ is bounded from above (in operator norm), i.e., 
\begin{equation}\label{eq_operatornormbound}
\| \Sigma_0^{-1/2} \|_{\operatorname{op}}=\mathrm{O}(1). 
\end{equation}

By triangular inequality, we can write 
\begin{align}\label{eq_decompose}
\sup_{\mathsf{A} \in \mathcal{A}} \left| \mathbb{E}\left[h_{\mathsf{A}, \epsilon_1}(\xb)- h_{\mathsf{A}, \epsilon_1}(\fb)\right] \right| & \leq \sup_{\mathsf{A} \in \mathcal{A}} \left| \mathbb{E}\left[h_{\mathsf{A}, \epsilon_1}(\xb)- h_{\mathsf{A}, \epsilon_1}(\xb^{\Mt})\right] \right|+ \sup_{\mathsf{A} \in \mathcal{A}} \left| \mathbb{E}\left[h_{\mathsf{A}, \epsilon_1}(\xb^\Mt)- h_{\mathsf{A}, \epsilon_1}(\overline{\xb}^\Mt)\right] \right| \nonumber \\
&+ \sup_{\mathsf{A} \in \mathcal{A}} \left| \mathbb{E}\left[h_{\mathsf{A}, \epsilon_1}(\overline{\xb}^\Mt)- h_{\mathsf{A}, \epsilon_1}(\overline{\fb}^\Mt)\right] \right|+\sup_{\mathsf{A} \in \mathcal{A}} \left| \mathbb{E}\left[h_{\mathsf{A}, \epsilon_1}(\overline{\fb}^\Mt)- h_{\mathsf{A}, \epsilon_1}(\fb)\right] \right| \nonumber \\
&:=L_1+L_2+L_3+L_4. 
\end{align}

Then we bound them term by term. For $L_1,$ using (\ref{eq_derivativebound}), we find that for some constant $C>0,$
\begin{equation*}
L_1 \leq \frac{C}{\epsilon_1}\| \xb-\xb^{\Mt} \|=\mathrm{O} \left( \epsilon_1^{-1} p^{3/2} \xi \varsigma \mathfrak{m}^{-\tau+1} \right), 
\end{equation*}
where we used (\ref{eq_boundneedone}) and (\ref{eq_operatornormbound}). Similarly, using (\ref{eq_needboundtwo}), we can control that
\begin{equation*}
L_2=\mathrm{O} \left( \epsilon_1^{-1}\sqrt{p n} \xi^2 \varsigma \hd^{-(q-1)} \right), 
\end{equation*}  
and using (\ref{eq_covarianceclose}) and (\ref{oldb1}), we can control that 
\begin{equation*}
L_4=\mathrm{O}\left( \epsilon_1^{-1} \sqrt{p} \varsigma^2 \left( p \xi^2 \fm^{-\tau+1}+ pn \xi^4 \hd^{-(q-2)} \right) \right). 
\end{equation*}
Finally, for $L_3,$ following the proof of Lemma \ref{lem_mvnapp}, especially equation (4.24) of \cite{MR3571252}, using the setup as in (\ref{eq_setup}), we have that for some constant $C>0$
\begin{equation*}
L_3 \leq 4 p^{1/4} \epsilon_1+C n^{-1/2}(\sqrt{p} \mathrm{h})^3 \mathfrak{m}^2 \epsilon_1^{-1} \left[p^{1/4}(\epsilon_1+\mathfrak{m}  n^{-1/2}\sqrt{p} \mathrm{h}) +L_3\right]. 
\end{equation*}
This yields that 
\begin{equation*}
L_3 \leq 4 p^{1/4} \epsilon_1+ C n^{-1/2} p^{7/4} \mathrm{h}^3  \mathfrak{m}^2+C n^{-1} p^{9/4} \epsilon_1^{-1} \mathrm{h}^4 \mathrm{m}^3+n^{-1/2}p^{3/2} \epsilon_1^{-1} \mathrm{h}^3 \mathrm{m}^2 L_3.  
\end{equation*}
Note that we can choose $\epsilon_1=n^{-1/2}p \mathrm{h}^2 \mathfrak{m}^{3/2}$ to optimize the control of $L_3$ so that 
\begin{equation*}
L_3=\mathrm{O} \left( n^{-1/2} p^{7/4} \mathrm{h}^3  \mathfrak{m}^2\right). 
\end{equation*}
Moreover, with  $\epsilon_1=n^{-1/2}p \mathrm{h}^2 \mathfrak{m}^{3/2}$, we readily obtain that 
\begin{align*}
L_1+L_2+L_4 & =\mathrm{O}\Big( \sqrt{pn} \mathrm{h}^{-2} \varsigma \xi \mathfrak{m}^{-\tau-1/2}+n p^{-1/2} \xi^2 \varsigma \mathfrak{m}^{-3/2} \mathrm{h}^{-q-1}  \\
&+n^{1/2} p^{-1/2} \mathrm{h}^{-2} \mathfrak{m}^{-3/2} \varsigma^2 \left( p \xi^2 \fm^{-\tau+1}+ pn \xi^4 \hd^{-(q-2)} \right) \Big).
\end{align*}
Combining all the above discussions with (\ref{eq_decompose}), (\ref{eq_reducednorm}) and (\ref{eq_fffafafearaefaef}), we have completed our proof. 
\end{proof}
}
}

\subsection{Consistency of bootstrap: proof of Theorem \ref{thm_boostrapping}}
In this subsection, {with the aid of Theorem \ref{thm_uniformconvergencegaussian},} we prove Theorem \ref{thm_boostrapping} following and generalizing \cite{DZ2,MR3174655}. {Recall the notations around Lemma \ref{lem_locallystationaryform}.} Denote 
\begin{equation}\label{eq_defnupsilon}
\Upsilon_{i,m} \in \mathbb{R}^{\mathsf{p}}:=(\epsilon_{i,m} \otimes \phib_0(t_i), \bm{u}_{1,m}(i) \otimes \phib_1(t_i), \cdots, \bm{u}_{r,m}(i) \otimes \phib_r(t_i)).
\end{equation}
Set 
\begin{equation}\label{eq_defnitionphi}
\Phi=\frac{1}{\sqrt{n-m-1} \sqrt{m}} \sum_{i=1}^{n-m} \Upsilon_{i,m} R_i,
\end{equation}
where $\{R_i\}$ are the same random variables as in  
(\ref{eq_defnstatisticXi}). We point out that $\Phi$ is a population version of $\Xi$ in (\ref{eq_defnstatisticXi}). Without loss of generality, we focus on the case $r=1$ as in Sections \ref{sec_consistencyproof} and \ref{sec_gassuianapproximation}. Recall that
$\mathsf{p}=p+c_0.$
Based on the above notations, we denote that
\begin{equation}\label{eq_bigupsilon}
\Upsilon=\frac{1}{(n-m-1)} \sum_{i=1}^{n-m} \Upsilon_{i,m} \Upsilon_{i,m}^\top.
\end{equation}
Note that $\Upsilon$ is the covariance matrix of $\Phi$ when conditioned on the data. 

The proof routine contains four steps. The first three steps concern $\Phi.$ In particular, step one (c.f. Lemma \ref{lem_stepone}) aims to establish the concentration of $\Upsilon_{i,m} \Upsilon_{i,m}^\top $ and $\Upsilon \Upsilon^\top;$ step two (c.f. Lemma \ref{lem_steptwo}) focuses on constructing a stationary time series whose covariance matrix can well approximate the concentration from step one; step three (c.f. Lemma \ref{lem_stepthree}) utilizes the stationary time series from step two and  shows that its covariance matrix is close to the integrated long-run covariance matrix $\overline{\Omega};$ step four (c.f. Lemma \ref{lem_residualclosepreparation}) will conclude the proof by replacing the error $\{\epsilon_i\}$ with the residuals of the sieve estimators (c.f. (\ref{eq_defnresidual})) and handle $\Xi$ to conclude the proof. 
 
Steps one, two and three focus on analyzing $\Upsilon$ and the long-run covariance matrices and are general, i.e., irrelevant of the inference problems. Combining them we can show that $\Upsilon$ is close to $\overline{\Omega}.$  Therefore, we separate them and prove them in Lemmas \ref{lem_stepone}--\ref{lem_stepthree} below. Step four is more specific and we will provide the details when we prove Theorem \ref{thm_boostrapping}.

\begin{lemma}\label{lem_stepone}
Suppose Assumptions \ref{assum_models}--\ref{assum_updc} hold and $m=\oo(n)$. We have that
\begin{equation}\label{eq_overallproof}
\left\| \Upsilon-\mathbb{E} \Upsilon \right\|_{\op}=\OO_{\mathbb{P}}\left( \mathsf{p} \xi^2 \varsigma \sqrt{\frac{m}{n}} \right). 
\end{equation}
\end{lemma}
\begin{proof}
Without of loss of generality, as before, we focus on some typical entry of $\Upsilon_{i,m} \Upsilon_{i,m}^\top$ as follows {
\begin{equation*}
\mathsf{L}_{1,i}=\left(\sum_{j=i}^{i+m} \varphi_1(X_{j}) \phi_1(t_i) \epsilon_j \right)^2. 
\end{equation*}}
By a discussion similar to Lemma \ref{lem_locallystationaryform}, $\{\mathsf{L}_{1,i}\}$ can be regarded as a locally stationary time series whose physical dependence measure $\delta_{\mathsf{L}_1}(l, q)$ satisfies that 
\begin{equation*}
\delta_{\mathsf{L}_1}(l,q) \leq C \xi^2 \sqrt{m} \left( \sum_{j=l-m}^l \delta_u(j,q) \right),
\end{equation*}
where $\delta_u(j,q)$ is defined in (\ref{eq_deltaxdefinition}). Combing (\ref{eq_transferbound}), (1) of Lemma \ref{lem_concentration}, Lemma \ref{lem_circle}, as well as (1) of Lemma \ref{lem_collectionprobineq}, we can complete the proof.
%
\end{proof}

{
Recall $\bm{u}_j(i), 1 \leq j \leq r$ in (\ref{eq_locallystationaryform})   and $\epsilon_i$ in (\ref{eq_setting2}).} We denote {
\begin{equation*}
\widetilde{\bm{u}}_{j,o}(i)=\Ub_j(t_i, \mathcal{F}_o), \ \widetilde{\epsilon}_{o}(i)=D(t_i, \mathcal{F}_o), \ i \leq o \leq i+m. 
\end{equation*}}
Corresponding to (\ref{eq_defnupsilon}), we define {
\begin{equation*}
\widetilde{\Upsilon}_{i,m}=(\widetilde{\epsilon}_{m}(i) \otimes \phib_0(t_i), \widetilde{\bm{u}}_{1,m}(i) \otimes \phib_1(t_i), \cdots, \widetilde{\bm{u}}_{r,m}(i) \otimes \phib_r(t_i)) \in \mathbb{R}^{\mathsf{p}}. 
\end{equation*} }

\begin{lemma}\label{lem_steptwo}
Suppose the assumptions of Lemma \ref{lem_stepone} hold. We have that 
\begin{equation*}
\sup_{1 \leq i \leq n-m} \left\| \mathbb{E} \left( \Upsilon_{i,m} \Upsilon_{i,m}^\top \right)-\mathbb{E} \left( \widetilde{\Upsilon}_{i,m} \widetilde{\Upsilon}_{i,m}^\top \right) \right\|_{\op}=\OO\left( \mathsf{p} \xi^2 \varsigma^2 \left( \frac{m}{n}\right)^{1-\frac{1}{\tau}} \right). 
\end{equation*}
\end{lemma}
\begin{proof}
We again focus on some typical entry as follows
\begin{align}\label{eq_decompositionm11}
\mathsf{M}_{11}& =\left(\sum_{j=i}^{i+m} \varphi_1(X_{j}) \phi_1(t_i) \epsilon_j \right)^2-\left(\sum_{j=i}^{i+m} \varphi_1(\widetilde{X}_{j}) \phi_1(t_i) \widetilde{\epsilon}_j \right)^2 \\
& =\left[\sum_{j=i}^{i+m} \phi_1(t_i) \left(\varphi_1(X_{j}) \epsilon_j-\varphi_1(\widetilde{X}_{j}) \widetilde{\epsilon}_j \right)   \right] \left[\sum_{j=i}^{i+m} \phi_1(t_i)\left(\varphi_1(X_{j}) \epsilon_j+\varphi_1(\widetilde{X}_{j}) \widetilde{\epsilon}_j \right)   \right]:=\mathrm{P}_1\mathrm{P}_2. \nonumber 
\end{align}
By Lemma \ref{lem_locallystationaryform} and (1) of Lemma \ref{lem_concentration}, we see that 
\begin{equation*}
\| \mathrm{P}_2 \|=\OO( \xi \varsigma \sqrt{m}). 
\end{equation*}
For $\mathrm{P}_1,$ by Lemma \ref{lem_locallystationaryform}, we find that 
\begin{equation*}
\| \mathrm{P}_1 \|=\OO\left( \sum_{j=i}^{i+m} \phi_1(t_i)\varphi_1(X_{j-1})(\widetilde{\epsilon}_j-\epsilon_j) \right). 
\end{equation*}
Together with Lemma \ref{lem_locallystationaryform}, the stochastic continuity property (\ref{eq_slc}) and (1) of Lemma \ref{lem_concentration}, we have that
\begin{equation*}
\| \mathrm{P}_1 \|=\OO\left(\xi \varsigma \sqrt{m} \sum_{j=0}^{\infty} \min\{\frac{m}{n}, j^{-\tau}\} \right)=\OO\left( \xi \varsigma \sqrt{m} \left( \frac{m}{n} \right)^{1-1/\tau}\right).
\end{equation*} 
Combining the above arguments, we obtain that
\begin{equation*}
\| \mathsf{M}_{11} \|=\OO\left( \varsigma^2 \xi^2 m \left( \frac{m}{n} \right)^{1-1/\tau} \right).
\end{equation*}
Together with  Lemma \ref{lem_circle} and (1) of Lemma \ref{lem_collectionprobineq}, we conclude the proof.  
\end{proof}

Corresponding to (\ref{eq_bigupsilon}), denote 
\begin{equation*}
\widetilde{\Upsilon}:=\frac{1}{n-m-1} \sum_{i=1}^{n-m} \widetilde{\Upsilon}_{i,m} \widetilde{\Upsilon}_{i,m}^\top.
\end{equation*}
Recall (\ref{eq_defnpsim}). 
\begin{lemma}\label{lem_stepthree}
Suppose the assumptions of Lemma \ref{lem_stepone} hold. Then we have 
\begin{equation}\label{eq_firstproof}
\left\| \mathbb{E} \widetilde{\Upsilon}-\Omega  \right\|_{\op}=\OO\left(\frac{1}{(n-m-1)^2}+\frac{\mathsf{p} \varsigma \xi^2}{m} \right).
\end{equation}
Consequently, we have  
\begin{equation}\label{eq_finalcovclose}
\left\| \Upsilon-\Omega \right\|_{\op}=\OO(\Psi(m)). 
\end{equation}
\end{lemma}
\begin{proof}
For (\ref{eq_firstproof}), first of all, using the definition of $\Omega(t)$ in (\ref{eq_longrunwitht}), by (2) of Lemma \ref{lem_concentration} , we have that (also see Lemma 4 of \cite{MR3174655})
\begin{equation*}
\sup_{1 \leq i \leq n-m} \left\| \mathbb{E} \widetilde{\Upsilon}_{i,m} \widetilde{\Upsilon}_{i,m}^\top-\Omega(t_i) \right\|_{\op}=\OO\left( \frac{\mathsf{p} \xi^2 \varsigma^2}{m} \right).
\end{equation*}
Together with Lemma \ref{lem_intergralappoximation}, we can complete our proof of (\ref{eq_firstproof}). 

The proof of (\ref{eq_finalcovclose}) follows from (\ref{eq_firstproof}), Lemmas \ref{lem_stepone} and \ref{lem_steptwo}. 
\end{proof}

The following lemma indicates that the residual (\ref{eq_defnresidual}) is close to $\{\epsilon_i\}$ such that $\Xi$ is close to $\Phi.$ Recall $\{\widehat{\bm{u}}(i)\}$ in (\ref{eq_xihatdefinition}). Recall (\ref{eq_defnstatisticXi}). For consistency and notional simplicity, we set 
\begin{equation*}
\widehat{\Upsilon}_{i,m} \equiv \widehat{\Ub}(i,m). 
\end{equation*}
Accordingly, we can define $\widehat{\Upsilon}$ as in (\ref{eq_bigupsilon}) using $\widehat{\Upsilon}_{i,m}.$ Recall (\ref{eq_defnresidual}).
\begin{lemma}\label{lem_residualclosepreparation}
Suppose the assumptions of Lemma \ref{lem_stepone} and Theorem \ref{thm_asymptoticdistribution} hold. We have that 
\begin{equation*}
\sup_{1 \leq i \leq n-m}\| \Upsilon_{i,m} \Upsilon_{i,m}^\top-  \widehat{\Upsilon}_{i,m} \widehat{\Upsilon}_{i,m}^\top\|_{\op}=\OO_{\mathbb{P}}\left( \mathsf{p} \xi^2 \varsigma^2 \left[\xi \varsigma( \gamma  \iota+\zeta) \sqrt{\frac{\mathsf{p}}{n}} \right] \right).
\end{equation*}
As a result, we have that 
\begin{equation*}
\left\| \Upsilon-\widehat{\Upsilon} \right\|_{\op}=\OO_{\mathbb{P}}\left( \frac{\mathsf{p} \xi^2 \varsigma^2}{\sqrt{n}} \left[\xi \varsigma( \gamma  \iota+\zeta) \sqrt{\frac{\mathsf{p}}{n}} \right] \right).
\end{equation*}
\end{lemma}
\begin{proof}
{
Note that according to the definitions in (\ref{eq:model}) and (\ref{eq_defnresidual}), by Proposition \ref{thm_approximation}, we have that 
\begin{align*}
\epsilon_i-\widehat{\epsilon}_i &=m_0(t_i)+\sum_{j=1}^r m_j(t_i,X_{j,i})-\left(\widehat{m}_{0,c}^*(t_i)+\sum_{j=1}^r \widehat{m}^*_{j,c,d}(t_i,X_{j,i}) \right) \\
&= m_{0,c}(t_i)+\sum_{j=1}^r m_{j,c,d}(t_i,X_{j,i})-\left(\widehat{m}_{0,c}^*(t_i)+\sum_{j=1}^r \widehat{m}^*_{j,c,d}(t_i,X_{j,i}) \right)+\mathrm{o}(n^{-1/2}),
\end{align*}
where in the second step we used Assumption \ref{assum_debiasassumption}. Since we do not have identifiability issue, we can directly use (\ref{eq_ll1}) with (\ref{eq_betabound}) and (\ref{eq_adddadada}) to obtain that
}
\begin{equation}\label{eq_residualbound} 
\sup_{1 \leq i \leq n-m}\| \epsilon_i-\widehat{\epsilon}_i\|=\OO\left( \xi \varsigma( \gamma  \iota+\zeta)\sqrt{\frac{\mathsf{p}}{n}} \right).
\end{equation}
Using a discussion similar to  (\ref{eq_decompositionm11}) such that $\mathrm{P}_2$ in  (\ref{eq_decompositionm11}) is controlled using (\ref{eq_residualbound}), we can conclude the proof of the first part of the results. The second part of the results follows from a discussion similar to (\ref{eq_overallproof}).  
\end{proof}

Armed with Theorem \ref{thm_uniformconvergencegaussian} as well as Lemmas \ref{lem_stepone}--\ref{lem_residualclosepreparation}, we proceed to finish the proof of Theorem \ref{thm_boostrapping}.  {Before proceeding to the proof, we collect some results concerning the maximum of a Gaussian process. The results can be easily obtained using the standard technique of chaining and control of Dudley’s entropy integral \cite{talagrand2005generic}. Let $\bm{\beta}(\theta) \in \mathbb{R}^p$ be a smooth function of $\theta$ on a compact set, and $\mathbf{f}$ be a $p$-dimensional standard Gaussian random vector. Then we have that \cite{vanHandel2016}
\begin{equation}\label{eq_controlasdadasfaf}
\mathbb{E} \left( \sup_{\theta} \left| \frac{\mathbf{f}^\top \bm{\beta}(\theta)}{|\bm{\beta}(\theta)|} \right| \right)=\mathrm{O}(\log p). 
\end{equation}
{Moreover, denote $\mathsf{M}=\operatorname{ess} \sup \rho (x),$ where $\rho(x)$ is the density function of $\sup_\theta|\mathbf{f}^\top \bm{\beta}(\theta)|,$ we have from Theorem 2 of \cite{giessing2023anti} that, for $\epsilon=\oo(\mathsf{M}),$
\begin{equation}\label{eq_gaussionprocessanticoncentrationbound}
\sup_{\mathsf{x} \in \mathbb{R}_+}\mathbb{P}(\mathsf{x} \leq \sup_\theta|\mathbf{f}^\top \bm{\beta}(\theta)| \leq \mathsf{x}+\epsilon)=\OO\left(\epsilon \mathsf{M} \right). 
\end{equation}
Note that according to Proposition 5 of \cite{giessing2023anti}, when $\inf_{\theta} |\bm{\beta}(\theta)|$ is bounded from below, $\mathsf{M}$ can be bounded from above. 
}
}

\begin{proof}[\bf Proof of Theorem \ref{thm_boostrapping}] 
As before, we focus on the case $r=1$ and omit the subscript $j.$ {Recall $\mathsf{T}(t,x):=\widehat{m}(t,x)-m(t,x)$, (\ref{eq_fundementalexpression}) and (\ref{eq_originaldefinition}). Using the assumption that $|\bm{r}|$ is bounded from below uniformly, according to the definition (\ref{eq_defhtx}), we see from Assumption \ref{assum_updc} that $h(t,x)$ is bounded from below uniformly in $t$ and $x.$ Consequently, similar to the analysis of Theorem \ref{thm_uniformconvergencegaussian} and under the assumption of (\ref{eq_boothstrappingextraassumption}), we find that 
\begin{equation*}
\sup_{y \in \mathbb{R}} \left| \mathbb{P} \left( \sup_{t,x}\left|\frac{\mathsf{T}(t,x)}{h(t,x)} \right| \leq y \right)-\mathbb{P}( \sup_{t,x} \left| \frac{\bm{v}^\top \bm{l}}{h(t,x)} \right| \leq y) \right|=\oo(1).
\end{equation*}

Consequently, it suffices to show that the bootstrapped statistics can mimic the above distribution when conditioned on the dataset} {$\mathcal{X}:=\{(X_i, Y_i)_{1 \leq i \leq n}\}$.}  
%
{That is 
\begin{equation}\label{eq_intermediateresults}
\sup_{y \in \mathbb{R}} \left| \mathbb{P} \left( \sup_{t,x}\left|\frac{\widehat{\mathsf{T}}(t,x)}{\widehat{h}(t,x)} \right| \leq y \Big| \mathcal{X} \right)-\mathbb{P}( \sup_{t,x} \left| \frac{\bm{v}^\top \bm{l}}{h(t,x)} \right| \leq y) \right|=\oo_{\mathbb{P}}(1). 
\end{equation}  
}

 
To prove (\ref{eq_intermediateresults}), we first construct an intermediate quantity 
\begin{equation}\label{eq_Ttitle}
\widetilde{\mathsf{T}}:=\Phi^\top \widehat{\bm{l}}, \ \ \text{where} \  \ \widehat{\bm{l}}:=\widehat{\overline{\Pi}}^{-1} \widehat{\overline{\bm{r}}}. 
\end{equation}
Note that when conditional on the data, by the definition of $\Phi$ in (\ref{eq_defnitionphi}), $\Phi$ is a Gaussian random vector. We now write
\begin{equation}\label{eq_phigaussianrepression}
\Phi=\Lambda^{1/2} \mathbf{f}_1,
\end{equation}
where without loss of generality $\mathbf{f}_1 \sim \mathcal{N}(\mathbf{0}, \mathbf{I}_p)$ is from (\ref{eq_ddddd}). In fact, conditional on the data $\mathcal{X}$, $\Lambda=\Upsilon.$ Consequently, when conditional on the data, we can rewrite 
\begin{equation}\label{eq_rewritingasasaas}
\widetilde{\mathsf{T}}=\mathbf{f}^\top_1 \Upsilon^{1/2} \widehat{\bm{l}}.
\end{equation}
We also recall from (\ref{eq_ddddd}), $\bm{v}^\top \bm{l}=\mathbf{f}^\top_1 \Sigma_0^{1/2} \bm{l},$ where $\Sigma_0=\operatorname{Cov}(\bm{z}).$  We first establish the closeness between $\Sigma_0, \bm{l}$ and $\Upsilon, \widehat{\bm{l}}.$ First, according to Lemmas \ref{lem_stepthree} and \ref{lem_covarianceofz}, we have that {
\begin{equation}\label{eq_closenessset}
\left\| \Upsilon-\Sigma_0 \right\|_{\op}=\OO_{\mathbb{P}}\left( \mathsf{B} \right), \ \ \mathsf{B}:= \left[ \Psi(m)+\frac{p\xi^2 \varsigma^2 n^{2/\tau}}{\sqrt{n}} +p\xi^2 \varsigma^2 n^{-1+\frac{2}{\tau+1}}  \right] .
\end{equation}}

Second, by definition, we have that 
\begin{equation*}
\widehat{\bm{l}}-\bm{l}=  \widehat{\overline{\Pi}}^{-1} \widehat{\overline{\bm{r}}}-\overline{\Pi}^{-1} \overline{\bm{r}}.
\end{equation*}
Consequently, we can bound {
\begin{align*}
| \widehat{\bm{l}}-\bm{l} | \leq \| \widehat{\overline{\Pi}}^{-1}-\overline{\Pi}^{-1} \|_{\op} | \overline{\bm{r}} |+\| \widehat{\overline{\Pi}}^{-1} \|_{\op}| \widehat{\bm{r}}-\overline{\bm{r}} |. 
\end{align*}}
The first part of the right-hand side of  the above equation can be bounded by (\ref{eq_consistencyconvergency}), the definitions of (\ref{eq_verctorconstruction}) and (\ref{eq_defnxic}), as well as (\ref{eq_bbbbbb}). For the second part, since $| \widehat{\bm{r}}-\overline{\bm{r}} |=|\bm{f}-\widehat{\bm{f}} |,$ it can be bounded by (\ref{eq_bbbbbb11111}). Combining the above discussions, under Assumption \ref{assum_debiasassumption}, we readily obtain that{
\begin{equation}\label{eq_bbbababababab2222}
\sup_{t \in [0,1], x \in \mathbb{R}} | \widehat{\bm{l}}-\overline{\bm{l}} |=\mathrm{O}_{\mathbb{P}}(\mathsf{A}), \ \mathsf{A}:=\left[ p(\gamma  \iota+\zeta)\left( \frac{\varsigma^2 \xi^2}{\sqrt{n}}+\frac{\xi^2 \varsigma^2 n^{\frac{2}{\tau+1}}}{n}\right) \right]. 
\end{equation} } 

With (\ref{eq_rewritingasasaas}) and the above preparation, we now proceed to prove the following result, which serves as a key step toward establishing (\ref{eq_intermediateresults}) {
\begin{equation}\label{eq_intermediateone}
\sup_{y \in \mathbb{R}} \left| \mathbb{P} \left( \sup_{t,x}\left|\frac{\widetilde{\mathsf{T}}(t,x)}{h(t,x)} \right| \leq y \Big| \mathcal{X} \right)-\mathbb{P}( \sup_{t,x} \left| \frac{\bm{v}^\top \bm{l}}{h(t,x)} \right| \leq y) \right|=\oo_{\mathbb{P}}(1).
\end{equation} 
To be more specific, we can rewrite that the left-hand side of the above equation as follows
\begin{equation}\label{eq_zhanzhanzhanzhan}
\sup_{y \in \mathbb{R}} \left| \mathbb{P} \left( \sup_{t,x} \left| \frac{\bm{v}^\top \bm{l}}{h(t,x)} \right| \leq y+ \Delta \Big| \mathcal{X} \right)-\mathbb{P}\left( \sup_{t,x} \left| \frac{\bm{v}^\top \bm{l}}{h(t,x)} \right| \leq y \right) \right|, 
\end{equation} } 
where $\Delta$ is defined as 
\begin{equation*}
\Delta:=\left[\sup_{t \in [0,1], x \in \mathbb{R}}  \left|\frac{\mathbf{f}^\top  \Sigma_0^{1/2} \bm{l}}{h(t,x)} \right| - \sup_{t \in [0,1], x \in \mathbb{R}}  \left|\frac{\mathbf{f}^\top  \Upsilon^{1/2} \widehat{\bm{l}}}{h(t,x)} \right| \right].  
\end{equation*}
Note that for some constant $C>0$
\begin{equation*}
|\Delta| \leq   C\sup_{t \in [0,1], x \in \mathbb{R}}|\mathbf{f}^\top ( \Sigma_0^{1/2} \bm{l}-\Upsilon^{1/2} \widehat{\bm{l}})|, 
\end{equation*}
{where we used Remark \ref{rmk_hj(tx)}.} 

Using (\ref{eq_closenessset}), (\ref{eq_bbbababababab2222}) and (\ref{eq_controlasdadasfaf}), we can show that for some constant $C>0$ {
\begin{equation}\label{eq_boundprobability}
|\Delta|=\OO_{\mathbb{P}}\left(\log p  \left( \mathsf{B}+\mathsf{A}/\mathrm{b} \right)\mathrm{b}\right)=\OO_{\mathbb{P}}(\mathrm{a}' \mathrm{b}).  
\end{equation}}
{Then we can follow the augments between (\ref{eq_allprobabilityevents}) and (\ref{eq_stepthreeconclusion}) verbatim to conclude the proof, except we need to use (\ref{eq_boundprobability}) and (\ref{eq_gaussionprocessanticoncentrationbound}). To be more specific,
analogous to (\ref{eq_oooooo1212121}), we need to control the probability of the following event 
\begin{equation*}
\mathcal{A}_{123,o}:=\{\mathsf{x} \leq \sup_{t,x} |\bm{v}^\top \frac{\bm{l}}{h(t,x)}| \leq \mathsf{x}+(\mathrm{a'b})^{2/3}\}, 
\end{equation*}  
which can be controlled using (\ref{eq_gaussionprocessanticoncentrationbound}) with Remark \ref{rmk_hj(tx)}. Therefore, we can readily obtain that 
\begin{equation}\label{eq_partoneresultsresults}
\sup_{y \in \mathbb{R}} \left| \mathbb{P} \left( \sup_{t,x}|\frac{\widetilde{\mathsf{T}}(t,x)}{h(t,x)}| \leq y \Big| \mathcal{X} \right)-\mathbb{P}( \sup_{t,x} |\frac{\bm{v}^\top \bm{l}}{h(t,x)}| \leq y) \right|=\mathrm{O}_{\mathbb{P}} \left( \left[ \log p (\mathsf{B} (\zeta+\gamma \iota)+\mathsf{A}) \right]^{2/3} \right). 
\end{equation} 
Under the assumption of (\ref{eq_boothstrappingextraassumption}), we have proved (\ref{eq_intermediateone}) holds.

Next, we prove another intermediate result 
\begin{equation}\label{eq_intermediateotwo}
\sup_{y \in \mathbb{R}} \left| \mathbb{P} \left( \sup_{t,x}\left|\frac{\widetilde{\mathsf{T}}(t,x)}{h(t,x)} \right| \leq y \Big| \mathcal{X} \right)-\mathbb{P}(  \sup_{t,x}\left|\frac{\widehat{\mathsf{T}}(t,x)}{h(t,x)} \right|  \leq y \Big| \mathcal{X}) \right|=\oo_{\mathbb{P}}(1).
\end{equation} }
Note that the only difference between $\widehat{\mathsf{T}}$ and $\widetilde{\mathsf{T}}$ lies in the fact that $\widetilde{\mathsf{T}}$ used the unobserved $\{\epsilon_i\}$ while $\widehat{\mathsf{T}}$ used the residual $\{\widehat{\epsilon}_i\}.$ Conditional on the data, by construction, analogous to (\ref{eq_phigaussianrepression}), we can write
\begin{equation}\label{eq_gaussionrepapprox}
\Xi=\widehat{\Lambda}^{1/2} \mathbf{f}. 
\end{equation} 
Accordingly, we can write
\begin{equation*}
\widehat{\mathsf{T}}=\mathbf{f}^\top \widehat{\mathsf{H}},  \ \widehat{\mathsf{H}}= \widehat{\Lambda}^{1/2} \widehat{\overline{\Pi}}^{-1} \widehat{\overline{\bm{r}}}. 
\end{equation*}
Note that condition on the data, we have that
\begin{equation}\label{eq_close}
\Lambda=\Upsilon, \ \widehat{\Lambda}=\widehat{\Upsilon}. 
\end{equation}
For notional simplicity, we denote $\bm{s} \equiv \bm{s}(t,x):=\widehat{\overline{\Pi}}^{-1} \widehat{\overline{\bm{r}}}.$ {Similar to (\ref{eq_zhanzhanzhanzhan}), we can rewrite the left-hand side of (\ref{eq_intermediateotwo}) as follows
\begin{equation*}
\sup_{y \in \mathbb{R}} \left| \mathbb{P} \left( \sup_{t,x} \left| \frac{\mathbf{f}^\top \Upsilon^{1/2}  \bm{s}}{h(t,x)} \right| \leq y+ \Delta_1 \Big| \mathcal{X} \right)-\mathbb{P}( \sup_{t,x} \left| \frac{\mathbf{f}^\top \Upsilon^{1/2}  \bm{s}}{h(t,x)} \right| \leq y\Big| \mathcal{X})  \right|,
\end{equation*}
where $\Delta_1$ is defined as 
\begin{equation*}
\Delta_1:=\left[\sup_{t \in [0,1], x \in \mathbb{R}} \left|\frac{\mathbf{f}^\top \Upsilon^{1/2} \bm{s}}{h(t,x)} \right|-\sup_{t \in [0,1], x \in \mathbb{R}} \left|\frac{\mathbf{f}^\top \widehat{\Upsilon}^{1/2} \bm{s}}{h(t,x)} \right| \right].  
\end{equation*}
By a discussion similar to (\ref{eq_bbbababababab2222}),  for $\mathsf{g}_1:=\sup_{t,x} | \bm{s} |,$ we have that $\mathsf{g}_1=\mathrm{O}_{\mathbb{P}}(\mathrm{b}_1), \ \mathrm{b}_1:=\gamma \iota+\zeta.$ Moreover, by Lemma \ref{lem_residualclosepreparation} and (\ref{oldb1}), we can bound 
\begin{equation*}
| \Delta_1 |=\mathrm{O}_{\mathbb{P}}\left(\log p \frac{p \xi^2 \varsigma^2}{\sqrt{n}} \left[\xi \varsigma( \gamma  \iota+\zeta) \sqrt{\frac{p}{n}} \right]\mathrm{b}_1 \right):=\OO_{\mathbb{P}}\left(\mathrm{a}_1 \mathrm{b}_1 \right). 
\end{equation*} 
Then we can again follow  the augments between (\ref{eq_zhanzhanzhanzhan}) and (\ref{eq_partoneresultsresults}) to conclude the proof of (\ref{eq_intermediateotwo}).

Combining (\ref{eq_intermediateone}) and (\ref{eq_intermediateotwo}), we have proved that
\begin{equation}\label{eq_intermediateresultssub}
\sup_{y \in \mathbb{R}} \left| \mathbb{P} \left( \sup_{t,x}\left|\frac{\widehat{\mathsf{T}}(t,x)}{h(t,x)} \right| \leq y \Big| \mathcal{X} \right)-\mathbb{P}( \sup_{t,x} \left| \frac{\bm{v}^\top \bm{l}}{h(t,x)} \right| \leq y) \right|=\oo_{\mathbb{P}}(1). 
\end{equation}  
Similarly, we can also prove
\begin{equation}\label{eq_intermediateresults1111}
\sup_{y \in \mathbb{R}} \left| \mathbb{P} \left( \sup_{t,x}\left|\widehat{\mathsf{T}}(t,x) \right| \leq y \Big| \mathcal{X} \right)-\mathbb{P}( \sup_{t,x} \left| \bm{v}^\top \bm{l} \right| \leq y) \right|=\oo_{\mathbb{P}}(1). 
\end{equation}
Note that according to (\ref{eq_intermediateresults1111}), when $B$ (i.e., the number of constructed bootstrapped statistics) is sufficiently large, by the construction of $\widehat{h}(t,x),$ we have that $\sup_{t,x} | \widehat{h}(t,x)-h(t,x) |=\mathrm{O}_{\mathbb{P}}(| h(t,x) |B^{-1/2}).$ Then by an argument similar to the discussions between (\ref{eq_zhanzhanzhanzhan}) and (\ref{eq_partoneresultsresults}), we can prove that when $B$ is sufficiently large, 
\begin{equation*}
\sup_{y \in \mathbb{R}} \left| \mathbb{P} \left( \sup_{t,x}\left|\frac{\widehat{\mathsf{T}}(t,x)}{h(t,x)} \right| \leq y \Big| \mathcal{X} \right)-\mathbb{P} \left( \sup_{t,x}\left|\frac{\widehat{\mathsf{T}}(t,x)}{\widehat{h}(t,x)} \right| \leq y \Big| \mathcal{X} \right) \right|=\oo_{\mathbb{P}}(1). 
\end{equation*}  
Combining with (\ref{eq_intermediateresultssub}), we have concluded the proof of (\ref{eq_intermediateresults}), which yields that
\begin{equation*}
\sup_{y \in \mathbb{R}} \left| \mathbb{P} \left( \sup_{t,x}\left|\frac{\widehat{\mathsf{T}}(t,x)}{\widehat{h}(t,x)} \right| \leq y \Big| \mathcal{X} \right)- \mathbb{P} \left( \sup_{t,x}\left|\frac{\mathsf{T}(t,x)}{h(t,x)} \right| \leq y \right) \right|=\oo_{\mathbb{P}}(1). 
\end{equation*}   
This completes our proof. 
}
\end{proof} 


\subsection{Proof of Proposition \ref{thm_approximation} and an auxiliary lemma}

First, the proof of Proposition \ref{thm_approximation} is standard. When $\mathbb{R}=[0,1],$ the results have been established without using the mapped basis functions; see \cite[Section 2.3.1]{CXH}. In our setting, since we are using the mapped sieves to map $\mathbb{R}$ to $[0,1],$ the proof is similar. We only point out the main routine here.

\begin{proof}[\bf Proof of Proposition \ref{thm_approximation}] The first part of the results has been proved in the literature, for example, see  \cite[Section 2.3.1]{CXH}. We now discuss the proof of the second part. For any fixed $x \in \mathbb{R},$ the approximation rate has been summarized in Lemma \ref{lem_deterministicapproximation} for $t.$ Similarly, for any fixed $t \in \mathbb{R},$ the results have been recorded in Lemma \ref{lem_deterministicapproximation2}. Based on these results, we can follow the arguments of \cite[Section 5.3]{MR1262128} to conclude the proof. We point out the \cite[Section 5.3]{MR1262128} deals with the $L_2$ norm assuming that the $L_2$ norm of the partial derivatives are bounded. However, as discussed in \cite[Section 6.2]{MR1176949}, the results can be generalized to the sup-norm under Assumption \ref{assum_smoothnessasumption}.
\end{proof}


Then we prove that both the high-dimensional time series $\{\bm{w}_j(i)\}$ in (\ref{eq_desigmatrixform}) and $\{\bm{u}_j(i)\}$ in (\ref{eq_ddd}) can be regarded as short-range dependent locally stationary time series in the physical representation form.

\begin{lemma}\label{lem_locallystationaryform}
Suppose  Assumptions \ref{assum_models} and \ref{assum_physical} hold. For $1 \leq j \leq r,$ there exist measurable functions $\mathbf{W}_j(t, \cdot)=(W_{j1}(t,\cdot), \cdots, W_{jd_j}(t, \cdot)), \ \mathbf{U}_j(t, \cdot)=(U_{j1}(t,\cdot), \cdots, U_{jd_j}(t,\cdot)) \in \mathbb{R}^{d_j}$ satisfying the stochastic Lipschitz continuity as in (\ref{eq_slc}) such that $\{\bm{u}_j(i)\}$ is mean zero and (\ref{eq_locallystationaryform}) holds.  Moreover, denote their physical dependence measures as
\begin{equation}\label{eq_deltaxdefinition}
\begin{gathered}
\delta_{w,j}(\mathsf{s},q)=\sup_{1 \leq k \leq d_j} \sup_t\| W_{jk}(t, \mathcal{F}_0)- W_{jk}(t, \mathcal{F}_{0,\mathsf{s}})\|_q, \\ 
 \delta_{u,j}(\mathsf{s},q)=\sup_{1 \leq k \leq d_j} \sup_t\| U_{jk}(t, \mathcal{F}_0)- U_{jk}(t, \mathcal{F}_{0,\mathsf{s}})\|_q.
\end{gathered}
\end{equation} 
Then we have that for $\varsigma_j$ defined in (\ref{eq_xbasisbound}) and some constant $C>0$
\begin{equation}\label{eq_transferbound}
\max\{\delta_{w,j}(\mathsf{s},q),\delta_{u,j}(\mathsf{s},q) \} \leq C \varsigma_j \mathsf{s}^{-\tau}, \ \ \text{for all} \ \mathsf{s} \geq 1. 
\end{equation}
\end{lemma}

\begin{proof}[\bf Proof of Lemma \ref{lem_locallystationaryform}]
Due to similarity, we only prove the results for $\{\bm{u}_i\}$.  First of all, under Assumption \ref{assum_models}, for $w_{ik}$ defined in (\ref{eq_desigmatrixform}), we have that 
\begin{equation}\label{eq_meanzero}
\mathbb{E}(\varphi_{\el_2(k)}(X_{\el_1(k)+1,i} \epsilon_i)=\mathbb{E}(\mathbb{E}(\varphi_{\el_2(k)}(X_{\el_1(k)+1,i}) \epsilon_i| X_{\el_1(k)+1,i}))=0. 
\end{equation}
This shows that $\{\bm{u}_i\}$ is mean zero. Second, for any $1 \leq k \leq rd,$ by (\ref{eq_ddd}) and (\ref{eq_desigmatrixform}), we have that 
\begin{equation}\label{eq_exactform}
u_{ik}=\varphi_{\el_2(k)}(X_{\el_1(k)+1,i})\epsilon_i.
\end{equation}
Under the assumption of (\ref{eq_setting2}), it is clear that for some measurable function $U_k(t,\cdot),$ we can write
\begin{equation*}
u_{ik}=U_k\left(\frac{i}{n}, \mathcal{F}_i\right).
\end{equation*} 
Third, for $s,t \in [0,1],$ using (\ref{eq_exactform}), the smoothness of the basis functions and the assumption (\ref{eq_slc}), we have that for some constant $C>0$
\begin{equation*}
\sup_i \| U_k(s, \mathcal{F}_i)-U_k(t, \mathcal{F}_i) \|_q \leq C|t-s|.
\end{equation*}
Finally, for any $1 \leq k \leq rd,$ using (\ref{eq_exactform}) and (\ref{eq_setting2}), we have that for some constant $C>0$
\begin{align*}
\sup_t \| U_k(t, \mathcal{F}_0)&-U_{k}(t, \mathcal{F}_{0,\mathsf{s}}) \|_q  \leq  \sup_t \| \varphi_{\ell_2(k)}(G_{\ell_1(k)+1}(t,\mathcal{F}_0)) D(t, \mathcal{F}_0)-\varphi_{\ell_2(k)}(G_{\ell_1(k)+1}(t,\mathcal{F}_0)) D(t, \mathcal{F}_{0,\mathsf{s}}) \|_q \\
& + \sup_t \|\varphi_{\ell_2(k)}(G_{\ell_1(k)+1}(t,\mathcal{F}_0)) D(t, \mathcal{F}_{0,\mathsf{s}})-\varphi_{\ell_2(k)}(G_{\ell_1(k)+1}(t,\mathcal{F}_{0,\mathsf{s}})) D(t, \mathcal{F}_{0,\mathsf{s}}) \|_q  \\
& \leq C \varsigma \mathsf{s}^{-\tau},
\end{align*}
where we recall (\ref{eq_xbasisbound}) and in the last inequality we used Assumption \ref{assum_physical}, the smoothness of the basis function and mean value theorem.  $\{\bm{w}_i\}$ can be proved similarly and we omit the details. This finishes our proof. 
\end{proof}

\section{Some auxiliary lemmas}\label{sec_auxililarylemma}
In this section, we collect some useful auxiliary lemmas. The first lemma, Lemma \ref{lem_mvnapp} will be used in the proof of Gaussian approximation, i.e., Theorem \ref{thm_gaussianapproximationcase}. 
\begin{lemma}[Multivariate Gaussian approximation]\label{lem_mvnapp}
Denote $[n]:=\{1,\cdots, n\}$ and $X_N:=\sum_{i \in N} X_i$ for some index set $N.$ Let $X_i \in \mathbb{R}^d$ and $W=\sum_{i=1}^n X_i.$ Moreover, we assume that 
\begin{equation*}
\mathbb{E} X_i=0, \ \operatorname{Cov}(W)=\Sigma.
\end{equation*}
For a standard $d$-dimensional Gaussian random vector $Z$, denote 
\begin{equation*}
\mathtt{d}_c(\mathcal{L}(W), \mathcal{L}(\Sigma^{1/2}Z)):=\sup_{A \in \mathcal{A}}\left| \mathbb{P}(W \in A)-\mathbb{P}(\Sigma^{1/2}Z \in A) \right|,
\end{equation*}
where $\mathcal{A}$ denotes the collection of all the convex sets in $\mathbb{R}^d.$  Suppose that $W$ can be decomposed as follows
\begin{itemize}
\item $\forall i \in [n], \exists i \in N_i \subset [n] $ such that $W-X_{N_i}$ is independent of $X_i;$ 
\item $\forall i \in [n], j \in N_i, \exists N_i \subset N_{ij} \subset [n]$ such that   $W-X_{N_{ij}}$ is independent of $\{X_i, X_j\};$
\item $\forall i \in [n], j \in  N_i, k \in N_{ij}, \exists N_{ij} \subset N_{ijk} \subset [n]$ such that $W-X_{N_{ijk}}$ is independent of $\{X_i,X_j, X_k\}.$
\end{itemize}
Suppose further that for each $i \in [n], j \in N_i$ and $k \in N_{ij},$
\begin{equation*}
\|X_i\|_2 \leq \beta, \ |N_i| \leq n_1, \ |N_{ij}| \leq n_2, \ |N_{ijk}| \leq n_3. 
\end{equation*}
Then there exists a universal constant $C>0$ such that 
\begin{equation*}
\mathtt{d}_c(\mathcal{L}(W), \mathcal{L}(\Sigma^{1/2}Z)) \leq C d^{1/4} n \| \Sigma^{-1/2} \|^3  \beta^3 n_1 \left( n_2+\frac{n_3}{d} \right). 
\end{equation*}
\end{lemma}
\begin{proof}
See Theorem 2.1 and Remark 2.2 of \cite{MR3571252}.
\end{proof}

In the following two lemmas, we consider that $x_i=G_i(\mathcal{F}_i),$ where $G_i(\cdot)$ is some measurable function and $\mathcal{F}_i=(\cdots, \eta_{i-1}, \eta_i)$ and $\eta_i$ are i.i.d. random variables. Suppose that $\mathbb{E} x_i=0$ and $\max_i \mathbb{E} |x_i|^q<\infty$ for some $q>1.$ For any integer $k>0,$ denote $\theta_{k,q}=\max_{1 \leq i \leq n}\| G_i(\mathcal{F}_i)-G_i(\mathcal{F}_{i,i-k}) \|_q,$ where $\mathcal{F}_{i,i-k}=(\cdots, \eta_{i-k-1}, \eta'_{i-k}, \eta_{i-k+1}, \eta_i)$ and $\{\eta_i'\}$ are i.i.d. copies of $\{\eta_i\}.$ Let $\Theta_{k,q}=\sum_{i=k}^{\infty} \theta_{i,q}$ and $\Lambda_{k,q}=\sum_{i=0}^k \theta_{i,q},$ where we use the convention that $\theta_{i,q}=0, i<0.$ 

\begin{lemma}[Concentration inequalities for locally stationary time series]\label{lem_concentration} Let $S_n=\sum_{i=1}^n x_i$ and $q'=\min(2,q).$ Then: \\
(1). We have that for some constant $C>0$ 
\begin{equation*}
\|S_n\|_q^{q'}  \leq C \sum_{i=-n}^{\infty} (\Lambda_{i+n}-\Lambda_{i,q})^{q'}. 
\end{equation*}
Moreover, we have that
\begin{equation*}
\| \max_{1 \leq i \leq n} |S_i| \|_q \leq C n^{1/{q'}} \Theta_{0,q}. 
\end{equation*}
(2). Suppose $x_i=G(i/n, \mathcal{F}_i)$ and $\Theta_{k,q}=\OO(k^{-\gamma}), \gamma>0,$ then we have that
\begin{equation*}
\sum_{i=1}^n \left|\mathbb{E} x_i^2-\sigma\left(\frac{i}{n}\right)\right|=\OO\left(n^{1-\frac{\gamma}{2+\gamma}}\right), 
\end{equation*}
where $\sigma(\cdot)$ is the long run covariance matrix defined as 
\begin{equation*}
\sigma(t)=\sum_{k=-\infty}^{\infty} \operatorname{Cov}\left( G(t, \mathcal{F}_0), G(t, \mathcal{F}_k) \right). 
\end{equation*}
(3). For some constant $C>0,$ we have that 
\begin{equation*}
|\operatorname{Cov}(x_i, x_j)| \leq  C \theta_{|i-j|,q}.
\end{equation*}
\end{lemma}
\begin{proof}
The first part can be found in Theorem 1 of \cite{WUAOP}; also see Lemma 6 of \cite{MR3174655}; the second part of the proof follows directly from Corollary 2 of \cite{MR2827528}; see the last equation in the end of the proof of Corollary 2 therein; the third part can be found in Lemma 6 of \cite{MR3161462} or Lemma 2.6 of \cite{DZ}. 
\end{proof}

\begin{lemma}[$m$-dependent approximation]\label{lem_mdependent} Let $S_n'=\sum_{i=1}^n x_i'$, where $x_i'=\mathbb{E}(x_i| \mathcal{F}_\fm(i)), \mathcal{F}_\fm(i)=\sigma(\eta_{i-\fm}, \cdots, \eta_i)$ is the sigma-algebra generated by  $(\eta_{i-\fm}, \cdots, \eta_i).$ Here $\fm \geq 0.$ Define $R_n=S_n-S_n'$ and $R_n^*=\max_{1 \leq i \leq n} |R_i|.$  Then for some constant $C>0,$ we have
\begin{equation*}
\| R_n \|_q^{q'} \leq C n \Theta_{\fm,q}^{q'}, 
\end{equation*}
and
\begin{equation*}
\| R_n^* \|_q^{q'} \leq C 
\begin{cases}
n \Theta_{\fm,q}^2, & q>2;\\
n (\log n)^q   \Theta_{\fm,q}^2, & 1 <q \leq 2. 
\end{cases}
\end{equation*}
\end{lemma}
\begin{proof}
See Lemma A.1 \cite{MR2485027}. 
\end{proof}
%

The next lemma provides a deterministic control for the Riemann summation. 
\begin{lemma}\label{lem_intergralappoximation}
Suppose that $f$ is twice differentiable and $f^{''}$ is bounded and almost everywhere continuous on a compact set $[a,b].$ Let $\Delta: a=s_0 \leq s_1 \leq \cdots \leq s_{n-1} \leq s_n=b$ and $s_{i-1} \leq \xi_i \leq s_i.$ Denote the Riemann sum as
\begin{equation*}
\mathcal{R} \equiv \mathcal{R}(f; \Delta, \xi_i):=\sum_{i=1}^n (s_i-s_{i-1}) f(\xi_i).
\end{equation*}
Then we have that for $s_i=a+i(b-a)/n$
\begin{equation*}
\int_a^b f(x) \dd x =\mathcal{R}+\OO(n^{-2}). 
\end{equation*}
\end{lemma}
\begin{proof}
See Theorem 1.1 of \cite{TASAKI2009477}. 
\end{proof}

Then we collect some elementary probability inequalities. 
\begin{lemma}\label{lem_collectionprobineq}
(1). If $\mathbb{E}|X|^k<\infty,$ then for $0<j<k,$ $\mathbb{E}|X|^j<\infty$ and
\begin{equation*}
\mathbb{E}|X|^j \leq (\mathbb{E}|X|^k)^{j/k}. 
\end{equation*}
(2). (Holder's inequality) For $p,q \in [1, \infty)$ with $1/p+1/q=1,$
\begin{equation*}
\mathbb{E}|XY| \leq \| X\|_p \| Y\|_q. 
\end{equation*}
(3). (Chebyshev’s inequality) If $g(x)$ is a  monotonically increasing nonnegative function for the nonnegative reals, then we have that 
\begin{equation*}
\mathbb{P}(|X|>a) \leq \frac{E(g(|X|))}{g(a)}. 
\end{equation*}
(4). (Bernstein’s concentration inequality) Let $\{x_i\}$ be i.i.d. standard Gaussian random variables. Then for every $0<t<1,$ we have that
\begin{equation*}
\mathbb{P}\left(\left|\frac{1}{n} \sum_{k=1}^n x_k^2-1 \right| \geq t \right) \leq 2 e^{-nt^2/8}. 
\end{equation*}
{(5). (Probability control for Gaussian random vector on convex set) Let $A$ be a convex set on $\mathbb{R}^p$ and $\mathbf{z} \in \mathbb{R}^p$ be a Gaussian random vector. For some $\epsilon>0,$ denote the set
\begin{equation*}
A^\epsilon:=\{x \in \mathbb{R}^p: \operatorname{dist}(x, A) \leq \epsilon\},
\end{equation*}
where $\operatorname{dist}(x, A) :=\inf_{v \in A}|x-v|.$ Let $\mathcal{A}$ be the collection of all the convex sets on $\mathbb{R}^p.$ We have that 
\begin{equation*}
\sup_{A \in \mathcal{A}} \mathbb{P}(\mathbf{z} \in A^\epsilon \backslash A) \leq 4 p^{1/4} \epsilon.
\end{equation*} 
}
\end{lemma}
\begin{proof}
(1) follows from Exercise 1.6.11 of \cite{probbook}; (2) is Theorem 1.6.3 of \cite{probbook}; (3) follows from Theorem 1.6.4 of \cite{probbook}; (4) is in Example 2.11 of \cite{MR3967104}; (5) in in equation (4.10) of \cite{FangRollin2015}.   
\end{proof}

The following lemma provides a deterministic control  for the norm of a symmetric matrix. 
\begin{lemma}[Gershgorin circle theorem]\label{lem_circle}
 Let  $A=(a_{ij})$ be a complex $ n\times n$ matrix. For  $1 \leq i \leq n,$ let  $R_{i}=\sum _{{j\neq {i}}}\left|a_{{ij}}\right| $ be the sum of the absolute values of the non-diagonal entries in the  $i$-th row. Let  $ D(a_{ii},R_{i})\subseteq \mathbb {C} $ be a closed disc centered at $a_{ii}$ with radius  $R_{i}$. Such a disc is called a \emph{Gershgorin disc.} The every eigenvalue of $ A=(a_{ij})$ lies within at least one of the Gershgorin discs  $D(a_{ii},R_{i})$, where $R_i=\sum_{j\ne i}|a_{ij}|$.
\end{lemma}
\begin{proof}
See Theorem 7.2.1 of \cite{MR3024913}. 
\end{proof}
Finally, we collect some approximation formulas for the tail probabilities of the maxima of a two-dimensional Gaussian random field in the setting of simutaneous confidence bands. Assuming that 
$Y_i=f(x_i)+\epsilon_i, \ 1 \leq i \leq n, \ x_i \in \mathbb{R}^2$ and $\epsilon_i$ are i.i.d. $\mathcal{N}(0, \sigma^2)$ random variables. Let $\widehat{f}(x)$ be an unbiased estimator of $f(x)$ such that
\begin{equation*}
\widehat{f}(x)=l(x)^\top Y, \ l(x)=(l_1(x), \cdots, l_n(x))^\top, \ Y=(Y_1, \cdots, Y_n)^\top. 
\end{equation*}  
For some constant $\mathsf{c}_{\alpha}>0,$ the simultaneous coverage probability of the confidence bands is 
\begin{equation*}
1-\alpha=\mathbb{P}\left(\left|f(x)-\widehat{f}(x) \right| \leq \mathsf{c}_{\alpha} \sqrt{\operatorname{Var}(\widehat{f}(x))}, \ x \in \mathcal{X} \right). 
\end{equation*} 
Moreover, according to equation (1.4) of \cite{MR1311978}, we have that
\begin{equation}\label{eq_alphatheortrep}
\alpha=\p \left( \sup_{x \in \mathcal{X}} \left| l(x)^\top \epsilon /\sqrt{\operatorname{Var}(\widehat{f}(x))} \right| \geq \mathsf{c}_{\alpha} \sigma \right). 
\end{equation}
\begin{lemma}\label{lem_volumeoftube} Suppose $\mathcal{X}$ is a rectangle in $\mathbb{R}^2.$ Assume the manifold $\mathcal{M}:=\left\{l(x)/\sqrt{\operatorname{Var}(\widehat{f}(x))}: x \in \mathcal{X} \right\}$ is $C^3$ with a positive radius. Let $\kappa_0$ be the area of $\mathcal{M}$ and $\zeta_0$ be the length of the boundary of $\mathcal{M},$ then  (\ref{eq_expansionformula}) holds true. 
\end{lemma}
\begin{proof}
The proof follows from Proposition 2 of \cite{MR1311978} by replacing $\|l(x) \|$ with $\sqrt{\operatorname{Var}(\widehat{f}(x))}.$ 
\end{proof}


%
%
%
%
%

\bibliographystyle{imsart-number}
\bibliography{lag}

\end{document}